\documentclass[a4paper,11pt,final]{article}
\usepackage{fourier}
\usepackage{amsfonts,amsmath,amsthm,amssymb,xcolor,colortbl}
\usepackage[final]{microtype} %Avoids text overflow
\usepackage{mathtools}
\usepackage{enumitem}
\usepackage{calrsfs}
\DeclareMathAlphabet{\pazocal}{OMS}{zplm}{m}{n}
\usepackage[margin=1in]{geometry}
\usepackage[onehalfspacing]{setspace}
\usepackage{soul}
\usepackage{fnpct}
\usepackage{natbib}
\usepackage{babel}
\usepackage{clipboard}	%Allows chunks of text in this document to be used in another;
\usepackage{graphicx}
\newclipboard{PaperClipboard}	%Where the chunks of text are saved.
\usepackage{xr-hyper} %Must load before hyperref and cleveref
\usepackage[colorlinks,linkcolor=links,citecolor=cites,urlcolor=MyDarkBlue]{hyperref}
\usepackage{cleveref}
\usepackage{tikz}
\usepackage{bbm}
\usepackage[justification=centering]{caption}
\usepackage[justification=centering]{subfig}
\usepackage[section]{placeins}
\makeatletter
\def\thm@space@setup{%
  \thm@preskip=\parskip \thm@postskip=0pt
}
\makeatother
\usepackage{parskip}
\usepackage{nicefrac}
\usepackage{appendix}
\usepackage[euler]{textgreek}
\definecolor{MyDarkBlue}{rgb}{0,0.08,0.45}
\definecolor{cites}{HTML}{324b13}
\definecolor{links}{HTML}{1a663b}
\definecolor{MyLightMagenta}{cmyk}{0.1,0.8,0,0.1}
\definecolor{scyan}{HTML}{CBEAFC}
\definecolor{red}{HTML}{B5595C}
\definecolor{green}{HTML}{609B57}
\definecolor{spink}{HTML}{FFB0FF}
\definecolor{yellow}{HTML}{E5A919}
\definecolor{gray}{HTML}{c8c8c8}
\definecolor{darkgray}{HTML}{646464}
\definecolor{blue}{RGB}{0,114,178}
\definecolor{violet}{RGB}{136, 93, 144}
\usetikzlibrary{positioning}
\usetikzlibrary{snakes}
\usetikzlibrary{calc}
\usetikzlibrary{arrows}
\usetikzlibrary{decorations.markings}
\usetikzlibrary{shapes.misc}
\usetikzlibrary{matrix,shapes,fit,tikzmark}
\usetikzlibrary{patterns}
\usetikzlibrary{trees,calc,positioning,snakes}
\usetikzlibrary{decorations.pathmorphing}
\usetikzlibrary{decorations.shapes}
\usepackage{pgfplots}
\pgfplotsset{width=7cm,compat=1.16}
\pgfplotsset{ytick style={draw=none}}
\pgfplotsset{ztick style={draw=none}}
\pgfplotsset{xtick style={draw=none}}

\usepgfplotslibrary{ternary} 
\usepgfplotslibrary{fillbetween}

\usepackage{mparhack}

 \usepackage{chngcntr}
\usepackage{apptools}
\AtAppendix{\counterwithin{equation}{section}}
\AtAppendix{\counterwithin{figure}{section}}
\AtAppendix{\counterwithin{table}{section}}
\AtAppendix{\counterwithin{example}{section}}
\AtAppendix{\counterwithin{lemma}{section}}
\AtAppendix{\counterwithin{remark}{section}}
\AtAppendix{\counterwithin{proposition}{section}}
\AtAppendix{\counterwithin{theorem}{section}}
\AtAppendix{\counterwithin{corollary}{section}}
\usepackage{etoolbox}
\makeatletter
\patchcmd{\hyper@makecurrent}{%
    \ifx\Hy@param\Hy@chapterstring
        \let\Hy@param\Hy@chapapp
    \fi
}{%
    \iftoggle{inappendix}{%true-branch
        \@checkappendixparam{chapter}%
        \@checkappendixparam{section}%
        \@checkappendixparam{subsection}%
        \@checkappendixparam{subsubsection}%
        \@checkappendixparam{paragraph}%
        \@checkappendixparam{subparagraph}%
    }{}%
}{}{\errmessage{failed to patch}}

\newcommand*{\@checkappendixparam}[1]{%
    \def\@checkappendixparamtmp{#1}%
    \ifx\Hy@param\@checkappendixparamtmp
        \let\Hy@param\Hy@appendixstring
    \fi
}
\makeatletter

\newtoggle{inappendix}
\togglefalse{inappendix}

\apptocmd{\appendix}{\toggletrue{inappendix}}{}{\errmessage{failed to patch}}
\apptocmd{\subappendices}{\toggletrue{inappendix}}{}{\errmessage{failed to patch}}

\usepackage[scr=dutchcal]{mathalpha}

\newtheorem{theorem}{Theorem}
\newtheorem{lemma}{Lemma} 
\newtheorem{proposition}{Proposition}
 
\newtheorem{example}{Example}
\newtheorem{definition}{Definition} 
\newtheorem{corollary}{Corollary}

\newtheorem{remark}{Remark}
\newtheorem{observation}{Observation}

\newcommand{\pay}{\ensuremath{u}}
\newcommand{\payb}{\ensuremath{\pay^\prime}}
\newcommand{\payi}{\ensuremath{\pay_i}}
\newcommand{\alloc}{\ensuremath{a}}
\newcommand{\alloci}{\ensuremath{\alloc_i}}
\newcommand{\allocmi}{\ensuremath{\alloc_{-i}}}
\newcommand{\Alloc}{\ensuremath{A}}

\newcommand{\Alloci}{\ensuremath{\Alloc_i}}
\newcommand{\Allocmi}{\ensuremath{\Alloc_{-i}}}
\newcommand{\type}{\ensuremath{\theta}}
\newcommand{\mint}{\ensuremath{\underline{\type}}}
\newcommand{\maxt}{\ensuremath{\overline{\type}}}
\newcommand{\typeb}{\ensuremath{\type^\prime}}
\newcommand{\typec}{\ensuremath{\tilde{\type}}}
\newcommand{\typeh}{\ensuremath{\hat\type}}
\newcommand{\typei}{\ensuremath{\type_i}}
\newcommand{\typebi}{\ensuremath{\typei^\prime}}

\newcommand{\typemi}{\ensuremath{\type_{-i}}}

\newcommand{\typecmi}{\ensuremath{\typec_{-i}}}
\newcommand{\Types}{\ensuremath{\Theta}}
\newcommand{\Typesi}{\ensuremath{\Types_i}}
\newcommand{\Typesmi}{\ensuremath{\Types_{-i}}}
\newcommand{\state}{\ensuremath{\omega}}
\newcommand{\stateb}{\ensuremath{\state^\prime}}
\newcommand{\statec}{\ensuremath{\tilde\state}}
\newcommand{\States}{\ensuremath{\Omega}}
\newcommand{\Beliefs}{\ensuremath{\Delta(\States)}}
\newcommand{\Posteriors}{\ensuremath{\Delta(\Types)}}
\newcommand{\cor}{\ensuremath{\varepsilon}}
\newcommand{\fcor}{\ensuremath{\pazocal{E}}}

\newcommand{\cori}{\ensuremath{\cor_i}}
\newcommand{\estate}{\ensuremath{\statec}}
\newcommand{\testate}{\ensuremath{\estate^\star}}
\newcommand{\estateb}{\ensuremath{\estate^\prime}}
\newcommand{\estates}{\ensuremath{\tilde\States}}
\newcommand{\ebeliefs}{\ensuremath{\Delta(\estates)}}

\newcommand{\aarg}{\ensuremath{(\alloc,\type,\state)}}
\newcommand{\argq}{\ensuremath{(q,\type,\state)}}

\newcommand{\argb}{\ensuremath{(\alloc,\type,\state,\belief)}}
\newcommand{\argme}{\ensuremath{(\alloc,\type,\mssg,\estate)}}
\newcommand{\argmeb}{\ensuremath{(\alloc,\type,\mssg,\estate,\belief)}}
\newcommand{\argi}{\ensuremath{(\alloci,\typei,\state)}}
\newcommand{\Arg}{\ensuremath{\Alloc\times\Types\times\States}}
\newcommand{\Argi}{\ensuremath{\Alloci\times\Typesi\times\States}}
\newcommand{\Argb}{\ensuremath{\Alloc\times\Types\times\States\times\Beliefs}}
\newcommand{\Argbe}{\ensuremath{\Alloc\times\Types\times\mssgs\times\estates\times\ebeliefs}}

\newcommand{\Argme}{\ensuremath{\Alloc\times\Types\times\mssgs\times\estates}}

\newcommand{\bsplit}{\ensuremath{\tau}}
\newcommand{\bsplitb}{\ensuremath{\bsplit^\prime}}
\newcommand{\belief}{\ensuremath{\mu}}

\newcommand{\beliefh}{\ensuremath{\hat\belief}}

\newcommand{\beliefi}{\ensuremath{\belief_i}}

\newcommand{\typed}{\ensuremath{f}}
\newcommand{\typecdf}{\ensuremath{F}}
\newcommand{\typecdfi}{\ensuremath{F_i}}
\newcommand{\typecdfmi}{\ensuremath{F_{-i}}}
\newcommand{\typedi}{\ensuremath{\typed_i}}
\newcommand{\typedmi}{\ensuremath{\typed_{-i}}}
\newcommand{\prior}{\ensuremath{\belief_0}}
\newcommand{\cord}{\ensuremath{\eta}}

\newcommand{\excdf}{\ensuremath{\pi}}
\newcommand{\cexp}{\ensuremath{\excdf_\mechanism}}
\newcommand{\cexpi}{\ensuremath{\excdf_{\mechanism,i}}}

\newcommand{\signal}{\ensuremath{s}}
\newcommand{\signali}{\ensuremath{\signal_i}}
\newcommand{\signals}{\ensuremath{S}}
\newcommand{\asignal}{\ensuremath{\hat{\signal}}}
\newcommand{\asignals}{\ensuremath{\hat{\signals}}}
\newcommand{\csignal}{\ensuremath{\signal^*}}
\newcommand{\csignali}{\ensuremath{\signali^*}}
\newcommand{\csignals}{\ensuremath{\signals^*}}
\newcommand{\csignalsi}{\ensuremath{\csignals_i}}

\newcommand{\mechanism}{\ensuremath{\phi}}
\newcommand{\mmech}{\ensuremath{\mechanism_{\text{\texttt{My}}}}}
\newcommand{\myerson}{\ensuremath{\text{\texttt{My}}}}
\newcommand{\cali}{\ensuremath{\text{\texttt{cal}}}}
\newcommand{\full}{\ensuremath{\text{\texttt{full}}}}
\newcommand{\tsmech}{\ensuremath{\psi}}
\newcommand{\rmech}{\ensuremath{\phi}}
\newcommand{\mssg}{\ensuremath{m}}

\newcommand{\mssgs}{\ensuremath{M}}

\newcommand{\allocr}{\ensuremath{\alpha}}
\newcommand{\allocrb}{\ensuremath{\allocr^\prime}}
\newcommand{\beliefr}{\ensuremath{\beta}}
\newcommand{\beliefrb}{\ensuremath{\beliefr^\prime}}
\newcommand{\dmech}{\ensuremath{\varphi}}

\newcommand{\dmechc}{\ensuremath{\tilde\dmech}}
\newcommand{\dmecht}{\ensuremath{\dmech_t}}
\newcommand{\virtual}{\ensuremath{J}}

\newcommand{\ooi}{\ensuremath{\alloc_{i\emptyset}}}
\newcommand{\oo}{\ensuremath{\alloc_\emptyset}}
\newcommand{\mechanisms}{\ensuremath{\pazocal{M}}}
\newcommand{\cmechanisms}{\ensuremath{\mechanisms_{\mathrm{cal}}}}
\newcommand{\val}{\ensuremath{v}}
\newcommand{\valb}{\ensuremath{\val^\prime}}
\newcommand{\vs}{\ensuremath{J}}
\newcommand{\DP}{\ensuremath{W}}
\newcommand{\ddp}{\ensuremath{w}}
\newcommand{\Qup}{\ensuremath{Q_\uparrow}}
\newcommand{\Qdown}{\ensuremath{Q_\downarrow}}
\newcommand{\Qmy}{\ensuremath{Q_\myerson}}
\newcommand{\Qcal}{\ensuremath{Q_\cali}}
\newcommand{\maxq}{\ensuremath{\bar q}}
\newcommand{\lotq}{\ensuremath{\xi}}

\newcommand{\strat}{\ensuremath{\sigma}}

\newcommand{\stratb}{\ensuremath{\strat^\prime}}

\newcommand{\terminals}{\ensuremath{\pazocal{H}}}
\newcommand{\terminalsh}{\ensuremath{\terminals^\infty}}
\newcommand{\weakc}{\ensuremath{\stackrel{w^*}{\longrightarrow}}}
\newcommand{\occup}{\ensuremath{\nu}}

\newcommand{\eoccup}{\ensuremath{\overline{\occup}}}

\newcommand{\rep}{\ensuremath{\rho}}
\newcommand{\freq}{\ensuremath{\mathrm{freq}}}
\newcommand{\nlast}{\ensuremath{n^{\mathrm{last}}}}
\newcommand{\Tlast}{\ensuremath{ T^{\mathrm{last}}}}
\newcommand{\flast}{\ensuremath{ f^{\mathrm{last}}}}

\newcommand{\pp}{\ensuremath{p}}

\newcommand{\Ntotal}{\ensuremath{\tilde N^{\mathrm{total}}}}
\newcommand{\Ltotal}{\ensuremath{\tilde L^{\mathrm{total}}}}

\newcommand{\Utruth}{\ensuremath{U^{\mathrm{truth}}}}
\newcommand{\outcome}{\ensuremath{\vartheta}}
\newcommand{\outcomeb}{\ensuremath{\outcome^\prime}}

\newcommand{\lebesgue}{\ensuremath{\lambda}}
\newcommand{\reals}{\ensuremath{\mathbb{R}}}
\newcommand{\naturals}{\ensuremath{\mathbb{N}}}
\newcommand{\borel}{\ensuremath{\pazocal{B}}}

\newcommand{\testf}{\ensuremath{g}}
\newcommand{\testfb}{\ensuremath{q}}

\newcommand{\cav}{\ensuremath{\mathrm{cav}\;}}
\newcommand{\expect}{\ensuremath{\mathbb{E}}}

\newcommand{\pr}{\ensuremath{\mathbb{P}}}

\newcommand{\setplayers}{\ensuremath{[N]}}
\newcommand{\estrat}{\ensuremath{(\pp,\strat)}}
\newcommand{\estratb}{\ensuremath{(\pp^\prime,\stratb)}}
\newcommand{\estratopt}{\ensuremath{(\pp^*,\strat^*)}}

\newcommand{\estratc}{\ensuremath{(\tilde\pp,\tilde\strat)}}

\title{Calibrated Mechanism Design\thanks{The latest version of the paper can be found \href{https://www.dropbox.com/scl/fi/dbucfzy2mhu954vo1spmn/calibrated-mechanism-design.pdf?rlkey=4atd3hai9v6o32meyv36uchvs&dl=0}{here}. We thank Dirk Bergemann, James Best, Elliot Lipnowski, Emir Kamenica, Navin Kartik, Stephen Morris , Phil Reny, Marzena Rostek, and \href{https://www.refine.ink}{Refine.ink} for valuable comments and suggestions, as well as seminar audiences at EARIE 2024, WUSTL Economic Theory Conference 2024, MIT IDSS Distinguished Speaker Series, the Workshop in Market Design, Bonn, Berlin HU, CERGE-EI, Gerzensee, UCL, Naples, Stanford University, MIT Sloan, and the CEPR Workshop on Contracts, Incentives and Information at Collegio Carlo Alberto. Laura Doval gratefully acknowledges financial support from the Sloan Foundation. Alex Smolin gratefully acknowledges
funding from the French National Research Agency (ANR) under the Investments for the Future program
(grant ANR-17-EURE-0010) and the AI Interdisciplinary Institute ANITI (grant ANR-23-IACL-0002). This paper was partly written while the first author was visiting Stanford University and the second author was visiting Northwestern University and Columbia Business School; we are grateful for their hospitality.}}
\author{Laura Doval\thanks{Columbia Business School and CEPR. E-mail: \href{mailto:laura.doval@columbia.edu}{\texttt{laura.doval@columbia.edu}.}} \and Alex Smolin\thanks{Toulouse School of Economics and CEPR. E-mail: \href{mailto:alexey.v.smolin@gmail.com}{\texttt{alexey.v.smolin@gmail.com}.}}}
\begin{document}
\pagenumbering{gobble}
\maketitle
\begin{abstract}

 We study mechanism design when a designer repeatedly uses a fixed mechanism to interact with strategic agents who learn from observing their allocations. We introduce a static framework, \emph{calibrated mechanism design}, requiring mechanisms to remain incentive compatible given the information they reveal about an underlying state through repeated use. In single-agent settings, we prove implementable outcomes correspond to two-stage mechanisms: the designer discloses information about the state, then commits to a state-independent allocation rule. This yields a tractable procedure to characterize calibrated mechanisms, combining information design and mechanism design. In private values environments, full transparency is optimal and correlation-based surplus extraction fails. We provide a microfoundation by showing calibrated mechanisms characterize exactly what is implementable when an infinitely patient agent repeatedly interacts with the same mechanism.  Dynamic mechanisms that condition on histories expand implementable outcomes only by weakening incentive constraints, but not by enriching the designer's ability to obfuscate learning.

\end{abstract}
\newpage
\clearpage
\pagenumbering{arabic}

\section{Introduction}\label{sec:intro}

Many economic institutions rely on mechanisms that remain fixed while agents interact with them repeatedly. Online platforms commit to stable auction formats for advertising slots, lenders use persistent scoring algorithms for loan decisions, and regulators establish durable rules for market participants.   When the mechanism's operation depends on information known only to the designer\textemdash such as the platform's data about match values, the lender's assessment of credit market conditions, or the regulator's understanding of market fundamentals\textemdash participants may infer this information by observing their outcomes across repeated interactions.  This learning creates a fundamental constraint: the information a mechanism reveals through repeated use limits what outcomes it can implement in the long run. Participants  can use the information gleaned from past interactions when deciding whether and how to participate, tightening the designer's incentive constraints.  A lender whose approval decisions depend on unobserved credit market conditions will gradually reveal   these conditions to borrowers through his lending decisions, constraining the lender's ability to provide credit efficiently. We study how this endogenous information leakage shapes the set of implementable outcomes in mechanism design.

A simple example illustrates how learning prevents the designer from exploiting his information. Consider a seller who repeatedly offers a good whose demand depends on an unobserved state, which can be either low ($L$) or high ($H$). Each state is equally likely. The seller faces a buyer whose value for the good can take one of two values, $1/2$ or $1$. The probability that the buyer's value is 1 is higher when the demand state is high. \autoref{table:cml-intro} summarizes the value distribution conditional on the demand state:
\begin{table}[ht!]
\centering
\begin{tabular}{c|ccc}
&  $\val=1/2$ & $\val=1$ \\ \hline
 $L$\  & $2/3$ & $1/3$ \\
 $H$\ & $1/3$ & $2/3$
\end{tabular}
\caption{Value distribution conditional on demand state.}\label{table:cml-intro}
\end{table}

Suppose the seller can design the terms of trade, that is, the probability with which he allocates the good to the buyer ($q\in[0,1]$) and the payment the buyer makes to the seller ($t\in\reals$). The buyer's payoff is $\val q-t$, and the seller's is $t$. The buyer can always choose to not trade with the seller and ensure a payoff of $0$. 

Suppose first the buyer and the seller interact only once. \autoref{table:cml-intro-mech} depicts an optimal mechanism for the seller in this case:
\begin{table}[ht!]
\centering
\begin{tabular}{c|cc}
&$\val=1/2$&$\val=1$\\
\hline
$L$&$(1,0)$&$(1,0)$\\
$H$&$(1,3/2)$&$(1,3/2)$
\end{tabular}
\caption{Trade probabilities and payments as a function of buyer's value and demand state.}\label{table:cml-intro-mech}
\end{table}

 In this mechanism, the buyer gets the good for free when the demand state is $L$ and pays a price of $3/2$ when it is $H$. If this mechanism were offered once without the buyer observing the demand state, the buyer obtains a payoff of $0$ from participating and truthfully reporting her type. Unsurprisingly, the seller extracts the buyer's surplus: the seller knows the demand state, which is correlated with the buyer's type, and exploits this information in the design of his mechanism (cf. \citealp{cremer1988full}).
 
 Suppose now the buyer interacts repeatedly with the mechanism, but the state remains fixed. If the buyer observes nothing from her interaction with the mechanism, the buyer is willing to participate and truthfully report her value into the mechanism, no matter how many times it is offered: In each period, she anticipates getting a (continuation) payoff of $0$ from engaging with the mechanism. Suppose, instead, the buyer observes her allocation in the mechanism.  If the demand state is $L$, the buyer gets the good for free at the end of the first period, and from now on knows this is what she will get in the mechanism. If the demand state is $H$, the buyer gets the good and pays a price of $3/2$ as she agreed to when she decided to participate in period 1, but anticipating a price of $3/2$ from then onwards, never again participates in the mechanism. Thus, whereas the seller can implement the outcomes in \autoref{table:cml-intro-mech} when the buyer does not observe her allocations, this is no longer the case when she can. 
 
This paper develops a framework for mechanism design in which agents' ability to learn about the designer's information from repeatedly playing a mechanism constrains implementable outcomes. In our framework, allocations depend on agents' reports and on a state known only to the designer. Through repeated participation, agents observe their allocations and gradually learn about this state. A mechanism therefore serves a dual role: it determines allocations based on reports, and it acts as an information structure that reveals the underlying state. The more the mechanism conditions on the state, the more information it leaks, and the tighter the constraints on implementable outcomes.
 
We approach our analysis in two steps. First, we introduce a static solution concept for mechanism design that directly models the feedback between the mechanism, the information it reveals, and participants' behavior. This solution concept allows us to tractably capture the limits on the set of implementable outcomes implied by agents' learning,  while abstracting from the dynamics of experimentation. Second, we provide a dynamic microfoundation showing this static solution concept precisely captures the implementable outcomes when an infinitely patient agent repeatedly interacts with the same mechanism.

In \autoref{sec:model}, we introduce a static solution concept\textemdash calibrated mechanism design\textemdash requiring that mechanisms remain incentive compatible and individually rational given the information they reveal about the state through their allocations.  We formalize this requirement through the notion of a \emph{calibrated mechanism}. We couple each mechanism with an information structure that describes what participants learn about the state from the mechanism. The information structure reveals to each agent an interim allocation rule\textemdash the mapping from her type reports to lotteries over her allocations\textemdash capturing what she would learn from repeatedly observing her outcomes in the mechanism. We require the information structure to be calibrated in the sense of \cite{foster1997calibrated}: the interim allocation rule each agent observes must accurately describe the allocation probabilities she faces. Throughout the paper, we study calibrated mechanism design: the designer chooses a mechanism that remains incentive compatible and individually rational when participants have access to the mechanism's calibrated information structure before playing. Calibration imposes a constraint on the designer relative to standard mechanism design: the more the allocation rule depends on the state, the more informative the calibrated information structure becomes, and hence the more incentive and participation constraints the designer must satisfy. 

In private values environments, the constraint that the mechanism must remain incentive compatible and individually rational given the information it reveals about the state pushes the designer to full transparency. We show in \autoref{theorem:pv} that, under the calibration constraint, the designer can do no better than inducing in each state the optimal direct mechanism when there is common knowledge of that state. In particular, in settings with transferable utility in which the designer has statistical information about the agents' types, \autoref{theorem:pv} implies the designer cannot extract full surplus.

 In \autoref{sec:main}, we characterize optimal calibrated mechanisms through a tractable class we dub \emph{two-stage} mechanisms. In a two-stage mechanism, the designer first discloses information about the state to the agent\textemdash inducing a belief about the state\textemdash then commits to an allocation rule that depends only on the agent's report, not  the state itself. \autoref{theorem:cmd-two-stage} shows that in single-agent settings, calibrated mechanisms and two-stage mechanisms implement exactly the same outcome distributions. This equivalence yields a practical algorithm for finding optimal calibrated mechanisms, combining tools from information design and mechanism design: for each possible belief the designer might induce, solve a standard mechanism design problem given that belief; then choose the optimal information disclosure by concavifying the resulting value function. 

 In the case of multiple agents, \autoref{theorem:cmd-two-stage-multi} shows calibrated mechanisms admit a similar representation via \emph{generalized} two-stage mechanisms: Like two-stage mechanisms, the designer individually discloses to each agent a belief about the state and offers an incentive compatible and individually rational interim allocation rule that no longer conditions on the state. Whereas the designer observes the disclosed belief profile, each agent only observes the belief disclosed to her.\footnote{In (generalized) two-stage mechanisms, the designer communicates with the agents \emph{before} the agents communicate with the mechanism. Whereas this communication is a restriction on the set of implementable outcomes relative to the single-designer Myersonian benchmark, \cite{attar2025keeping} show that allowing competing principals to first communicate with agents expands the set of implementable outcomes.} Moreover, each agent learns only her own interim allocation rule\textemdash how her reports map to her allocations\textemdash rather than the complete mapping from type profiles to allocations. This partial observability requires additional consistency conditions to ensure agents' interim allocation rules are mutually compatible. In contrast to the single-agent case, not every generalized two-stage mechanism induces a calibrated mechanism, as generalized two-stage mechanisms may reveal strictly less information than calibrated mechanisms. 
 %\ld{after we discuss Theorem 3 i have to add it here or to the previous paragraph.}

In \autoref{sec:app},  we study optimal calibrated mechanism design in the canonical setting of quasilinear utilities, single-dimensional types and allocations. In \autoref{sec:transparent}, we study the single-agent case. We show that if the order of types is state independent, then optimal two-stage mechanisms fully reveal the state, whereas this conclusion can be reversed when the order of types is state-dependent. In \autoref{sec:myerson}, we compare optimal calibrated mechanism design against the Myersonian benchmark. We provide sufficient conditions under which the designer realizes the payoff of the Myersonian benchmark under the calibration constraint; under these conditions, the optimal Myersonian mechanism satisfies the agent's incentive constraints state-by-state. Building on that result, we analyze multi-agent applications in \autoref{sec:optimal-auction}.

\autoref{sec:microf} provides a microfoundation for calibrated mechanism design. We analyze an infinite-horizon game where an infinitely patient agent repeatedly plays the same mechanism.\footnote{We assume the agent has limit-of-means preferences, so we can pass to the $\delta\rightarrow 1$ limit without approximation.}  The state remains fixed, but the agent's type is redrawn each period independently of the state.\footnote{\autoref{theorem:repeated} holds when the agent's type is drawn once at the beginning without further assumptions on the distribution. As we explain in \autoref{sec:microf}, we choose the i.i.d. specification for the evolution of the agent's private information to put repeated and dynamic mechanisms on a more equal footing.} Our notion of implementation is based on the long-run expected frequency of allocation-type-state tuples when the agent best responds to the mechanism. \autoref{theorem:repeated} shows that the implementable outcome distributions are precisely those induced by incentive compatible two-stage mechanisms. This result validates our static framework: calibrated mechanism design captures exactly what is implementable through repeated play.

We then ask whether giving the designer additional flexibility helps. In a dynamic mechanism, the designer can condition each period's allocation on the complete history of past reports and allocations, rather than using the same mechanism repeatedly. \autoref{theorem:dynamic} shows that dynamic mechanisms expand implementable outcomes in a specific way: they correspond to two-stage mechanisms with weaker incentive compatibility and individual rationality conditions. The designer can now exploit the ability to monitor the frequency of type reports over time, which allows him to punish detectable deviations\textemdash reporting strategies whose frequency distribution differs from the true type distribution. Instead, the mechanism must be robust to \emph{undetectable} ones. Importantly, in environments with transferable utility, this distinction vanishes: As shown in \cite{rahman2024detecting}, eliminating profitable undetectable deviations is equivalent to incentive compatibility, so dynamic mechanisms implement exactly the same distributions over physical allocations, types, and states as our static calibrated mechanisms. 

\paragraph{Related Literature} The paper lies at the intersection of four literatures: rational expectations equilibria, (public) information disclosure in mechanism design, the computer science literature on learning in repeated auctions, and dynamic implementation. 

The definition of a calibrated mechanism is in the spirit of rational expectations equilibria \citep{radner1979rational,green1977non,kreps1977note}. Indeed, requiring a mechanism to remain incentive compatible given the information it reveals about the state mirrors the rational-expectations requirement that prices clear markets given the information they convey. Unlike rational expectations equilibrium, where the only role of prices is to clear the market, calibrated mechanisms are chosen by a designer who understands the incentive implications of the mechanism's information leakage and trades this off against the value of conditioning the mechanism on the state. Similar to our analysis in \autoref{sec:microf}, some papers in the literature have studied the question of whether rational expectations equilibria emerge from learning dynamics (see, for instance, \citealp{milgrom1981rational,blume1982introduction}).

Following \cite{milgrom1982theory}, a literature has studied whether a designer should publicly disclose information he knows before a mechanism is played. \cite{ottaviani2001value} show revealing a signal affiliated with the buyer's value is optimal in a single-agent screening problem. When considering the case of an informed principal, they consider what we call two-stage mechanisms to bound the monopolist's profits. \cite{szabadi2018essays} and \cite{yamashita2018optimal} study the optimal release of public information followed by an optimal mechanism conditional on that disclosure, while \cite{fu2012ad} study this question in the context of a second price auction.  In those papers, the restriction to public disclosure and the independence of the mechanism on information other than the disclosed one is a  constraint on the class of mechanisms the designer can use. Instead, we show this class of mechanisms is without loss when the designer faces our calibration constraint in the single-agent case, but it may not be in the multi-agent case. Note, however, that when full or no disclosure are optimal in the Myersonian benchmark the distinction between private and public disclosure is immaterial. For that reason, the results on the achievability of the Myersonian benchmark are similar across their and our work. \cite{daskalakis2016does} lift the restriction to public disclosure and study the Myersonian benchmark in an auction setting, showing that the complexity of that problem is the same as that of a multi-product monopolist (cf. \citealp{guesnerie1984complete}).\footnote{There is also a literature that studies a designer's disclosure of information that must be elicited from the agents \citep{esHo2007optimal,bergemann2007information,li2017discriminatory,krahmer2020information,bergemann2022calibrated,bergemann2022screening,smolin2023disclosure}. By contrast, the designer knows the realization of the state and also what information the two-stage mechanism discloses to the agents, so he need not elicit this information.} 

Motivated by the prevalence of fixed auction formats with which bidders interact repeatedly, a literature in computer science studies the properties of bidder learning algorithms and the implications for the auctioneer (see, for instance, \citealp{golrezaei2019dynamic,nedelec2019learning,kanoria2020dynamic}, and \citealp{nedelec2022learning} for a survey treatment). A common finding is that learning bidders can take advantage of ``naive'' auction formats which are no longer incentive compatible when bidders learn. Inspired by this literature, we develop a framework which allows us to systematically study the question of optimal mechanism design in the presence of learning agents.

Our dynamic implementation results relate to the literature that studies whether a mechanism can be implemented either by linking decisions \citep{jackson2007overcoming,ball2023quota} or in the patient limit of a repeated interaction \citep{renou2015approximate,margaria2018dynamic,meng2021value}. Both strands identify cyclical monotonicity as the condition for implementation (cf. \citealp{rochet1987necessary}). \cite{rahman2024detecting} shows that cyclical monotonicity is equivalent to the absence of profitable undetectable deviations.

 By focusing on what agents learn from the designer's information, our paper is distinct from the literature on mechanism design with interdependent payoffs which focuses on agents' learning about others' types through their actions in the mechanism \citep{green1987posterior,niemeyer2022posterior,hafner2025mechanism}. Moreover, by focusing in the case of a designer with commitment, we are distinct from the literature on the informed principal \citep{myerson1983mechanism,maskin1990principal}.

Lastly, our paper contributes to two literatures. First, by studying the informational role of the mechanism, we contribute to the literature on feedback in auctions, which analyzes how different feedback rules affect bidders' information about other agents, and ultimately behavior in first price auctions (see, for instance, \citealp{esponda2008information,bergemann2018should,cesa2024role}). Second, by showing the designer's problem involves solving information and mechanism design problems, our paper joins a recent literature that highlights the dual role of the mechanism as an information structure and an allocation rule \citep{calzolari2006optimality,dworczak2020mechanism,doval2022mechanism}.
\section{Calibrated Mechanism Design}\label{sec:model}
In this section, we introduce the static setting and solution concept that captures the impact of agents' learning from the mechanism on the set of implementable outcomes. We defer to \autoref{sec:microf} the analysis of the dynamic game whose outcomes our static solution concept captures.

\paragraph{Primitives} A designer (he) interacts with $N$ privately informed agents (she) to determine an allocation. Let $\Typesi$ denote the set of types of agent $i$, and $\Types\equiv\times_{i=1}^N\Typesi$. Each agent knows her type, but not those of other agents. The allocation space is given by $\Alloc\equiv\times_{i=1}^N\Alloci$.\footnote{Assuming the allocation space is a product space is without loss of generality. Any restriction on the allocations, such as all agents must receive the same allocation, can be incorporated as restrictions on the support of the mechanism.} Finally, let $\States$ denote a set of states, which are known to the designer, but not to the agents. The sets \Typesi, \Alloci, and \States\ are assumed to be finite throughout. \footnote{Because we allow for lotteries over allocations, that the allocation space is finite does not preclude the case of transferable utility. Indeed, we could let each $\Alloci=\tilde\Alloci\times\{-K,K\}$ for some large enough $K>0$.} Agent $i$'s payoffs are given by $\payi:\Alloci\times\Typesi\times\States\to\reals$. That is, agent $i$ cares about her dimension of the allocation, her type, and the state, and not about other agents' allocations or types.

Denote by $\prior$ the distribution over \States. For each $\state\in\States$, let $\typed(\cdot|\state)\in\Posteriors$ denote the type distribution. We assume throughout the types are independently distributed conditional on the state, that is,
\begin{align}
\typed(\type|\state)=\prod_{i=1}^N\typedi(\typei|\state),
\end{align}
for all $\type\in\Types$ and $\state\in\States$. Together with the assumption on agents' payoffs, the assumption on $\typed(\cdot|\state)$ allows us to isolate the effect of learning about the state from that of learning about others' types (perhaps because others' types provide additional information about the state).

\paragraph{Mechanisms}
We model mechanisms as mappings
\begin{align}\label{eq:mech-def}
\mechanism:\Types\times\States\times[0,1]\to\Delta(\Alloc),
\end{align}
where $\cor\in[0,1]$ is a uniformly distributed random variable, which we refer to as the randomization device. 

Several comments are in order. First, to understand how a mechanism works, the timing of when the different random variables is drawn is important. In particular, we assume that both the state \state\ and the realization of the randomization device \cor\ are independently drawn at the beginning, but not observed by the agents. This determines the direct mechanism $\mechanism(\cdot,\state,\cor):\Types\to\Delta(\Alloc)$ to which the agents send type reports, which in turn determines the lottery from which the allocation is drawn. Thus, the allocation is random in our setting for two reasons: on the one hand, the agents do not know the realization of $(\state,\cor)$, and hence the direct mechanism $\mechanism(\cdot,\state,\cor)$ they face. Second, conditional on $(\state,\cor)$, the allocation may be drawn at random. Mathematically, we could have subsumed \emph{all} sources of randomness in the allocation into the randomization device. However, as we explain next, the definition in \autoref{eq:mech-def} allows us to distinguish the source of randomness in the allocation that is informative about the state from that which is not.

Second, it is useful to consider the reason for the randomization device in the definition of a mechanism. For simplicity, consider the case of the designer facing a single agent. If the agent had repeated access to the mechanism, the agent would stand to learn the mapping $\mechanism(\cdot,\state,\cor):\Types\to\Delta(\Alloc)$ by experimenting with different reports into the mechanism and observing the resulting allocations.\footnote{When the agent is infinitely patient as in \autoref{sec:microf}, we can exhibit a sequence of strategies under which the agent (approximately) learns this mapping. See the proof of \autoref{lemma:adequate-learning} in \autoref{appendix:aux}.} Without the randomization device, the agent would stand to learn a partition of the set of states, where states in the same cell of the partition induce the same direct mechanism $\mechanism(\cdot,\state,\cor)$. By allowing the designer to rely on the randomization device, we allow him to \emph{obfuscate} the agent's learning beyond a simple partitional structure. Contrast this with the Myersonian benchmark in which without loss of generality the designer would offer mechanisms that do not rely on such devices, that is, $\mmech:\Types\times\States\to\Delta(\Alloc)$. Indeed, the Myersonian designer is not concerned with the agents' learning: without loss of generality, he does not disclose anything about the state to the agents, so that the question of how to optimally release information about the state is moot. 

Lastly, note that we assume the mechanism asks the agents for type reports. In \autoref{appendix:aux}, we show that the revelation principle holds in the setting of this section: it is without loss of generality to focus on direct and incentive compatible mechanisms that induce full participation.

\paragraph{Calibrated information structures} We now describe how a mechanism induces an information structure, which we define using the language in \cite{green2022two} and \cite{gentzkow2017bayesian}. An information structure is a mapping\footnote{\cite{gentzkow2017bayesian}  highlight that the language in \cite{green2022two} allows one to describe the correlation across signal structures. This is exactly what we need to allow the designer to obfuscate the agents' ability to learn. It is again instructive to consider the single-agent case. For each type report $\type\in\Types$, the mechanism can be seen as an information structure $\mechanism(\type,\cdot):\States\times[0,1]\to\Delta(\Alloc)$. Thus, the randomization device allows the designer to control the correlation across these different signal structures, which in turn disciplines what the agent stands to learn when experimenting with different reports.}
\begin{align*}
\excdf:\States\times[0,1]\to \csignals_1\times\dots\times \csignals_N,
\end{align*}
where $\cor\in[0,1]$ is a uniformly distributed random variable\textemdash in fact, it is the same as in the definition of a mechanism\textemdash and 
\[\csignalsi=\Delta(\Alloci)^{\Typesi},\]
is the set of agent $i$'s interim allocation rules.\footnote{The terminology is by analogy to reduced form auctions where the map from own types to own probabilities of being allocated the good are referred to as the interim allocation.}  We choose this language for the information structure to capture the idea that if agent $i$ plays the mechanism repeatedly, she stands to learn how her reports influence her allocation probabilities, i.e., her interim allocation rule. The interim allocation rule, in turn, depends on the mechanism and the strategies of others. Below, we require the interim allocation rule is well-calibrated with the mechanism and others' strategies:
\begin{definition}[Calibrated information structures]
We say that the information structure is calibrated to mechanism \mechanism\ if for all $(\state,\cor)\in\States\times[0,1]$ such that 
$\excdf(\state,\cor)=(\csignal_1,\dots,\csignal_N)$ we have that for all $i\in\{1,\dots,N\}$, all $\typei\in\Typesi$, and all $\alloci\in\Alloci$
\begin{align}\label{eq:calibration}
\csignal_i(\alloci|\typei)=\mathbb{E}_{\typecmi\sim \typedmi(\cdot|\state)}\left[\sum_{\allocmi\in\Allocmi}\mechanism(\typei,\typecmi,\state,\cor)(\alloci,\allocmi)\right].
\end{align}
We denote by \cexp\ the information structure calibrated to mechanism \mechanism.
\end{definition}
In words, the information structure is calibrated if whenever agent $i$ observes that her interim allocation rule in the mechanism is \csignali, then $\csignali$ describes the true probabilities with which agent $i$ gets different allocations \alloci\ as a function of her different type reports \typebi\  in the mechanism. As the right hand side of \autoref{eq:calibration} shows, these probabilities depend on: (i) the mechanism $\mechanism(\cdot,\state,\cor)$, and (ii) others' type reports. Implicit in the definition is that other agents are submitting their reports truthfully. While this is a simplification,\footnote{When we consider mechanisms with arbitrary message spaces in \autoref{appendix:model}, the requirement of calibration is relative to both the mechanism and agents' equilibrium participation and reporting strategies.} it turns out to not be an issue because we study incentive compatible and individually rational mechanisms in the sense we define next.

\paragraph{Information leakage from a mechanism} To close our model, we consider how the mechanism and its induced information structure affect agents' incentives. The mechanism \mechanism\ and the calibrated information structure \cexp\ induce the following game of incomplete information among the agents, where we use Bayes Nash equilibrium as the solution concept. In this game, nature draws (i) the state \state\ from distribution \prior, (ii) $\cor\in[0,1]$ according to the uniform distribution, and (iii) the type profile \type\ from $\typed(\cdot|\state)$. Then, each agent $i$ observes her type $\typei$ and her signal $\csignali=\cexpi(\state,\cor)$. Finally, agents simultaneously decide whether to participate in the mechanism, and conditional on participating what type report to send. Conditional on an agent choosing not to participate, each agent $i$ gets outside option \ooi.\footnote{Thus, we are assuming that $\oo=(\ooi)_{i\in N}$ is an element of \Alloc. In \autoref{appendix:model}, we consider more general participation decisions, allowing the mechanism to condition on the set of participating agents, but even with this extra generality, it is still without loss to restrict attention to mechanisms that induce full participation.} 

Formally, given the mechanism \mechanism\ and its calibrated information structure \cexp, we say that the mechanism is incentive compatible if for all agents $i$, types $\typei\in\Typesi$, signals $\csignali\in\csignalsi$ on the support of \cexpi, the following holds:
\begin{align}\label{eq:ic}\tag{IC$(\typei,\csignali)$}
\typei\in\arg\max_{\typebi\in\Typesi}\mathbb{E}_{(\state,\cor,\typemi)}\left[\payi(\mechanism(\typebi,\typemi,\state,\cor),\typei,\state)|(\typei,\csignali)\right],
\end{align}
where we abuse notation and implicitly (linearly) extend the agent's payoff function to account for lotteries over allocations (conditional on $(\typemi,\state,\cor)$). Furthermore, we say that the mechanism is individually rational if for all agents $i$, types $\typei\in\Typesi$, and signals $\csignali\in\csignalsi$ on the support of \cexpi, the following holds:
\begin{align}\label{eq:ir}\tag{IR$(\typei,\csignali)$}
\mathbb{E}_{(\state,\cor,\typemi)}\left[\payi(\mechanism(\typei,\typemi,\state,\cor),\typei,\state)-\payi(\ooi,\typei,\state)|(\typei,\csignali)\right]\geq0.
\end{align}
Importantly, the agents' incentive and participation constraints must hold for each of their types and each of their private signals, reflecting the agents have access to the information leaked by the mechanism before they play in it. Note, however, the mechanism need not elicit the agents' observed signals, as the mechanism ``knows'' each agent's signal realization.

\paragraph{Calibrated Mechanism Design} In the rest of the paper, we study the problem of \emph{calibrated mechanism design} in which the designer selects a mechanism \mechanism\ that satisfies Equations \ref{eq:ic} and \ref{eq:ir} for all $(i,\typei,\csignali)$, when the signals are drawn according to the calibrated information structure \cexp.

\begin{definition}[Calibrated Mechanism Design]\label{definition:calibration} Let $\ddp:\Arg\to\reals$ denote the designer's payoff and let \cmechanisms\ denote the set of mechanisms that are incentive compatible and individually rational when agents have access to the calibrated information structure. The calibrated mechanism design problem is as follows:
\begin{align}\label{eq:opt-cal}\tag{OPT$_{\mathrm{cal}}$}
\max_{\mechanism\in\cmechanisms}\mathbb{E}_{(\state,\cor,\type)}\left[\ddp(\mechanism(\type,\state,\cor),\type,\state)\right].
\end{align}
We refer to elements of \cmechanisms\ as calibrated mechanisms and the solution to \ref{eq:opt-cal} as the optimal calibrated mechanism.
\end{definition}
Three comments are in order: 

First, calibration imposes a constraint on the designer vis-\`a-vis standard mechanism design. After all, the incentive and participation constraints faced by the designer are \emph{endogenous} to the mechanism. The more the designer's mechanism depends on the state, the more informative the calibrated information structure is, and the more incentive constraints the designer faces. Only when each agent's interim allocation rule is constant in \state\ does the mechanism not leak information and the incentive and participation constraints reduce to the standard ones.

Second, in the single-agent setting, the calibration constraint admits two complementary interpretations. Throughout the paper, we emphasize the learning-by-experimentation interpretation: the calibrated information structure represents what the agent can ultimately infer by repeatedly interacting with the mechanism. Accordingly, the designer should ensure incentive compatibility with respect to the full information the agent eventually obtains. At the same time, calibration can also be interpreted as a \emph{transparency} requirement. Indeed, upon observing signal $\csignal:\Types\to\Delta(\Alloc)$, the agent knows the consequences of her choices in the mechanism, even if she does not know the state.\footnote{It is common for online platforms to inform agents of the consequences of their choices: marketplaces inform sellers of their probability of sale at different posted prices, and transportation providers inform riders of their probability of receiving a seat upgrade at different bid levels. Even insurance companies provide consumers with projected expenditures under different plan choices.}

% Second, the calibration constraint admits two interpretations in the single-agent setting. Throughout the paper, we favor the learning-by-experimentation interpretation: the calibrated information structure represents the information the agent can eventually infer by repeatedly interacting with the mechanism. Hence, the designer should make his mechanism incentive compatible given all information the agent eventually has access to. At the same time, calibration can also be interpreted as a \emph{transparency} requirement. Indeed, upon observing signal $\csignal:\Types\to\Delta(\Alloc)$, the agent knows the consequences of her choices in the mechanism, even if she does not know the state.\footnote{It is common for online platforms to inform agents of the consequences of their choices: marketplaces inform sellers of their probability of sale at different posted prices, and transportation providers inform riders of their probability of receiving a seat upgrade at different bid levels. Even insurance companies provide consumers with projected expenditures under different plan choices.}

With multiple agents, these interpretations differ. The natural extension of the transparency requirement is that agents learn the mapping from profiles of type reports to lotteries over profiles of allocations before playing the mechanism. By contrast, the calibrated information structure reveals to each agent her interim allocation rule, that is, the mappings from her own reports to lotteries over her own allocations. As we discuss in the next section, the gap between these two interpretations is the gap between the designer publicly or privately disclosing information about the state to the agents.

Lastly, the definition of calibration assumes agents only learn about the state through their allocations in the mechanism, and not their payoffs.\footnote{Indeed, in the analysis of the dynamic interaction in \autoref{sec:microf}, we assume the agent only observes her type and her allocation, but not her payoffs.}\footnote{This assumption is routinely made in dynamic settings. See \cite{pavan2014dynamic} and \cite{cesa2024role} for two examples in the context of agents' behavior within mechanisms.}
 This assumption allows us to focus on the information that the mechanism leaks \emph{regardless of payoff assumptions}. This allows us to avoid situations in which the mechanism does not condition the allocation on the state, but the agents learn because they have different payoffs from the same allocation in different states; or the mechanism conditions on the state, but this information is not payoff relevant to (some types of) the agent. Our microfoundation in \autoref{sec:repeated} in fact deals with this last wrinkle: We show that even if the agent extracts less information than that in the calibrated information structure, she learns enough that her payoff is \emph{as if} she had access to the calibrated information structure.

We conclude this section by illustrating how our static solution concept captures the dynamics we alluded to in the introductory example:
\begin{example}[Selling a good under demand uncertainty]\label{example:cml} Consider again the example in the introduction, in which a buyer with binary values $\val\in\{1/2,1\}$ faces a seller who knows whether demand is high $(\state=H)$ or low $(\state=L)$. The left panel of \autoref{table:cml-calibration} describes the probabilities of trade and payments of the optimal (Myersonian) mechanism. In the introduction, we discussed this mechanism fails to extract full surplus in the long run as the buyer would quit the mechanism after seeing her allocation is $(1,3/2)$. We now describe this in the language of calibration.

The right panel of \autoref{table:cml-calibration} describes the information structure induced by the surplus extraction mechanism. Because in this mechanism the buyer's allocation does not depend on her values, we describe signals as allocations. The calibrated information structure is fully informative: when the state is $L$, the buyer sees signal $(1,0)$ with probability $1$, and when the state is $H$, she sees signal $(1,3/2)$ with probability $1$. 
\begin{table}[h!]
\centering
\begin{tabular}{lcc}
\begin{tabular}{c|cc}
&$\val=1/2$&$\val=1$\\
\hline
$\state=L$&$(1,0)$&$(1,0)$\\
$\state=H$&$(1,3/2)$&$(1,3/2)$
\end{tabular}&&
\begin{tabular}{c|cc}
&$(1,0)$&$(1,3/2)$\\
\hline
$\state=L$&$1$&$0$\\
$\state=H$&$0$&$1$
\end{tabular}
\end{tabular}
\caption{Trade probabilities and payments in optimal Myersonian mechanism (left); calibrated information structure (right). We describe signals as allocations, because the mechanism does not screen the buyer's values.}\label{table:cml-calibration}
\end{table}

When the buyer has access to the calibrated information structure before playing the mechanism, the surplus extraction mechanism does not satisfy the buyer's participation constraints, which must hold for each buyer value and each signal she observes. In particular, when the buyer sees signal $(1,3/2)$, she knows her payoff in the mechanism is negative and quits. Thus, the calibration constraint prevents the seller from extracting the buyer's surplus. In this case, the restriction induced by calibration endogenously provides the buyer with \emph{withdrawal rights}, which, as \cite{haberman2025auctions} show, prevent sellers from employing Cr\'emer-McLean-style schemes.\footnote{In \autoref{sec:main}, we provide an example in which when agents have access to the calibrated information structure the optimal Myersonian mechanism fails to be incentive compatible.}

Consider now the mechanism in the left panel of \autoref{table:opt-calibration}, which corresponds to posting a price of $1/2$ when the state is $L$ and a price of $1$ when the state is $H$. The right panel of \autoref{table:opt-calibration} depicts the calibrated information structure. Note that when the state is $H$, the information structure sends with probability $1$ the interim allocation rule $\{(1/2,(0,0)),(1,(1,1))\}$, representing that if the buyer reports her value is $1/2$ she gets nothing and pays nothing, whereas if her report is $1$, she obtains the good at a price of $1$.
\begin{table}[h!]
\centering
\begin{tabular}{ccc}
\begin{tabular}{c|cc}
&$\val=1/2$&$\val=1$\\
\hline
$\state=L$&$(1,1/2)$&$(1,1/2)$\\
$\state=H$&$(0,0)$&$(1,1)$
\end{tabular}&&
\begin{tabular}{c|cc}
&\{(1,1/2)\}&\{(1/2,(0,0)),(1,(1,1))\}\\
\hline
$\state=L$&$1$& $0$\\
$\state=H$&$0$&$1$
\end{tabular}
\end{tabular}
\caption{Trade probabilities and payments in optimal calibrated mechanism (left); calibrated information structure (right)}\label{table:opt-calibration}
\end{table}

Note that the mechanism is incentive compatible and individually rational when the buyer has access to the calibrated information structure. As the results that follow allow us to establish, this is indeed the optimal calibrated mechanism. 
\end{example}

\paragraph{Private value environments} A natural case to consider is that when agents' payoffs are state independent, that is,  for each agent $i$, the agent's utility function can be written as $\payi(\alloci,\typei)$. Under private values, the state describes either statistical information about the agents' types as in \autoref{example:cml}, or a payoff-relevant variable for the designer.
%
%
%so that either the designer has statistical information about the agents' types as in \autoref{example:cml}, or the state is payoff relevant only to the designer. Formally, for each agent $i$, the agent's utility function can be written as $\payi(\alloci,\typei)$, and hence this corresponds to a private values environment.

\autoref{theorem:pv} collects our main characterization result for this case. To state it, let $\mechanism_\full$ denote the following mechanism: For each $(\state,\cor)\in\States\times[0,1]$, $\mechanism_\full(\cdot,\state,\cor):\Types\to\Delta(\Alloc)$ is the designer optimal incentive compatible and individually rational direct mechanism when it is common knowledge that the state is \state. 
\begin{theorem}[Private values]\label{theorem:pv}
Under private values, the designer's payoff under the optimal calibrated mechanism is the same payoff he would obtain by choosing $\mechanism_\full$.
\end{theorem}
That is, in private values environments, the calibration constraint pushes the designer toward full transparency. In particular, in settings with transferable utility in which the designer possesses statistical information about the agents' types, \autoref{theorem:pv} implies the designer cannot engage in Cr\'emer-McLean style schemes under calibration, and hence extract full surplus. Whereas the implication of calibrated mechanism design in private values environments is powerful, the result is fairly intuitive: The designer benefits from making the mechanism opaque by pooling states inasmuch as it weakens the incentive or participation constraints of the agents. Under private values, however, agents' incentive constraints depend on the state only through the mechanism, and calibration imposes constraints on the mechanism state-by-state.\footnote{A tempting comparison is \citet[Prop. 11]{maskin1990principal}: with private values and quasilinear utilities, the informed principal's unique equilibrium payoff coincides with the state-by-state optimum. The authors show this conclusion depends on quasilinearity:  absent this assumption, an informed principal can benefit from concealing his information in the case of private values. Instead, \autoref{theorem:pv} relies neither on quasilinearity nor on the designer's lack of commitment.}

\section{Two-stage mechanisms}\label{sec:main}
In this section, we introduce an alternative representation of calibrated mechanisms that we use throughout our illustrations. We introduce it first for the case of a single agent and then for multiple agents.

\paragraph{Single-agent case and two-stage mechanisms} We find it instructive to first consider the case $N=1$, and for simplicity drop the subscripts $1$  from the notation. Consider a mechanism \mechanism\ and its calibrated information structure \cexp. When the agent of type \type\ observes signal \csignal, two things happen: On the one hand, the agent updates her prior, $\prior(\state|\type)$\footnote{Formally, 
\begin{align*}
\prior(\state|\type)=\frac{\prior(\state)\typed(\type|\state)}{\sum_{\stateb\in\States}\prior(\stateb)\typed(\type|\stateb)}.
\end{align*}}, to some belief $\belief(\type,\csignal)\in\Beliefs$. On the other hand, the agent learns that she faces allocation rule \csignal\ in the mechanism. Thus, her payoff in the mechanism when her type is \type, observes signal \csignal, and reports \typeb\ can be written as follows:
\begin{align}\label{eq:calibrated-payoff}
\mathbb{E}_{(\state,\cor)}\left[\pay(\mechanism(\typeb,\state,\cor),\type,\state)|(\type,\csignal)\right]=\sum_{\alloc\in\Alloc}\csignal(\alloc|\typeb)\left(\sum_{\state\in\States}\belief(\state|\type,\csignal)\pay\aarg\right).
\end{align}
In other words, the information structure \cexp\ provides the agent with all the necessary information to evaluate her payoffs in the mechanism: her belief about the state \emph{and} her allocation rule.  This allocation rule $\csignal:\Types\to\Delta(\Alloc)$ satisfies two properties. First,  because under calibration \csignal\ is the true interim allocation rule faced by the agent, she learns no further information about the state beyond that contained in $\belief(\type,\csignal)$. Second,  Equations \ref{eq:ic} and \ref{eq:ir} imply the allocation rule is incentive compatible and individually rational when the agent holds belief $\belief(\type,\csignal)$. 

The above discussion suggests an alternative representation of a calibrated mechanism, which we dub a two-stage mechanism and define as follows:
\begin{definition}[Two-stage mechanisms]\label{definition:ts}
A two-stage mechanism is a mapping $\tsmech:\Types\times\States\to\Delta(\Alloc\times\Beliefs)$ such that a Bayes plausible Blackwell experiment $\beliefr:\States\to\Delta(\Beliefs)$ and an allocation rule $\allocr:\Types\times\Beliefs\to\Delta(\Alloc)$ exist such that for  all $(\type,\state)\in\Types\times\States$ and all measurable subsets $\tilde\Delta\subset\Beliefs$,\footnote{We refer the reader to the appendix for our mathematical conventions, in particular, the definition of the corresponding $\sigma$-algebras.}
\begin{align*}
\tsmech(\{\alloc\}\times\tilde{\Delta}|\type,\state)=\int_{\tilde{\Delta}}\allocr(\alloc|\type,\belief)\beliefr(d\belief|\state).
\end{align*}
We say the two-stage mechanism is incentive compatible and individually rational if on the support of $\prior\otimes\beliefr$, the allocation rule $\allocr(\cdot|\cdot,\belief):\Types\to\Delta(\Alloc)$ is incentive compatible and individually rational conditional on the agent observing \belief.
\end{definition}
In a two-stage mechanism, the designer first discloses information about \state\ in the form of a belief $\belief$ about \States, and conditional on that belief\textemdash but not the state\textemdash offers a direct mechanism $\allocr(\cdot|\cdot,\belief):\Types\to\Delta(\Alloc)$. Two aspects of two-stage mechanisms are worth highlighting: First, the disclosure is type-independent. The designer discloses information to the agent without first communicating with the agent. Second, because the direct mechanism $\allocr(\cdot|\cdot,\belief)$ does not depend on \state, observing the allocation reveals no further information about the state.

Lastly, when we say the experiment \beliefr\ is Bayes plausible,  we mean that the distribution of posteriors induced by \beliefr\ has mean \prior, and hence we can interpret \belief\ as the designer's belief about the state conditional on observing \belief.\footnote{Formally, define the belief distribution induced by $\beliefr$, $\bsplit_\beliefr=\prior\otimes\beliefr$. The claim is that $\mathbb{E}_{\bsplit_\beliefr}\left[\belief\right]=\prior$.} Whereas the designer and the agent do not necessarily have the same beliefs about the state, the agent's beliefs about the state conditional on observing \belief\ obtain from a known transformation from those of the designer \citep{alonso2016bayesian,laclau2017public}. Thus, ensuring Bayes plausibility with respect to \prior\ suffices.

\autoref{theorem:cmd-two-stage} shows that (incentive compatible and individually rational) calibrated mechanisms and two-stage mechanisms implement the same distributions over outcomes $\outcome\in\Delta(\Arg)$:
\begin{theorem}[Two-stage and calibrated mechanisms]\label{theorem:cmd-two-stage}Suppose $N=1$. An outcome distribution $\outcome\in\Delta\left(\Arg\right)$ is implementable by an incentive compatible and individually rational calibrated mechanism if and only if it is implementable by an incentive compatible and individually rational two-stage mechanism. That is, if and only if
\begin{align}\label{eq:ts-impl-0}
\outcome\aarg=\prior(\state)\typed(\type|\state)\int_{\Beliefs}\allocr(\alloc|\type,\belief)\beliefr(d\belief|\state),
\end{align}
for some Bayes plausible $\beliefr:\States\to\Delta(\Beliefs)$ and incentive compatible and individually rational $\allocr:\Types\times\Beliefs\to\Delta(\Alloc)$.
\end{theorem}
The proof of this and all results in this section can be found in \autoref{appendix:main}. 

In the single-agent case, \autoref{theorem:cmd-two-stage} shows that the calibrated mechanism design problem is equivalent to a standard mechanism design problem in which we restrict the designer to using a specific class of mechanisms; namely, incentive compatible and individually rational two-stage mechanisms. As we explained above, a mechanism \mechanism\ and its calibrated information structure \cexp\ can be seen as actually inducing a joint distribution over \Argb. \autoref{theorem:cmd-two-stage} implies this joint distribution admits two conditional independence properties. First, the allocation is conditionally independent of the state, conditional on the agent's type and the induced belief.\footnote{This is a consequence of Bayes rule: beliefs are a sufficient statistic for \state. Hence, conditional on $(\type,\belief)$, the allocation rule carries no more information about the state.} This follows from the signals \csignal\ carrying no further information about the state than that what is contained in the agent's belief.  Second, the designer disclosed belief is conditionally independent of the agent's type conditional on the state. In the static setting of \autoref{sec:model}, this is because the calibrated information structure discloses information to the agent uniformly across her types. In the dynamic setting of \autoref{sec:repeated}, this type-independent disclosure arises endogenously because the agent's experimentation opportunities are independent of her type.

\paragraph{Two-stage mechanisms solve calibrated mechanism design} \autoref{theorem:cmd-two-stage} is of practical import as it provides a recipe of sorts for characterizing the designer's optimal calibrated mechanism (see the applications in \autoref{sec:app}). For each $\belief\in\Beliefs$, the designer chooses a mechanism $\allocr(\cdot|\cdot,\belief):\Types\to\Delta(\Alloc)$ that maximizes his expected payoff when the designer believes \belief\ is the distribution of  states, and subject to the agent's incentive compatibility and individually rational constraints conditional on the designer's belief being \belief. Proceeding in this way, we obtain the designer's value function $\DP:\Beliefs\to\reals$. The optimal Blackwell experiment obtains from the concavification of \DP. We illustrate this procedure with two examples: 

\setcounter{example}{0}
\begin{example}[continued] Consider again the seller-buyer example, in which the buyer is privately informed about her value for the good and the seller knows the demand state. By \autoref{theorem:cmd-two-stage}, we can find the seller's optimal calibrated mechanism as follows. First, equate \belief\ with the probability that the state is $H$. For each $\belief\in[0,1]$, consider the following problem:
\begin{align}\label{eq:cmd-cml}
\DP(\belief)\equiv\max_{(q,t):V\to[0,1]\times\reals}&\belief\left(\frac{2}{3}t(1)+\frac{1}{3}t(1/2)\right)+(1-\belief)\left(\frac{1}{3}t(1)+\frac{2}{3}t(1/2)\right)
\\
\text{s.t.}&\left\{\begin{array}{ll}(\forall \val\in\{1/2,1\})&\val q(\val)-t(\val)\geq0\\
(\forall\val,\valb\in\{1/2,1\},\val\neq\valb)&\val q(\val)-t(\val)\geq\val q(\valb)-t(\valb)\end{array}\right.\nonumber.
\end{align}
That is, the seller chooses an incentive compatible and individually rational selling mechanism that maximizes his expected revenue when his belief is \belief. Because \state\ is not payoff relevant to the buyer\textemdash it is just statistical information about the buyer's valuation\textemdash\emph{and} the mechanism does not depend on state, the buyer's belief about \state\ does not enter her incentive constraints. 

The solution to the seller's problem in \autoref{eq:cmd-cml} is simple: the seller posts a price of $1/2$ when $\belief\leq 1/2$ and a price of $1$ when $\belief>1/2$. Hence, the seller's value function is given by
\begin{align*}
\DP(\belief)=\max\left\{\frac{1}{2},\belief\frac{2}{3}+(1-\belief)\frac{1}{3}\right\},
\end{align*}
and is illustrated by the solid line in blue on \autoref{fig:cmd-cml}. In words, the seller either sells the good at a price of $1/2$ and the buyer buys with probability 1, or he sells the good at a price of $1$ and the buyer buys whenever her value is $1$, which happens with the probability in the second argument of the $\max$. 

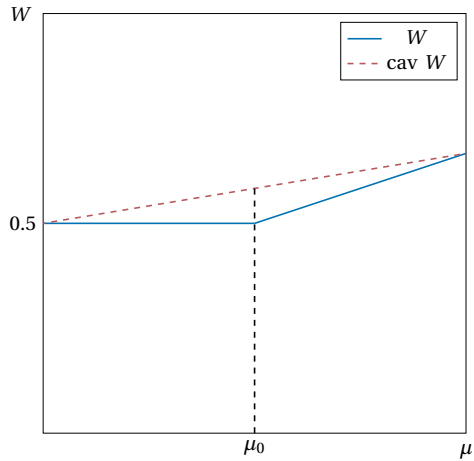
\begin{figure}[th!]
\centering
\scalebox{0.75}{%
\begin{tikzpicture}
\begin{axis}[xmin=0,xmax=1,ymin=0,ymax=1,xtick={1/2},xticklabels={$\prior$},xticklabel style={below},ytick={0.5},yticklabels={0.5},yticklabel style={left},xlabel=$\belief$,ylabel=$\DP$,x label style={at={(axis description cs:1,-0.01)}},
    y label style={at={(axis description cs:-0.01,1)},rotate=-90}
    ,width=9cm,height=9cm]
    \addplot[domain=0:1/2,blue,thick]{0.5};
\addplot[domain=1/2:1,blue,thick,forget plot]{(1/3)*(1+x)};
\addlegendentry{$\DP$}
\addplot[dashed, thick,red]{(1/6)*x+1/2};
\addlegendentry{$\cav\;\DP$}
            \addplot[black,thick,dashed] coordinates {(0.5,0) (0.5,7/12)};
\end{axis}
\end{tikzpicture}}
\caption{Seller's payoff in \autoref{example:cml}.}\label{fig:cmd-cml}
\end{figure}

The optimal calibrated mechanism can be read from the concavification of \DP, which is the dashed, red line in \autoref{fig:cmd-cml}: The seller first reveals the state to the agent, and offers a price of $1/2$ when $\state=L$ and a price of $1$ when $\state=H$.
\end{example}

\autoref{example:cml} illustrates a more general principle that provides additional intuition for \autoref{theorem:pv}. In the private values case and when $N=1$, the designer's value function $\DP:\Beliefs\to\reals$ is convex. As \autoref{eq:cmd-cml} illustrates, the designer maximizes a linear function in beliefs subject to constraints that do not depend on the induced belief. Convexity of \DP\ implies full disclosure is (weakly) optimal, and \autoref{theorem:pv} follows.

\begin{example}[Horizontal differentiation]\label{example:horizontal} Consider a seller who owns a good of unknown type, $\state\in\{L,R\}$, and a buyer whose private information is indexed by $\Types=\{\type_1,\type_2,\type_3\}$. Assume the good's type (the state) and the buyer's types are independent, and equally likely. \autoref{table:values} describes the buyer's value for the seller's good as a function of hers and the good's type, $\val(\type,\state)$. When the good is $\state=L$, the buyer of type $\type_3$ has the highest value for the good, whereas when the good is $\state=R$, the buyer of type $\type_3$ has the lowest value for the good. 

\begin{table}[h!]
\centering
\begin{tabular}{c|ccc}
&$\type_1$&$\type_2$&$\type_3$\\
\hline
$\state=L$&$1$&$2$&3\\
$\state=R$&$2$&$2$&1
\end{tabular}
\caption{Buyer's values.}\label{table:values}
\end{table}

Suppose the buyer's utility is quasilinear, that is, $\pay(q,t,\type,\state)=q\val(\type,\state)-t$, and the seller wishes to maximize his revenue. Furthermore, assume the buyer's outside option is no trade.

Consider first the optimal mechanism the designer would offer absent the calibration constraint, depicted in the top panel of \autoref{table:horizontal-myerson}. This mechanism asks types $\type_2$ and $\type_3$ for a payment of $2$ and allocates the good with probability 1, regardless of its kind. Instead, it asks the buyer of $\type_1$ to pay $1$ in exchange for getting the good only when it is of her favorite kind $(\state=R)$. 

\begin{table}[h!]
\centering

\begin{tabular}{c|ccc}
&$\type_1$&$\type_2$&$\type_3$\\
\hline
$\state=L$&$(0,1)$&$(1,2)$&(1,2)\\
$\state=R$&$(1,1)$&$(1,2)$&(1,2)
\end{tabular}
\vspace{0.75cm}

\begin{tabular}{c|cc}
&$\{(\type_1,(0,1)),(\type_2,(1,2)),(\type_3,(1,2))\}$&$\{(\type_1,(1,1)),(\type_2,(1,2)),(\type_3,(1,2))\}$\\
\hline
$\state=L$&$1$&0\\
$\state=R$&$0$&1
\end{tabular}
\caption{ Trade probabilities and transfers in the optimal  mechanism (top); calibrated information structure (bottom).}\label{table:horizontal-myerson}
\end{table}

The bottom panel of \autoref{table:horizontal-myerson} depicts the information structure calibrated to the optimal mechanism. It sends two signals: when the good is $L$, the buyer can choose to either not get the good and pay $1$, or get the good and pay $2$. Instead, when the good is $R$, the buyer is choosing between paying $1$ or $2$ to obtain the good with probability 1. 

Under the calibrated information structure, the optimal mechanism is neither incentive compatible nor individually rational. When the good is $R$, the buyer would prefer to choose $(1,1)$ \emph{regardless} of her type. Instead, when the good is $L$, the buyer of $\type_1$ would quit the mechanism instead of paying $1$ and getting nothing.

To characterize the optimal calibrated mechanism, we rely again on two-stage mechanisms. Equate $\belief$ with the probability that the good is $R$. Note that because states and types are independent, if the seller assigns probability \belief\ to the state being $R$, so does the buyer (and vice versa). For each $\belief\in[0,1]$, the seller solves the following problem
\begin{align}\label{eq:cmd-horizontal}
\DP(\belief)\equiv&\max_{(q,t):\Types\to[0,1]\times\reals}\sum_{\type\in\Types}\frac{1}{3}t(\type)\\
\text{s.t.}&\left\{\begin{array}{ll}(\forall \type\in\{\type_1,\type_2,\type_3\})&q(\type)\mathbb{E}_{\belief}\val(\type,\cdot)-t(\type)\geq0\\
(\forall \type,\typeb\in\{\type_1,\type_2,\type_3\},\typeb\neq\type)&q(\type)\mathbb{E}_{\belief}\val(\type,\cdot)-t(\type)\geq q(\typeb)\mathbb{E}_{\belief}\val(\type,\cdot)-t(\typeb)\end{array}\right.\nonumber.
\end{align}

In this case, the seller's  objective function does not depend on the induced belief $\belief$ as types and states are independent. Instead, the buyer's incentive and individual rationality constraints do depend on $\belief$ as the state is payoff relevant. The solution to the problem in \autoref{eq:cmd-horizontal} is a posted price, whose value depends on \belief. For instance, when $\belief\in\{0,1\}$, the optimal price is $2$ and the seller's revenue is $4/3$. Instead, when $\belief=2/3$, the optimal price is $5/3$ and profits are maximal and equal to $5/3$. Indeed, when $\belief=2/3$, the heterogeneity across buyer types is minimized (and hence, their rents), and by setting $p=5/3$ all buyer types buy. The blue line in \autoref{fig:cmd-horizontal} depicts the seller's expected profit as a function of his belief \belief.

\begin{figure}[th!]
\centering
\scalebox{0.75}{
\begin{tikzpicture}
\begin{axis}[xmin=0,xmax=1,ymin=1,ymax=2,xtick={1/2},xticklabels={\prior},yticklabels={},xticklabel style=below,xlabel=$\belief$,ylabel=$\DP(\belief)$,x label style={at={(axis description cs:1,-0.01)}},
    y label style={at={(axis description cs:-0.01,1)},rotate=-90}
    ,width=9cm,height=9cm]
    \addplot[thick,blue,forget plot,domain=0:1/3]{4/3};
        \addplot[thick,blue,forget plot,domain=1/3:2/3]{1+x};
    \addplot[thick,blue,forget plot,domain=2/3:7/8]{3-2*x};
        \addplot[thick,blue,domain=7/8:1]{2/3+(2/3)*x};
            \addlegendentry{$\DP(\belief)$}
            \addplot[thick,dashed,red,forget plot,domain=0:2/3]{0.5*x+4/3};
            \addplot[thick,dashed,red,domain=2/3:1]{-x+7/3};
            \addlegendentry{$\cav\;\DP(\belief)$}
            \addplot[black,thick,dashed] coordinates {(0.5,0) (0.5,19/12)};
\end{axis}
\end{tikzpicture}}
\caption{Seller's profit in the two-stage mechanism}\label{fig:cmd-horizontal}
\end{figure}
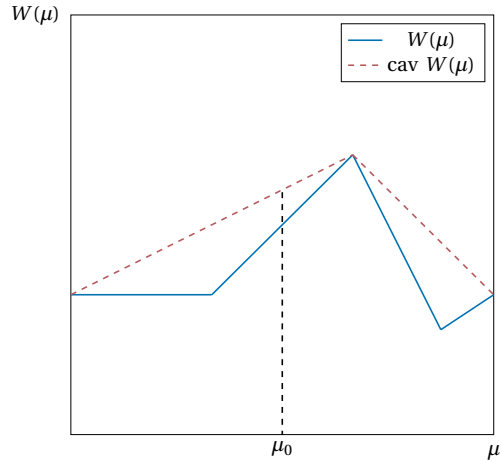

The optimal calibrated mechanism can be read from the concavification of \DP\ at $\prior=1/2$, depicted by the dashed red line in \autoref{fig:cmd-horizontal}. %\annotation{Should we cite Loertscher and Muir who also make the point of obfuscation in horizontal differentiation?}
The seller provides the buyer with partial information about the good: He either reveals the good is $L$ and sells the good at a price of $2$, or he obfuscates the good\textemdash inducing a belief of $2/3$\textemdash and sets a price of $5/3$.
\end{example}

Another consequence of \autoref{theorem:cmd-two-stage} is that without loss of generality, we can focus on calibrated mechanisms with finite calibrated information structures:
\begin{corollary}[Support of calibrated information structures]
It is without loss of generality to restrict attention to two-stage mechanisms that induce at most $|\States|$ beliefs.%It is without loss of generality to restrict attention to calibrated mechanisms \mechanism\ such that \cexp\ induces at most $|\States|$ signals.
\end{corollary}
In other words, it is without loss of generality to focus on calibrated mechanisms that induce at most $|\States|$ \emph{allocation rules}.

\paragraph{Multiple agents and generalized two-stage mechanisms} In the case of multiple agents, we can also interpret a calibrated mechanism as conveying to each agent $i$ both the information she should have about the state upon seeing signal \csignali, $\belief_i(\typei,\csignali)$, \emph{and} her  interim allocation rule, $\csignali:\Typesi\to\Delta(\Alloci)$. However, two differences arise relative to the single-agent case: First, each agent $i$ receives her information privately from that of other agents. Second, even if the agents put together the information they receive, this is not enough to learn the ex-post allocation rule, that is, the map from type profiles to allocations. After all, each agent $i$ observes her  interim allocation rule alone. These differences are natural when we think of calibrated mechanisms as capturing the information agents stand to learn from experimenting with the mechanism:   There is no reason all agents will learn the same information, and from observing her own allocations, and not those of others, an agent can only learn about her interim allocation rule, not the ex-post one.%\footnote{Provided others' types are redrawn each period.}

These observations together imply that to describe the analogue of a two-stage mechanism in multi-agent settings we need to (i) allow for agent-by-agent information disclosure, and (ii) keep track that the interim allocation rules are consistent with the same ex-post allocation rule. These considerations motivate the following generalization of a two-stage mechanism:

\begin{definition}[Generalized two-stage mechanism]\label{definition:ts-gral}A generalized two-stage mechanism is a mapping $\tsmech:\Types\times\States\to\Delta(\Beliefs^N\times\Alloc)$ for which a tuple of mappings
\begin{align*}
\beliefr:\States\to\Delta(\Beliefs^N),\;\;\;\allocr_i:\Typesi\times\Beliefs\to\Delta(\Alloci),\;\;\; \allocr:\Types\times\States\times\Beliefs^N\to\Delta(\Alloc),
\end{align*}
exist such that:
\begin{enumerate}
\item For all $(\type,\state)\in\Types\times\States$, and all measurable subsets $(\tilde{\Delta}_i)_{i=1}^N\subset\Beliefs^N$, we have
\begin{align*}
\tsmech(\times_{i=1}^N\tilde{\Delta}_i\times\{\alloc\}|\type,\state)=\int_{\times_{i=1}^N\tilde{\Delta}_i}\allocr(\alloc|\type,\state,\belief_1,\dots,\belief_N)\beta(d(\belief_1,\dots,\belief_N)|\state)
\end{align*}

\item The Blackwell experiment $\beta$ is Bayes plausible,
\item\label{itm:gts-iar} For all $i\in \{1,\dots,N\}$, the interim allocation rule $\allocr_i$ satisfies that for all measurable subsets $\tilde{\Delta}$ of \Beliefs\ and all $\argi\in\Argi$
\begin{align*}
\int_{\tilde\Delta\times\Beliefs^{N-1}}\left\{\allocr_i(\alloci|\typei,\belief_i)-\mathbb{E}_{\typedmi(\cdot|\state)}\left[\sum_{\allocmi\in\Allocmi}\allocr(\alloci,\allocmi|\typei,\typemi,\state,\belief_i,\belief_{-i})\right]\right\}\beta(d(\belief_i,\belief_{-i})|\state)=0.
\end{align*}
\end{enumerate}
We say the generalized two-stage mechanism is incentive compatible and individually rational if for all $i\in\{1,\dots,N\}$, on the support of $\prior\otimes\beliefr$, $\allocr_i(\cdot|\cdot,\belief_i)$ is incentive compatible and individually rational for agent $i$ when she learns $\belief_i$.
\end{definition}
As anticipated, generalized two-stage mechanisms differ from two-stage mechanisms in three ways when $N>1$. First, because disclosures are private, the experiment \beliefr\ now outputs a profile of beliefs, one for each agent. As shown in \cite{arieli2024feasible}, \beliefr\ is Bayes plausible if and only if for each agent $i$, the marginal Blackwell experiment $\beliefr_i$ is Bayes plausible. Second, while the individual interim allocation rule $\allocr_i$ only depends on the disclosed belief to agent $i$, $\belief_i$, and not the state, the ex-post allocation rule $\allocr$ may depend on the state, even conditional on the belief profile $(\belief_1,\dots,\belief_N)$. The reason is that this belief profile is no longer a sufficient statistic for the ex-post allocation rule as each agent $i$ only observes their interim allocation. Third and relatedly, we need to keep track of both the interim allocation rules $(\allocr_i)_{i=1}^N$ and the ex-post allocation rule $\allocr$ to check that the interim allocation rules are consistent with the same mechanism. An interesting question for future work would be to characterize which interim allocation rules $(\allocr_i)_{i=1}^N$ are consistent with some ex-post allocation rule \allocr, so that one could focus on the interim allocation rules alone.

As we show in \autoref{theorem:cmd-two-stage-multi}, a calibrated mechanism induces a generalized two-stage mechanism:
\begin{proposition}\label{theorem:cmd-two-stage-multi} If outcome distribution $\outcome\in\Delta\left(\Arg\right)$ is implementable by an incentive compatible and individually rational calibrated mechanism, then it is implementable by an incentive compatible and individually rational generalized two-stage mechanism.
\end{proposition}
In contrast to the single-agent case, not every outcome distribution implemented by a generalized two-stage mechanism can be implemented by a calibrated mechanism. On the one hand, no agent's beliefs are a sufficient statistic for the information the mechanism leaks about the state, so that the allocation rule $\bar\allocr$ may still leak information about the state or others' beliefs, which in turn leak information about the state. On the other hand, because in a calibrated mechanism each agent learns her interim allocation rule conditional on $(\state,\cor)$, the incentive and participation constraints associated to a generalized two-stage mechanism are weaker than those implied by a calibrated mechanism whenever multiple interim allocation rules underlie the same belief: Even if the average interim allocation rule $\allocr_i$ is incentive compatible and individually rational, each of the interim allocation rules underlying that average need not be.

\section{Applications}\label{sec:app}
In this section, we study optimal calibrated mechanism design in canonical mechanism design settings with quasilinear utilities. We first consider the case of a single agent, with single-dimensional types and allocations, and supermodular payoffs. In \autoref{sec:transparent}, we show that if the order of types is state independent, then optimal two-stage mechanisms fully reveal the state, whereas this conclusion can be reversed when the order of types is state-dependent. In \autoref{sec:myerson}, we compare optimal calibrated mechanism design against the Myersonian benchmark. Lastly, we analyze a multi-agent application in \autoref{sec:optimal-auction}.

\subsection{Calibrated Screening}\label{sec:transparent}
We consider the following version of the model in \autoref{sec:model}. Suppose $N=1$ and let $\Types=[\mint,\maxt]$ denote the set of types. Assume \type\ is distributed according to a full support distribution \typecdf\ with density \typed. Hence, throughout, we consider the case in which the agent's type is independent of \state. Denote the set of allocations by $\Alloc=[0,\maxq]\times\reals$, where $q\in[0,\maxq]$ is the (physical) allocation and $t\in\reals$ is a payment from the agent to the designer.\footnote{In contrast to the model of \autoref{sec:model}, we are assuming the set of types and allocations to be intervals in the real line. The results in the previous sections go through with richer type and allocation spaces, at the cost of more notation.}

The agent's and the designer's payoffs are given by $\pay\argq-t$ and $\ddp\argq+t$, respectively. Assume that if the agent does not participate, then the outside option is $\oo=(0,0)$, and that this yields a payoff of $0$ to both the designer and the agent. Throughout, we assume that for each $\state\in\States$, the family of functions $\left\{\type\mapsto\pay\argq:q\in[0,\maxq]\right\}$ is equi-Lipschitz on \Types: a positive constant $L_\state$ exists such that for all $\type,\typeb\in\Types$ and $q\in[0,\maxq]$, $|\pay(q,\type,\state)-\pay(q,\typeb,\state)|\leq L_\state|\type-\typeb|$.\footnote{This assumption ensures the Lipschitz continuity of the agent's indirect utility function when the allocation space is infinite. Because \States\ is finite, requiring the condition to hold state-by-state suffices.} Furthermore, the analysis that follows restricts attention to mechanisms that do not randomize on the allocation (beyond the inherent randomness of $\States\times[0,1]$). \autoref{remark:deterministic} at the end of this section discusses settings in which this is not a restriction and how to generalize the observations herein when random allocations are allowed.

Our goal is to characterize the designer optimal calibrated mechanism and how its properties depend on how the state affects the order of types.

\paragraph{State-independent type ranking} We consider first the case in which the order of types is independent of the state. Formally, assume that for all $\state\in\States$, the function $\pay(\cdot,\state)$ is supermodular in $(q,\type)$. That is, in all states, the agent with higher value of \type\ values $q$ more.  %To ensure that optimal mechanisms do not randomize on $q$, we also assume that for all $(\type,\state)$, $\pay(\cdot,\type,\state)$ is concave in $q$ and $\pay_{qq}$ is increasing in \type.\footnote{Throughout the paper, “increasing” and “concave” are understood in the weak sense (i.e., non-decreasing and weakly concave).} 
These assumptions are satisfied, for instance, for $\pay\argq=\type\state q$ or $\pay\argq=(\type+\state)q$. 
   
 By \autoref{theorem:cmd-two-stage}, we can characterize the optimal calibrated mechanism via two-stage mechanisms. To do so, we solve the problem ``backward'': For each $\belief\in\Beliefs$ the designer may induce about the state, the designer chooses an optimal direct mechanism $(q_\belief,t_\belief):\Types\to\Alloc$. This determines the designer's value function $\DP:\Beliefs\to\reals$. We obtain the designer's optimal Blackwell experiment  by studying the properties of \DP.

Given belief $\belief$, define the agent's and the designer's (expected) payoff at $(q,t,\type)$ as follows:
\begin{align*}
    \pay(q,\type|\belief)&\equiv\sum_{\state\in\States} \belief(\state)\pay\argq,\;\;
    \ddp(q,\type|\belief)\equiv\sum_{\state\in\States} \belief(\state)\ddp\argq.
\end{align*}
Thus, conditional on inducing belief \belief, the designer's problem can be written as follows:
\begin{align}
\DP(\belief)\equiv\max_{(q,t):\Types\to\Alloc}&\int_{\Types}\left[\ddp(q(\type),\type,\belief)+t(\type)\right]\typecdf(d\type)\\
\text{s.t.}&\left\{\begin{array}{ll}(\forall\type\in\Types)&\pay(q(\type),\type|\belief)-t(\type)\geq0\\
(\forall\type,\typeb\in\Types)&\pay(q(\type),\type|\belief)-t(\type)\geq\pay(q(\typeb),\type|\belief)-t(\typeb)
\end{array}\right..\nonumber
\end{align}
Our assumptions imply that $\pay(\cdot|\belief)$ is supermodular in $(q,\type)$. It follows that the designer can only choose among those $q:\Types\to[0,\maxq]$ that are (weakly) increasing in \type. Let \Qup\ denote the set of all such $q(\cdot)$. Furthermore, at the optimum, the participation constraint of $\type=\mint$ binds. 

Define the virtual surplus at $\argq$ as follows:
\begin{align*}
 \vs(\argq;\typecdf)= \ddp\argq+\pay\argq-\pay_2\argq\frac{1-\typecdf(\type)}{\typed(\type)},
\end{align*}
where $\pay_2$ is the derivative of \pay\ against its second coordinate; the equi-Lipschitz assumption implies it exists almost everywhere. Then, conditional on inducing belief \belief, the designer's payoff can be written as follows:
\begin{align}\label{eq:dp-up}
\DP(\belief)=\max_{q\in\Qup}\int_{\Types}\expect_{\belief}\left[ \vs((q(\type),\type,\state);\typecdf)\right]\typecdf(d\type).
\end{align}
Note the objective is linear in \belief\ and the constraint set is independent of \belief. We conclude that \DP\ is convex, as it is the maximum of linear functionals in \belief. It follows that full disclosure is an optimal experiment for the designer. Equivalently, an optimal calibrated mechanism exists in which the designer chooses the mechanism $\mechanism_\full^D$, where for all $\state\in\States$, $\mechanism_{\full}^D(\cdot,\state,\cdot):\Types\times[0,1]\to\Alloc$ is the optimal deterministic mechanism when it is common knowledge that the state is \state.

\autoref{prop:transparent} summarizes the above discussion:
\begin{proposition}[State-Independent Type Ranking]\label{prop:transparent}
In a single-dimensional screening problem with state-independent type ranking, the designer can do no better than choosing $\mechanism_\full^D$ among deterministic mechanisms.
\end{proposition}

By \autoref{prop:transparent}, in screening problems with state-independent ranking of types across states, the calibration constraint makes any pooling of mechanisms across states unprofitable.\footnote{The calibrated mechanism which state-by-state implements the optimal direct mechanism under common knowledge of the state may reveal less than full information about the state, e.g., because at \state\ and \stateb\ the same direct mechanism is optimal. The point is that pooling those states does not weaken the incentive constraints of the agent, so it is as if the designer were forced to reveal the state.} Remarkably, this result holds for any designer objective, such as profit, revenue, or efficiency. It also requires no regularity assumptions on the type distribution, as we do not obtain the result by looking at the relaxed problem. Instead, our argument relies on the restriction to deterministic mechanisms (conditional on the induced belief), which in turn delivers that the set of implementable allocations does not depend on the induced belief. 
 \autoref{remark:deterministic} discusses conditions under which (i) the restriction to deterministic mechanisms is without loss of optimality, and (ii) the set of implementable allocations does not depend on the induced belief, even when randomized mechanisms are allowed. Readers interested in the case of state-dependent ranking can skip this remark with little loss of continuity.

\begin{remark}[\autoref{prop:transparent} without deterministic mechanisms]\label{remark:deterministic}
Under our assumptions, deterministic mechanisms are without loss of optimality if the agent's payoff is linear in $q$ and the designer's payoff is concave in $q$. (See \autoref{sec:myerson} for yet another condition.) However, the driving force behind \autoref{prop:transparent} is that the designer's constraint set does not depend on the induced belief. The state-by-state supermodularity assumption and the restriction to deterministic mechanisms is \emph{one} way to ensure this is the case. We now discuss two other cases in which the designer's constraint set does not depend on the induced belief and thus \autoref{prop:transparent} holds for the optimal (not necessarily deterministic) calibrated mechanism. %In \autoref{sec:myerson}, we provide further conditions under which the optimal calibrated mechanism is deterministic.

First, suppose the agent's payoff is linear in $q$, so that $\pay\argq=q\val(\type,\state)$, where $\val(\cdot,\state)$ is increasing for all \state. Then, the set of implementable lotteries over $q$ when the belief is \belief\ is given by:
\begin{align*}
Q_{\uparrow,\mathrm{random}}=\left\{\lotq:\Types\to\Delta([0,\maxq]): \mathbb{E}_{\lotq(\type)}\left[q\right] \text{is increasing in \type}\right\}.
\end{align*}
In this case, we obtain that the designer cannot do any better than choosing $\mechanism_{\full}$, which is the mechanism that implements in each state  \state\ the optimal mechanism under common knowledge that the state is \state. This relates to the results in  \cite{szabadi2018essays} and \cite{yamashita2018optimal}, who study the optimal mechanism design preceded by public information disclosure. Both papers consider settings in which the agent's payoff is linear in $q$ and obtain that full disclosure is optimal when the ranking of types is independent of the state.\footnote{Restricting attention to deterministic mechanisms, \cite{ottaviani2001value} obtain the optimality of full disclosure without such a linearity assumption. Their model, however, is different from ours and that of the aforementioned papers: the agent's type is not payoff relevant and the agent's type and the state are affiliated. Thus, while related in spirit, \autoref{prop:transparent} is distinct from their result.}

Second, suppose the agent's payoff has the form
\begin{align*}
\pay\argq&=b(\type)c(\state)\val(q)+k_1(q,\state)+k_2(\type,\state),
\end{align*}
where $b$ is increasing in \type\ and $c(\cdot)$ does not change sign on \States. Under this assumption, $\pay(q,\type|\belief)$ satisfies monotonic expectational differences for all $\belief\in\Beliefs$ (see, e.g., \citealp{kartik2024single}). Consequently, one can define a linear order $\succeq$ over $\Delta([0,\maxq])$ as follows: $\lotq\succeq\lotq^\prime$ if $\pay(\lotq,\type|\belief)-\pay(\lotq^\prime,\type|\belief)$ is increasing in $\type$, where $\pay(\lotq,\type|\belief)$ is the linear extension of $\pay(\cdot,\type|\belief)$ to $\Delta([0,\maxq])$. This linear order implies the ranking of types is state independent. Indeed, the analog of $Q_{\uparrow,\mathrm{random}}$ is the set of all $\lotq:\Types\to\Delta([0,\maxq])$ such that $\type\geq\typeb$ implies $\lotq(\type)\succeq\lotq(\typeb)$, which is again independent of the induced belief.
\end{remark}

\paragraph{State-dependent type ranking} \autoref{example:horizontal} illustrates that when the ranking of types is not uniform across states, full transparency may not be optimal.\footnote{See \cite{szabadi2018essays} for a similar observation.} We now provide a more systematic analysis of this phenomenon, using the previous results. To provide the starkest contrast with \autoref{prop:transparent}, we consider a setting that shares a key feature of \autoref{example:horizontal}: we can partition \Beliefs\ into two regions such that within each region the ranking of types\textemdash as determined by $\pay(q,\type|\belief)$\textemdash is the same, but it differs across regions.

Concretely, suppose that $\States=\{\state_1,\state_2\}$. Furthermore, assume 
\begin{align*}
\pay\argq=\left\{\begin{array}{ll}q\type&\text{if } \state=\state_2\\
q(c-b\type)&\text{otherwise},
\end{array}\right.,
\end{align*}
where $c\in\reals$, and $b>0$.  Identify beliefs with the probability that the state is $\state_2$ and define
\begin{align*}
\beliefh=\frac{b}{1+b}.
\end{align*}
For $\belief<\beliefh$, we have that $\pay(q,\type|\belief)$ is decreasing in \type, whereas if $\belief>\beliefh$, then $\pay(q,\type|\belief)$ is increasing in \type. 

Consider now the designer's optimal payoff $\DP:\Beliefs\to\reals$ as a function of the different beliefs he may induce. When $\belief\geq\beliefh$, the designer's payoff 
 can be obtained by solving the program in \autoref{eq:dp-up} as before. Instead, when $\belief<\beliefh$, the designer's payoff can be obtained by solving a problem analogous to \autoref{eq:dp-up}, but where the space of implementable allocations is the set of decreasing $q$, \Qdown, and the participation constraint of \maxt\ binds. It follows that $\DP$ is convex on $[0,\beliefh)$ and $(\beliefh,1]$. Thus, the support of designer's optimal experiment is included in $\{0,\beliefh,1\}$. 

\begin{proposition}[State-dependent type ranking]\label{prop:reversal}
Suppose the agent's payoff satisfies the assumptions above.  If $(1-\beliefh)\DP(0)+\beliefh\DP(1)\geq\DP(\beliefh)$, full transparency is optimal. Otherwise, full transparency is not optimal: if $\prior<\beliefh$, it is optimal to  split \prior\ to $0$ and $\hat{\belief}$; if $\prior>\beliefh$, it is optimal to  split \prior\ to $\beliefh$ and $1$. In particular, if $\DP(\beliefh)>\max\{\DP(0),\DP(1)\}$, then full transparency is not optimal.
\end{proposition}

By \autoref{prop:reversal}, whether full transparency is optimal depends on the designer and agent's payoffs and the type distribution, but only through their impact on the value the function $\DP$ takes at points $\{0,\beliefh,1\}$. At \beliefh, the agent earns no rents\textemdash as $\pay(\cdot|\beliefh)$ is constant across types\textemdash which pushes against full transparency. At the same time, efficiency may dictate the designer to condition the allocation rule on the state, which favors information disclosure. The piecewise convexity of \DP\ implies that if $\DP(\beliefh)$ dominates \DP\ at the extreme beliefs, the rent extraction motive dominates and the designer does not engage in full disclosure.

\subsection{Comparison with Myersonian Mechanism Design}\label{sec:myerson}

We now compare optimal calibrated mechanism design and the Myersonian benchmark. In the Myersonian benchmark, the designer is not concerned with the information the mechanism reveals about the state, and hence provides a natural upper bound on the designer's payoffs in calibrated mechanism design. The gap between the designer's optimal payoff across both benchmarks quantifies the loss from the calibration constraint. If no gap exists, the calibration constraint is non-binding and an optimal calibrated mechanism can be found solving the Myersonian benchmark. Instead, if a gap exists, the optimal mechanism in the Myersonian benchmark reveals information about the state in a way that it fails to be incentive compatible or individually rational under calibration.

\paragraph{Myersonian benchmark} In the Myersonian benchmark, the designer chooses a direct mechanism $(\lotq,t):\Types\times\States\to\Delta([0,\maxq])\times\reals$ subject to incentive and participation constraints that must hold on average across states under the prior \prior.\footnote{Because payoffs are quasilinear, considering mechanisms that do not randomize on transfers is without loss of generality.} Formally, 
\begin{align}\label{eq:myerson}\tag{OPT$_\myerson$}
\DP_{\myerson} \equiv &\max_{(q,t):\Types\times\States\to[0,\maxq]\times\reals} \int_\Types\mathbb{E}_{\prior}\left[\ddp(\lotq(\type,\state),\type,\state)+t(\type,\state)\right] \typecdf(d\type)\\
    \text{s.t.}&\left\{\begin{array}{ll}(\forall\type\in\Types)\expect_{\prior}\left[u(\lotq(\type,\state),\type,\state)-t(\type,\state)\right]\geq0\\
    (\forall\type,\typeb\in\Types)\expect_{\prior}\left[\pay(\lotq(\type,\state),\type,\state)-t(\type,\state)\right]\geq \expect_{\prior}\left[\pay(\lotq(\typeb,\state),\type,\state)-t(\typeb,\state)\right]\end{array}\right.,\nonumber
\end{align}
where $\ddp(\lotq,\type,\state)$ and $\pay(\lotq,\type,\state)$ are the linear extensions of $\ddp(\cdot,\type,\state)$ and $\pay(\cdot,\type,\state)$, respectively.

Program \ref{eq:myerson} is a mechanism design problem with a multidimensional allocation, corresponding to assigning (a distribution over) $q$ in each state. As a result, the distinction between the Myersonian benchmark and optimal calibrated design shows in the monotonicity requirements the allocation $\lotq(\type,\state)$ must satisfy for a transfer $t:\Types\times\States\to\reals$ to exist that implements $\lotq(\type,\state)$. Indeed, implementability of $\lotq:\Types\times\States\to\Delta([0,\maxq])$ is equivalent to \emph{integral monotonicity} \citep{rochet1987necessary,pavan2014dynamic}:\footnote{The equi-Lipschitz condition on \pay\ ensures we can take the derivative inside the integral.}
\begin{align}\label{eq:integral-monotonicity}\tag{IM}
(\forall\type,\typeb\in\Types)\int_{\typeb}^{\type}\int_\States\left[\pay_2(\lotq(s,\state),s,\state)-\pay_2(\lotq(\typeb,\state),s,\state)\right]d\prior ds\geq0,
\end{align}
where recall $\pay_2$ is the derivative of \pay\ in its second coordinate. %Let \Qmy\ denote the set of all $\lotq(\cdot)$ that satisfy integral monotonicity. 

\paragraph{Comparison with calibrated mechanism design} To facilitate the comparison with \autoref{prop:transparent}, we focus on deterministic mechanisms $(q,t):\Types\times\States\to[0,\maxq]\times\reals$. Remarkably, even if the agent's payoff net of transfers, $\pay\argq$, is supermodular in $(q,\type)$ for all $\state\in\States$, the characterization of the set of implementable $q(\cdot)$ cannot be simplified beyond integral monotonicity without further assumptions. Because integral monotonicity is a global, implicitly defined constraint, verifying implementability and computing the optimal mechanism is more computationally involved in the Myersonian benchmark than in calibrated mechanism design. Indeed, \autoref{prop:transparent} implies the optimal deterministic calibrated mechanism coincides with the state-by-state optimal deterministic mechanism under this assumptions. In other words, the optimal deterministic calibrated mechanism can be obtained by selecting allocations $q(\cdot)$ that satisfy 
\begin{align*}
\Qcal=\left\{q:\Types\times\States\to[0,\maxq]: (\forall\state\in\States) q(\cdot,\state)\text{ is increasing}\right\}.
\end{align*}
That is, the allocation in the optimal calibrated mechanism must satisfy monotonicity state-by-state. Instead, the optimal deterministic Myersonian mechanism can be obtained by selecting allocations $q(\cdot)$ that satisfy \autoref{eq:integral-monotonicity}, which we denote by \Qmy.

When $\pay\argq$ is supermodular in $(q,\type)$ for all $\state\in\States$, the above discussion implies the designer's optimal payoff in the Myersonian and calibration settings can be written as follows:
\begin{align}\label{eq:comparison}
\DP_\myerson^D&=\max_{q\in\Qmy}\int_{\Types}\mathbb{E}_{\prior}\left[\vs(q(\type,\state),\type,\state;\typecdf)\right]\typecdf(d\type),\\
\DP_\cali^D&=\max_{q\in\Qcal}\int_{\Types}\mathbb{E}_{\prior}\left[\vs(q(\type,\state),\type,\state;\typecdf)\right]\typecdf(d\type),\nonumber
\end{align}
where the superscript $D$ in the objective is a reminder that we restrict attention to mechanisms that are deterministic conditional on the state, or the induced belief.

By reducing the comparison across settings to monotonicity requirements on the space of allocations, the above expressions provide us with an immediate way of comparing the designer's payoffs across settings. In particular, when the optimal Myersonian mechanism satisfies the state-by-state monotonicity constraints, we have that the calibration constraint entails no loss to the designer. We record this observation for future use:
\begin{observation}\label{observation:myerson}
Suppose $\pay\argq$ is supermodular in $(q,\type)$ for all $\state\in\States$. Then, if the allocation rule in the Myersonian benchmark satisfies monotonicity state-by-state, $\DP_\cali^D=\DP_\myerson^D$.
\end{observation}

Two natural questions are under what conditions the solution to \ref{eq:myerson} is deterministic and satisfies state-by-state monotonicity. We answer them simultaneously by studying the relaxed program. Inspection of \autoref{eq:comparison} reveals that if the virtual surplus is supermodular in $(q,\type)$ for every $\state$, then the solution $q_{\mathrm{rel}}$ to the relaxed problem
\begin{align}\label{eq:rel}\DP_{\mathrm{rel}}=\max_{q:\Types\times\States\to[0,\maxq]}\int_\Types\mathbb{E}_{\prior}\left[\vs(q(\type,\state),\type,\state;\typecdf)\right]\typecdf(d\type),\end{align}
satisfies monotonicity state-by-state by Topkis' theorem. Moreover, a stochastic mechanism is equivalent to a deterministic mechanism which depends on the random reports of a fictitious agent \citep{pavan2014dynamic}. The virtual surplus in this fictitious setting coincides with that in the integrand on the right-hand side of \autoref{eq:rel}\textemdash the type reports of the fictitious agent are payoff irrelevant\textemdash and is maximized by $q_{\mathrm{rel}}$.
\begin{proposition}[Sufficient condition  for no gap]\label{prop:no-gap}
Suppose the virtual surplus $\vs(\argq;\typecdf)$ is supermodular in $(q,\type)$ for all $\state$. Then,  the designer's payoffs under the optimal Myersonian and calibrated mechanisms coincide.   
\end{proposition} 

% \begin{corollary}[No gap in regular screening problems]
%   Suppose
% \begin{enumerate}
% \item The agent's payoff net of transfers satisfies $\pay\argq=\type\tilde\pay(q,\state)$, where $\tilde\pay(\cdot,\state)$ is increasing for all $\state\in\States$,
% \item The agent's virtual values, $\type-(1-\typecdf(\type))/\typed(\type)$, are increasing,
% %
% \item The designer's payoff net of transfers \ddp\ is supermodular in $(q,\type)$ for all $\state\in\States$.
% \end{enumerate}
% Then, the designer's payoffs under the optimal Myersonian and calibrated mechanisms coincide.
% \end{corollary}
%

By contrast to \autoref{prop:transparent}, \autoref{prop:no-gap} relies on assumptions on the type distribution and the designer's payoff. As \autoref{example:horizontal} illustrates, the supermodularity of the virtual surplus can fail when the type distribution is not regular, creating a gap between the designer's payoff at the optimal Myersonian and calibrated mechanisms. \autoref{example:ddp-nsm} illustrates such a gap can also arise when the designer's payoff is not supermodular:
% the supermodularity of the virtual surplus can fail when the designer's payoff is not supermodular, which can create a gap between the designer's payoff at the optimal Myersonian and calibrated mechanisms. We illustrate with an example:
\begin{example}[Payoff gap when \ddp\ is not supermodular]\label{example:ddp-nsm} Suppose states are binary, $\States=\{\state_L,\state_H\}=\{1,3\}$, and equally likely. Suppose types are uniformly distributed, $\type\sim U[0,1]$. Finally, let $q\in[0,1]$ denote the probability the seller's good is allocated. Payoffs are given by:
\begin{align*}
\pay\argq&=q\type\state\\
\ddp\argq&=2(1-2\type)q.
\end{align*}
Note that \ddp\ is increasing in $q$ when $\type<1/2$ and decreasing in $q$ when $\type>1/2$.\footnote{The designer can be viewed as an online advertising platform and the agent as an advertiser. State $\state$ represents the click-through rate of an ad slot, and higher $\type$ corresponds to a larger advertiser willing to pay more for exposure. The designer's payoff captures both the value created by advertising and the disutility from showing ads of large advertisers, e.g., due to user brand fatigue.}
In this case, the  virtual surplus evaluated at different states is:
\begin{align*}
\vs(\argq;\typecdf)=\left\{\begin{array}{ll}(1-2\type)q&\text{ if }\state=\state_L\\
(2\type-1)q&\text{otherwise}\end{array}\right. .
\end{align*}
In the Myersonian benchmark, implementable allocations are elements of \Qmy, which in this case is equivalent to requiring that $\mathbb{E}_{\prior}\left[q(\cdot,\state)\state\right]$ is increasing. The optimal Myersonian allocation obtains from pointwise maximizing the virtual surplus, and is given by:
\begin{align*}
q_\myerson(\type,\state)=\left\{ \begin{array}{ll}1 &\text{ if $\state=\state_L$ and $\type<1/2$}\\
1 & \text{ if $ \state=\state_H$ and $\type>1/2$}\\
0&\text{ otherwise}\end{array}\right..
\end{align*}
The designer's payoff under the Myersonian mechanism is $1/4$.

By \autoref{prop:transparent}, $q_\myerson$ cannot be implemented by a calibrated mechanism as it is not increasing state-by-state. Intuitively, when  $\state=\state_L$, types above $1/2$ would learn from the calibrated information structure that they do not obtain the good, whereas types below $1/2$ do, and would misreport their types.

Instead, in the optimal calibrated mechanism, the designer sets $q_\cali(\type,\state_H)=\mathbbm{1}[\type\geq 1/2]$ and sets $q(\type,\state_L)$ to be constant in $\type$. The designer's payoff under calibration is $\DP_\cali=1/8<\DP_\myerson$.

\end{example}

\subsection{Optimal Calibrated Auction}\label{sec:optimal-auction}
In this section, we consider a multiple agent application and study the design of the optimal calibrated auction. \autoref{theorem:cmd-two-stage-multi} implies the optimal calibrated auction induces a generalized two-stage mechanism, and hence the optimal generalized two-stage mechanism provides an upper bound on the designer's optimal payoff under calibration. However, computing the optimal generalized two-stage mechanism is complicated because (i) no tractable characterization of joint distributions over posterior beliefs is available, and (ii) the allocation rule may condition on the state and not only the agents' beliefs. For that reason, our analysis below relies on \autoref{observation:myerson}: We show the optimal Myersonian auction can be implemented by fully revealing the state, and hence, remains incentive compatible and individually rational when the agents have access to the calibrated information structure. Below, we first specialize our multi-agent model and notation to the auction application and then link our assumptions to online advertising.

Suppose there is a single good for sale and the state is  multidimensional, $\state=(\state_i,\state_{0i})_{i\in\setplayers}\in\reals_+^{2N}$, and distributed according to prior distribution \prior.  Suppose that for all $i\in\setplayers$, $\Typesi=[0,1]$, with $\typei\sim\typecdfi$ with full-support density $\typedi$. That is, we are assuming agents' types are independent of the state, and hence, independent across each other. Denote by $q_i\in[0,1]$ the probability agent $i$ is allocated the good, and note that feasibility implies that $0\leq\sum_{i=1}^Nq_i\leq 1$.

We assume the agents' and the designer's utilities are quasilinear in transfers. Agent $i$'s payoff net of transfers is $\payi(q_i,\typei,\state)=q_i (\state_i\typei+\state_{0i})$. Thus, state components $\state_i$ capture the value responsiveness  to agent's private information, whereas state components $\state_{0i}$ capture the overall shift.  The state components can be  correlated (and asymmetric) across agents,  allowing for interdependent values.  
The designer's payoff net of transfers is $\ddp\argq=\sum_i q_i \ddp_i(\type,\state)$ for some functions $(\ddp_i)_{i\in[N]}$.  Below, we study the  designer-optimal calibrated mechanism.

To fix ideas, consider the following mapping to an online advertising environment. The designer is an advertising platform, and the good is an advertising slot on a given webpage targeted to a selected category of users in a given week. Agents are firms that wish to display their ads, and their private types represent the expected revenue from a click on their ad. State components $\state_i$ could capture individual click-through rates or match values, while state components $\state_{0i}$ could capture individual display values, that is, the expected revenue from an ad being displayed irrespective of whether it is clicked (for instance, due to brand-building effects). The state is observed through proprietary data available to the platform and can be used in the design of the auction. The platform values the resulting revenue but may also have additional efficiency considerations, summarized by $\ddp_i$.

As anticipated, we characterize the optimal calibrated mechanism by showing that it coincides with the Myersonian optimal one. To this end, consider the Myersonian problem, in which the designer chooses $(q(\type,\state),t(\type,\state))\in[0,1]^N\times\reals^N$. Because agent $i$'s payoff is linear in $\typei$, arguments analogous to those in \autoref{sec:myerson} imply a feasible $q(\type,\state)$ is implementable if and only if for all $i$, $\expect_{F_{-i},\belief_0} [q_i(\typei,\typemi,\state)\state_i]$ is increasing in $\typei$. In a slight abuse of notation, denote by $\Qmy$ the set of all such functions and define the virtual surplus as 
\begin{align}
 \vs(\argq;\typecdf)&=
   \sum_{i=1}^N  q_i(\type,\state) \left(\ddp_i(\type,\state)+\left(\typei-\frac{1-\typecdfi(\typei)}{\typedi(\typei)}\right)\state_i+\state_{0i}\right).
\end{align} 
% where in a slight abuse of notation we denote by $\virtual_i(\typei,\typecdfi)\triangleq\typei-(1-\typecdfi(\typei))/\typedi(\typei)$ agent $i$'s virtual type. 
Standard arguments imply the individual  rationality constraint of $\typei=0$ binds for all $i$, and an optimal mechanism solves 
\begin{align*}
\DP_\myerson=\max_{q\in\Qmy}\int_{[0,1]^N}  \expect_{\belief_0}[\vs(q(\type,\state),\type,\state;F)]f(\type)d\type.
\end{align*}

\begin{proposition}[No gap in regular auctions]\label{prop:no-gap-auction}
Suppose that (i) for all $i\in\{1,\dots,N\}$, \typecdfi\ is Myerson regular, and for all $i,j,\type,$ and $\state$, $\ddp_{i\typei}(\type,\state)\geq0$, $\ddp_{i\typei}(\type,\state)\geq \ddp_{j\typei}(\type,\state)$.\footnote{In the statement, $\ddp_{k\typei}$ denotes the derivative of $\ddp_k$ with respect to \typei, for $k\in\{1,\dots,N\}$.} Then, $\DP_\cali=\DP_\myerson$.     
\end{proposition} 
The proof of \autoref{prop:no-gap-auction} in \autoref{appendix:app} shows that under our assumptions the optimal Myersonian mechanism can be obtained by solving the relaxed program. Importantly, the assumption that $\ddp_{i\typei}(\type,\state)\geq \ddp_{j\typei}(\type,\state)$ ensures that an increase in agent $i$'s type increases the designer's payoff of giving the object to agent $i$ by more than the value of giving it to other agents. This, in turn, ensures agent $i$'s allocation probability is increasing in her type.

Viewed through the lens of the online advertising example, Proposition \ref{prop:no-gap-auction} implies that in regular environments, while the advertising platform  benefits from having the data on click-through rates and display values, it does not benefit from the informational advantage over bidders that such data entails. Its objective is maximized by making the click-through rates and display values readily available to bidders and running optimal auctions in all instances.

\section{Microfoundation}\label{sec:microf}
In this section, we provide a microfoundation for calibrated mechanism design by analyzing the outcome distributions that can arise when an agent repeatedly engages with the same mechanism (\autoref{sec:repeated}) and contrast this to what can be implemented when the designer can offer the agent a \emph{fully} dynamic mechanism (\autoref{sec:dynamic}). To keep the presentation simple, we present the results with minimal notation, and refer the reader to \autoref{appendix:microf} for details. 

Throughout, we consider the case of a single agent, whose type (i) is redrawn each period from the same distribution and (ii) is independent of the state.  The reason for (i) is as follows. When the designer offers the agent a fully dynamic mechanism, the revelation principle implies that it is without loss of generality for the designer to ask the agent for type reports. Moreover, logic similar to that in \cite{myerson1986multistage} implies that the designer only elicits one type report when the agent's type is persistent, and hence, the agent has no possibility of experimenting with the mechanism. Hence, to put repeated and dynamic mechanisms on a more similar footing, assuming the agent's type is redrawn each period is necessary. However, when the agent's type is repeatedly drawn from a distribution that depends on the state, the agent learns about the state both through her own type and her allocations in the mechanism.\footnote{To be sure, \autoref{theorem:repeated} extends to the case in which the agent's type is fully persistent and correlated with the state.} Thus, we assume (ii) so that the agent learns about the state only through her interaction with the mechanism.  Lastly, we consider the single-agent case as extending the results in this section to multiple agents requires addressing subtle issues in strategic experimentation, which we plan to pursue in future work.

\subsection{Repeated Interactions with a Mechanism}\label{sec:repeated}
We consider first the case in which the agent interacts repeatedly with the same mechanism \rmech\ in each period of an infinite horizon interaction. In line with \autoref{sec:model}, a repeated mechanism is a mapping
\begin{align*}
\rmech:\mssgs\times\States\times\fcor\to\Delta(\Alloc),
\end{align*}
where \mssgs\ is a finite set of messages and \fcor\ is a finite set endowed with some measure, denoted \cord. The results in \autoref{sec:main} imply that assuming \fcor\ is finite is without loss of generality and it simplifies the proofs. In contrast to \autoref{sec:model}, we allow the mechanism to have an arbitrary message space. The reason is that we cannot invoke the revelation principle when the designer offers the same mechanism repeatedly: unless the agent's best response is the same across periods, the composition of the mechanism with the agent's reporting strategy yields a time-dependent, direct mechanism. To avoid keeping track of participation and reporting strategies separately in what follows, we assume a message $\mssg_\emptyset\in\mssgs$ exists such that for all $(\state,\cor)\in\States\times\fcor$, $\rmech(\mssg_\emptyset,\state,\cor)=\delta_{\oo}$.

\paragraph{Timing} Given \rmech, the agent faces the following extensive form. Nature draws $(\state,\cor)$ once at the beginning, unobserved to the agent. In each period, nature first draws the agent's type, which the agent observes. The agent then sends a message $\mssg$ into the mechanism. The mechanism then draws the allocation from $\rmech(\cdot|\mssg,\state,\cor)$, which the agent observes.

Given the mechanism \rmech\ and the extensive form game it induces, the agent's strategy specifies for each period $t$ and each period-$t$ type $\type\in\Types$, a distribution over \mssgs, as a function of the agent's past observations, which include her past types, messages, and allocations. Importantly, we assume the agent does not observe her payoffs to focus on the agent learning through the mechanism.

We assume the agent is infinitely patient, that is, she has limit-of-means preferences. Her average payoff through period $T$ when the realization is $(\state,\cor)$ and the type-message-allocation sequence is $(\type_t,\mssg_t,\alloc_t)_{t=1}^T$ is given by:
\begin{align*}
U_T((\type_t,\mssg_t,\alloc_t)_{t=1}^T,\state,\cor)=\frac{1}{T}\sum_{t=1}^Tu(\alloc_t,\type_t,\state).
\end{align*}
A strategy \strat\ is a best response for the agent if for all alternative strategies \stratb, we have that
\begin{align}\label{eq:br}
\lim\inf_{T\to\infty}\mathbb{E}_\strat\left[U_T\right]\geq\lim\sup_{T\to\infty}\mathbb{E}_{\stratb}\left[U_T\right],
\end{align}
where $\mathbb{E}_\strat$ is the expectation relative to the measure induced over the terminal histories by the prior on \States, the distribution on \fcor, the agent's type distribution \typed, the mechanism \rmech, and the agent's reporting strategy \strat.\footnote{The Ionescu-Tulcea extension theorem implies this measure is always well-defined for any mechanism and any agent's strategy. See \autoref{appendix:repeated} for details.}

\paragraph{Implementation} Our notion of implementation is based on the induced \emph{occupation measure} on \Arg, that is, the (limit) expected frequency of tuples \aarg\ when the agent best responds to the mechanism.  For this reason, we restrict attention to mechanisms \rmech\ for which (i) a best-response \strat\ exists, and (ii) its induced occupation measure $\occup_\strat$ over $\Alloc\times\Types\times\mssgs\times\States\times\fcor$ exists, defined as follows\footnote{Throughout this section, limits of measures should always be understood in the weak$*$ sense. }
\begin{align}\label{eq:induced-occup}
\occup_\strat(\alloc,\type,\mssg,\state,\cor)=\lim_{T\to\infty}\frac{1}{T}\mathbb{E}_\strat\left[\sum_{t=1}^T\mathbbm{1}\left[(\alloc_t,\type_t,\mssg_t,\state^\prime,\cor^\prime)=(\alloc,\type,\mssg,\state,\cor)\right]\right]=\lim_{T\to\infty}\occup_\strat^T(\alloc,\type,\mssg,\state,\cor),
\end{align}
where the last identity defines $\occup_\strat$ as the limit of the up to period $T$ occupation measures $\occup_\strat^T$, which are always well-defined. 

Under our definition of best response, which is the same as in \cite{hart1985nonzero}, existence of a best response implies the agent's payoff at the best-response strategy is well-defined.\footnote{Note that in mechanism design one always focuses on mechanisms that have well-defined best responses in single-agent settings, and equilibria in multi-agent ones.} Even if the occupation measure in \autoref{eq:induced-occup} is enough to calculate the agent's payoffs, that the agent's payoffs are well-defined does not mean the occupation measure is well-defined. Because outcome distributions\textemdash and not payoffs\textemdash are usually the focus of mechanism design, we require that both the mechanism has a best response \emph{and} it induces a well-defined occupation measure. 

\begin{definition}[Implementation]
Outcome distribution $\outcome\in\Delta(\Arg)$ can be implemented by a repeated mechanism if a mechanism \rmech\ and a best-response strategy \strat\ exist such that 
\begin{align*}
\outcome\aarg=\sum_{\cor\in\fcor,\mssg\in\mssgs}\occup_\strat(\alloc,\type,\mssg,\state,\cor).
\end{align*}
\end{definition}

We are now ready to state the main result of this section. \autoref{theorem:repeated} shows that the outcome distributions implemented by repeated mechanisms can be implemented by two-stage mechanisms, and hence by calibrated mechanisms:
\begin{theorem}[Microfoundation of Calibrated Mechanism Design]\label{theorem:repeated}
Outcome distribution $\outcome\in\Delta\left(\Arg\right)$ is implementable by a repeated mechanism if and only if \outcome\ can be implemented by an incentive compatible and individually rational two-stage mechanism, that is for all $\aarg\in\Arg$,
\begin{align}\label{eq:ts-impl}
\outcome\aarg=\prior(\state)\typed(\type)\int_{\Beliefs}\allocr(\alloc|\type,\belief)\beliefr(d\belief|\state),
\end{align}
where $\beliefr:\States\to\Delta(\Beliefs)$ is Bayes plausible and $\allocr(\cdot|\cdot,\belief):\Types\to\Delta(\Alloc)$ is incentive compatible and individually rational on the support of $\prior\otimes\beliefr$.
\end{theorem}

The proof of this and all results in this section can be found in \autoref{appendix:microf}.

\autoref{theorem:repeated} provides a microfoundation for calibrated mechanism design. Whenever the designer is concerned with agents learning from the outcome of the mechanism and cares only about the long-run outcome distribution, it is as if he is designing a two-stage mechanism. 

We now provide a proof sketch for \autoref{theorem:repeated}, which is also useful to understand the proof of the result in the next section. For simplicity, let $\estates=\States\times\fcor$ with elements $\estate$. Suppose repeated mechanism \mechanism\ implements \outcome, and let $\occup_\strat\in\Delta(\Argme)$ denote the induced occupation measure. As the analysis so far illustrates, tracking the joint distribution over allocations, types, states, \emph{and} beliefs is important to show that \outcome\ can be implemented via a two-stage mechanism. To this end, we extend the up to period $T$ occupation measures, $\occup_{\strat}^T\in\Delta(\Argme)$, to account for the frequency of beliefs through period $T$. In fact, we define two sequences of extended occupation measures over $\Argbe$: the first, $\eoccup_\strat^{T,1}$, calculates the frequency of a tuple \argmeb\ by counting the beliefs at the beginning of period $t$ and the second, $\eoccup_\strat^{T,2}$, by counting the beliefs at the end of period $t$. Whereas the martingale property of beliefs implies these two sequences have the same (subsequential) limits, they have different conditional independence properties, which we use to derive the representation of \outcome\ via a two-stage mechanism. Suppose for simplicity that $\bar\occup_\strat^{T,1}$ (and hence, $\bar\occup_\strat^{T,2}$), have limit $\bar\occup_\strat$, though this assumption is not needed for the proof.\footnote{By assumption, the marginal of $\bar\occup_\strat^T$ on \Argme\ converges to $\occup_\strat$. Moreover, we show that the marginal of $\bar\occup_\strat^T$ on \ebeliefs\ also converges (\autoref{prop:belief-occup}). However, this is not enough to ensure the convergence of $\bar\occup_\strat^T$.}  A consequence of the martingale property of beliefs is that only the long-run beliefs of the agent are in the support of $\bar\occup_\strat$.

The proof consists of three steps. In the first step, we show that $\bar\occup_\strat$ admits the following decomposition:
\begin{align*}
\bar\occup_\strat(\{\argme\}\times\tilde\Delta)=\int_{\tilde\Delta}\belief(\estate)\typed(\type)\rep(\mssg|\type,\belief)\allocrb(\alloc|\mssg,\belief)\bsplit(d\belief),
\end{align*}
where (i) $\bsplit\in\Delta(\ebeliefs)$ has mean $\prior\otimes\cord$, where recall $\cord$ is the measure on \fcor, and (ii) $\rep:\Types\times\ebeliefs\to\Delta(\mssgs)$ is a ``Markovian reporting strategy'', and (iii) $\allocrb$ is \emph{almost} the allocation rule in the two-stage mechanism, and hence the prime notation. Moreover, on the support of $\belief$, $\allocrb(\cdot|\cdot,\belief)$ coincides with $\rmech(\cdot|\cdot,\estate)$, implying that $\rmech(\cdot|\cdot,\estate)$ is constant in \estate\ on the support of \belief. This is the step which exploits the different conditional independence properties of $\bar\occup_\strat^{T,1}$ and $\bar\occup_\strat^{T,2}$. We use $\bar\occup_\strat^{T,1}$ to show the conditional independence of types and beliefs\textemdash all agent types in period $t$ have the same belief at the beginning of period $t$\textemdash and $\bar\occup_\strat^{T,2}$ to show the conditional independence of the allocation and the state\textemdash  the belief at the end of period $t$ contains all the information about the state contained in the allocation. 

In the second step, we show that $\rep:\Types\times\ebeliefs\to\Delta(\mssgs)$ is indeed a best response for the agent when her type is \type\ and her belief is \belief. In other words, the support of $\rep(\cdot|\type,\belief)$ is contained in 
\begin{align*}
\arg\max_{\mssg\in\mssgs}\sum_{\estate\in\estates}\belief(\estate)\sum_{\alloc\in\Alloc}\rmech(\alloc|\mssg,\estate)\pay\aarg,
\end{align*}
for beliefs on the support of $\bsplit$. Hence, we can use \rep\ and \rmech\ to define a direct mechanism $\allocr(\cdot|\cdot,\belief):\Types\to\Delta(\Alloc)$ that satisfies the agent's participation and incentive constraints when her belief is \belief. Together, Steps 1 and 2 allow us to obtain the representation of \outcome\ as in \autoref{eq:ts-impl}.\footnote{Whereas the above two-stage mechanism is described in terms of beliefs over $\States\times\fcor$, we show in the appendix how to derive from it a two-stage mechanism in terms of beliefs over $\States$.
%%The Blackwell experiment can be defined from the distribution over posteriors \bsplit\ as follows:
%\begin{align*}
%\beta(\tilde\Delta|\state)=\int_{\tilde\Delta}\frac{\belief(\state)}{\prior(\state)}\bsplit(d\belief).
%\end{align*}
}  

Whereas the above steps are enough to show that \outcome\ is implementable by \emph{some} individually rational and incentive compatible two-stage mechanism, they do not necessarily imply that the distribution over posteriors \bsplit\ is the one induced by the information structure calibrated to \rmech, \cexp. The last step of the proof shows that even if this is not the case, the agent \emph{adequately learns} the information contained in \cexp\ in the sense of \cite{aghion1991optimal}. Indeed, \autoref{lemma:adequate-learning} shows a strategy exists that approximately delivers the payoff from learning \cexp, so that the agent's payoff under \strat\ is at least the payoff she would obtain if she had learned \cexp. Because the payoff from learning \cexp\ is the maximal payoff the agent can possibly attain, we conclude that the payoff under \strat\ is the payoff the agent would attain when facing the calibrated information structure \cexp\ (and best responding to it).

\subsection{Dynamic Mechanisms}\label{sec:dynamic}
In this section, we consider the case in which the designer can offer the agent a \emph{dynamic} mechanism, that is, one that conditions the allocation in each period on the history of past allocations and reports. The analysis herein allows us to describe the limits implied by calibration on the set of implementable outcomes.

\paragraph{Dynamic mechanisms} A dynamic mechanism  $\dmech=(\dmecht)_{t\in\naturals}$ is a sequence of mappings that condition on the state, the history of participation decisions, type reports and allocations, and today's report, and output an allocation. Formally, expand the set of type reports and allocations by a non-participation message and the outside option, which we denote by  $\Types\Alloc_\emptyset=\Types\times\Alloc\cup\{(\emptyset,\oo)\}$.\footnote{This notation allows us to keep the definitions of the histories when the agent participates and does not participate symmetric, and saves us on including the agent's participation strategy in the histories.} For each $t\in\naturals$, define the mechanism in period $t$, $\dmecht:\States\times(\Types\Alloc_\emptyset)^{t-1}\times\Types\to\Delta(\Alloc)$.  Because the designer can flexibly design the mechanism in each period we no longer rely on the randomization device. 

A dynamic mechanism induces an extensive-form game for the agent, in which in each period, the agent decides whether to participate, and conditional on participation what type to report. Whenever the agent chooses not to participate, she obtains her outside option \oo. We denote by \pp\ the agent's participation strategy and by \strat\ the agent's reporting strategy.

\paragraph{Implementation} Our notion of implementation continues to be based on the occupation measure over the set of allocations, types, participation decisions, type reports, and states, induced by the distributions \prior\ and \typed, the mechanism \dmech, and the agent's participation and reporting strategy. However, as we show in \autoref{appendix:dyn-rp}, it is without loss to focus on mechanisms such that (i) participation with probability 1 and truthtelling is a best response for the agent, and (ii) the mechanism implements the outside option with probability 1 in all future periods following a non-participation decision by the agent.\footnote{That is, starting from a dynamic mechanism \dmech\ and a best response strategy $(\pp,\strat)$, one can construct an alternative mechanism $\dmech^\prime$ such that participation and truthtelling are a best response for the agent and preserves the distribution over $(\States\times\Types\times\Alloc)^\infty$ induced by \estrat\ and \dmech.} Thus, we focus on dynamic mechanisms \dmech\ such that (i) a best response exists, and (ii) the occupation measure over \Arg\ is well-defined. 

\paragraph{Incentives in dynamic mechanisms} Dynamic mechanisms allow the designer to condition the agent's allocation on the history of past participation decisions and reports (and allocations), and hence allow the designer to implement outcomes that satisfy weaker notions of truthtelling and participation, which we explain next.

Because the designer can condition the mechanism on the history of past reports, he can compare the frequency of type reports against the type distribution. So long as the agent is telling the truth, the frequency of reports will match the type distribution \typed\ over large blocks of time. In fact, any reporting strategy whose expected frequency of reports matches the type distribution will be indistinguishable from truthtelling. %
\begin{definition}[Undetectable deviations]\label{def:undetectable}
An undetectable deviation is a reporting strategy $\strat:\Types\to\Delta(\Types)$ such that for all $\typeb\in\Types$
\begin{align*}
\sum_{\type\in\Types}\typed(\type)\strat(\typeb|\type)=\typed(\typeb).
\end{align*}
\end{definition}
By tracking the empirical distribution of type reports, the designer can dissuade the agent from employing detectable deviations. Thus, in a dynamic mechanism, the designer should be concerned with only discouraging undetectable deviations. This leads to a weaker notion of incentive compatibility for allocation rules:
\begin{definition}[Unprofitable undetectable deviations]\label{def:rahman-ic}
The allocation rule $\allocr:\Types\times\Beliefs\to\Delta(\Alloc)$ lacks profitable undetectable deviations at belief $\belief\in\Beliefs$ if for all undetectable deviations $\strat$, 
\begin{align*}
\sum_{\type\in\Types}\typed(\type)\sum_{\alloc\in\Alloc}\allocr(\alloc|\type,\belief)\sum_{\state\in\States}\belief(\state)\pay\aarg\geq\sum_{\type\in\Types}\typed(\type)\sum_{\typeb\in\Types}\strat(\typeb|\type)\sum_{\alloc\in\Alloc}\allocr(\alloc|\typeb,\belief)\sum_{\state\in\States}\belief(\state)\pay\aarg.
\end{align*}
A two-stage mechanism \tsmech\ with allocation rule \allocr\ lacks profitable undetectable deviations if $\allocr(\cdot|\cdot,\belief)$ lacks profitable undetectable deviations for all beliefs in the support of the mechanism.
\end{definition}
To illustrate the difference between the lack of profitable undetectable deviations and incentive compatibility, consider the following example from \cite{ball2023quota}. Suppose the agent types are binary, $\{\type_1,\type_2\}$, and equally likely. The set of allocations, $q\in\{0,1\}$, describes whether the agent receives a good. Finally, suppose the agent's payoff is $\pay(q,\type)=q\type$ and $\type_1<\type_2$. Consider the mechanism that allocates the good to $\type_2$: While it is not incentive compatible, it lacks profitable undetectable deviations. The constraint that the deviation must be undetectable implies the gains from $\type_1$ obtaining the good come at the expense of $\type_2$ getting the good.

Consider now the agent's \emph{participation} incentives in the dynamic mechanism: once the agent rejects the mechanism once, the agent obtains her outside option in all continuation histories independent of her participation decision and her types. In other words, whereas the agent can always ensure her outside option by rejecting the mechanism in a given period, she is effectively quitting the mechanism forever for all her types. The following definition introduces the notion of individual rationality satisfied by the mechanism in the long run.
\begin{definition}[Ex ante individual rationality]\label{def:dynamic-ir}
The allocation rule $\allocr:\Types\times\Beliefs\to\Delta(\Alloc)$ is ex ante individually rational at belief $\belief\in\Beliefs$ if   
\begin{align*}
\sum_{\type\in\Types}\typed(\type)\sum_{\alloc\in\Alloc}\allocr(\alloc|\type,\belief)\sum_{\state\in\States}\belief(\state)\pay\aarg\geq\sum_{\type\in\Types}\typed(\type)\sum_{\state\in\States}\belief(\state)\pay(\oo,\type,\state).
\end{align*}
A two-stage mechanism \tsmech\ with allocation rule \allocr\ is ex ante individually rational if $\allocr(\cdot|\cdot,\belief)$ is ex ante individually rational for all beliefs in the support of the mechanism.
\end{definition}

We are now ready to state the main result of this section:
\begin{theorem}[Implementable Outcomes via Dynamic Mechanisms]\label{theorem:dynamic}
A dynamic mechanism exists that implements outcome $\outcome\in\Delta\left(\Arg\right)$ if and only if \outcome\ can be implemented by an ex ante individually rational two-stage mechanism which lacks profitable undetectable deviations. That is, if and only if for all $\aarg\in\Arg$
\begin{align}\label{eq:rahman-outcomes}
\outcome\aarg=\prior(\state)\typed(\type)\int_{\Beliefs}\allocr(\alloc|\type,\belief)\beliefr(d\belief|\state),
\end{align}
where $\beliefr:\States\to\Delta(\Beliefs)$ is Bayes plausible and $\allocr(\cdot|\cdot,\belief):\Types\to\Delta(\Alloc)$ lacks profitable undetectable deviations and is ex ante individually rational on the support of $\prior\otimes\beliefr$.
\end{theorem}
\autoref{theorem:dynamic} characterizes the outcome distributions implementable by dynamic mechanisms as those implemented by two-stage mechanisms that satisfy the incentive constraints: unprofitability of undetectable deviations and ex ante individual rationality. Notably, both notions of incentive constraints apply in the aggregate over the type distribution, which reflects the transient nature of the agent's private information.

Comparing \autoref{theorem:repeated} and \autoref{theorem:dynamic}, we see that dynamic mechanisms allow the designer to weaken the incentive constraints of the agent, but do not allow him to engage in richer\textemdash i.e., type-dependent\textemdash disclosures. Despite dynamic mechanisms implying weaker incentive constraints, we can build on the results of \cite{rochet1987necessary} and \cite{rahman2024detecting} to show that in settings with transferable utility, where $\alloc=(q,t)$, dynamic and repeated mechanisms implement the same set of \emph{physical} allocations $q:\Types\times\States\to\reals$.\footnote{This is easily seen in the example after \autoref{def:rahman-ic}. Note the mechanism that allocates the good to the agent if and only if her type is $\type_2$ satisfies monotonicity. Hence, a transfer scheme exists that implements this allocation rule with transfers.} Indeed, \cite{rahman2024detecting} shows the lack of profitable undetectable deviations is equivalent to cyclical monotonicity in \cite{rochet1987necessary}. Thus, in settings with transferable utility, Theorems \ref{theorem:repeated} and \ref{theorem:dynamic} imply that dynamic mechanisms do not allow the designer to expand on the set of implementable distributions over $(q,\type,\state)$.

The proof of the only if direction is similar to that of \autoref{theorem:repeated}, in that we similarly extend the occupation measure to account for the agent's beliefs and show it satisfies the conditional independence properties implied by a two-stage mechanism. In a dynamic mechanism, however, the agent can ensure the payoff of some, but not all deviations. The latter property is what delivers that the two-stage mechanism must lack profitable undetectable deviations.

The proof of the if direction, instead, harnesses a construction in \cite{margaria2018dynamic}. The proof proceeds in two steps. In the first step, we analyze a fictitious model without state uncertainty in which a designer faces a privately informed agent, so that implementable outcomes are elements of $\Delta(\Alloc\times\Types)$. We show that if $\outcomeb(\alloc,\type)=\typed(\type)\allocrb(\alloc|\type)\in\Delta(\Alloc\times\Types)$ is such that \allocrb\ lacks profitable undetectable deviations and is ex ante individually rational, then a dynamic mechanism exists that implements  $\outcomeb$.\footnote{In a repeated principal-agent game with communication, \cite{meng2021value} shows that the principal can guarantee in the patient limit his complete information payoff subject to the constraint that her actions satisfy the cyclic monotonicity condition in \cite{rochet1987necessary}. We view the results as complementary: We focus on implementable outcome distributions, instead of payoffs, when the agent has limit of the means preferences, which makes our notion of implementation exact.} This is the step that relies on \cite{margaria2018dynamic}. We construct a dynamic mechanism, which can be split into blocks of random length. Each block consists of two phases: a reporting phase and an adjustment phase. In the reporting phase, the mechanism uses the agent's reports to determine the allocation. Instead, in the adjustment phase, the mechanism simulates type reports so that the frequency of type reports matches the type distribution (in expectation) over the length of the block, whenever this is not the case at the end of the reporting phase. These two steps ensure that the expected frequency of type reports and allocations matches $\outcomeb$. We then leverage that $\allocrb(\cdot|\type)$ lacks profitable undetectable deviations to show the agent cannot do better than by telling the truth. Hence, the induced frequency of types and allocations also matches  $\outcomeb$. Moreover, the construction ensures that after any history, truthtelling delivers a continuation payoff equal to the ex ante payoff. Because $\allocrb$ is ex ante individually rational, we conclude the participation constraints are satisfied.

The second step uses the above result and the representation of the outcome distribution via a two-stage mechanism to construct a dynamic mechanism that implements any outcome distribution that satisfies the properties in \autoref{theorem:dynamic}. Indeed, one can construct a dynamic mechanism which uses a finite number of steps to disclose information to the agent via the realized allocations,\footnote{A standard argument implies that if \outcome\ satisfies \autoref{eq:rahman-outcomes}, then a finite support $\beliefrb$ exists such that \outcome\ satisfies \autoref{eq:rahman-outcomes} with \beliefrb\ in place of \beliefr.} and then continues as in the above construction to implement the allocation rule $\allocr(\cdot|\cdot,\belief)$.
\section{Conclusions}\label{sec:conclusions}
Many economic institutions\textemdash online platforms, lenders, regulators\textemdash rely on mechanisms that remain fixed while agents interact with them repeatedly. When the mechanism's operation depends on a state known only to the designer, agents can learn this state from their outcomes, constraining what the mechanism can implement. We introduce calibrated mechanism design, a static solution concept that requires mechanisms to remain incentive compatible given the information they endogenously reveal about the designer's private state through repeated use. In private value environments, the calibration constraint pushes the designer toward full transparency, precluding Crémer-McLean-style schemes under transferable utility. In single agent-settings, calibrated mechanisms are equivalent to two-stage mechanisms. This equivalence yields a practical algorithm for finding optimal calibrated mechanisms, combining tools from information design and mechanism design. We provide a microfoundation by showing calibrated mechanisms
characterize exactly what is implementable when an infinitely patient agent repeatedly interacts with the same mechanism, and study the implications on implementable outcomes of allowing the designer to offer fully dynamic mechanisms.

The most important direction for future work is deepening the analysis of multi-agent settings. On the one hand, understanding when generalized two-stage mechanisms coincide with calibrated mechanisms would enable the study of multi-agent applications, while abstracting from the dynamics of experimentation. On the other hand, extending our microfoundation to the multi-agent case would further ground calibrated mechanism design. More broadly, our framework suggests that any institution whose repeated operation leaks information about its designer's knowledge faces a fundamental  tradeoff between conditioning the mechanism on this information and the information this leaks to participants, and calibrated mechanism design offers a disciplined way to analyze it.

{\singlespacing{
\bibliographystyle{ecta}
\bibliography{calibrated}}}

\appendix
\vspace{0.5cm}

%\begin{center}{\huge{\textbf{Appendix}}}\end{center}
\section*{Mathematical conventions}
Throughout the appendix, we take all sets to be Polish spaces, that is, completely metrizable, separable, topological spaces, and endow them with their Borel $\sigma$-algebra. We endow product spaces with their product $\sigma$-algebra. For a Polish space $X$, we let $\borel_X$ denote its Borel $\sigma$-algebra and $\Delta(X)$ the set of all Borel probability measures on $X$, endowed with the weak$^*$ topology. Thus, $\Delta(X)$ is also a Polish space \citep{aliprantis2006infinite}, and it is compact, whenever $X$ is compact \citep[Theorem 15.11 and Theorem 15.15]{aliprantis2006infinite}.

\paragraph{Notational conventions} If $X$ is a Polish space, $\tilde X$ denotes a measurable subset of $X$, i.e., an element of the Borel $\sigma$-algebra on $X$, and $C_b(X)$ denotes the set of continuous and bounded functions on $X$. Given a measure $\nu\in\Delta(\times_{i=1}^N Y_i)$, we denote by $\nu_{Y_jY_k\dots Y_l}$ the marginal of $\nu$ on $Y_jY_k\dots Y_l$. When one of the $Y_i=\Delta(X_i)$, we write $\Delta$ instead of $Y_i$ in the subscript, when it is unlikely to generate confusion. 

Throughout the appendix, we define different distributions that arise in our proofs. Because we endow product spaces with their product topology and their product Borel $\strat$-algebra, it is enough to define these new measures on the measurable rectangles and we follow this convention throughout.
\paragraph{Disintegration} We rely on the notion of disintegration in many of our proofs \citep[Chapter 10.6]{bogachev2007measure}. We define disintegration in the context of product sets $X\times Y$, as this is the one that shows up in the proof, but it is more general than this. Given a measure $\nu\in\Delta(X\times Y)$, $\lambda:X\times \borel_Y\to[0,1]$ is the disintegration of $\nu$ along $X$ if the following holds
\begin{enumerate}
\item For all $\tilde Y\in\borel_Y$, $x\mapsto\lambda_x(\tilde Y)$ is measurable,
\item For $\nu_X$-almost everywhere $x\in X$, $\tilde Y\mapsto\lambda_x(\tilde Y)$ is a probability measure, and
\item For every bounded measurable function $\testf:X\times Y\to\reals$, $$\int_{X\times Y}\testf(x,y)\nu(d(x,y))=\int_X\int_Y\testf(x,y)\lambda_x(dy)\nu_X(dx).$$
\end{enumerate}
\citet[Theorem 1.23]{kallenberg2017random} ensures that $\{\lambda_x:x\in X\}$ exists and is unique $\nu_X$-almost everywhere.

\section{Omitted proofs from \autoref{sec:model}}
\begin{proof}[Proof of \autoref{theorem:pv}] Suppose the agents' payoffs are state independent and in a slight abuse of notation let $\payi(\alloci,\typei)$ denote agent $i$'s utility. 

The calibrated mechanism design problem is
\begin{align*}
\max_{\mechanism:\Types\times\States\times[0,1]\to\Delta(\Alloc)}\sum_{\state\in\States}\prior(\state)\sum_{\type\in\Types}f(\type|\state)\int_0^1\ddp(\mechanism(\type,\state,\cor),\type,\state)\lebesgue(d\cor),
\end{align*}
subject to the following constraints holding for all $(\state,\cor)\in\States\times[0,1]$, $i\in\setplayers$, $\typei\in\Types$ and $\typebi\in\Types$:
\begin{align*}
&\sum_{\alloci\in\Alloci}\mathbb{E}_{\typedmi(\cdot|\state)}[\sum_{\allocmi\in\Allocmi}\mechanism(\typei,\typemi,\state,\cor)(\alloci,\allocmi)]\payi(\alloci,\typei)\geq\sum_{(\alloci\in\Alloci)}\mathbb{E}_{\typedmi(\cdot|\state)}[\sum_{\allocmi\in\Allocmi}\mechanism(\typebi,\typemi,\state,\cor)(\alloci,\allocmi)]\payi(\alloci,\typei)\\
&\sum_{\alloci\in\Alloci}\mathbb{E}_{\typedmi(\cdot|\state)}[\sum_{\allocmi\in\Allocmi}\mechanism(\typei,\typemi,\state,\cor)(\alloci,\allocmi)]\payi(\alloci,\typei)\geq\payi(\ooi,\typei).
\end{align*}
In other words, for each agent $i$, her interim allocation rule $\cexpi(\state,\cor)$ must be an element of $\csignals_{IC/IR,i}$, where the latter is the set of interim allocation rules $\csignalsi:\Typesi\to\Delta(\Alloci)$ that satisfy the following incentive compatibility and individual rationality constraints:
\begin{align*}
(\forall\typei,\typebi\in\Typesi)\sum_{\alloci\in\Alloci}  \csignali(\alloci|\typei) \payi(\alloci,\typei)&\geq \sum_{\alloci\in\Alloci} \csignali(\alloci|\typebi)\payi(\alloci,\typei)\\
(\forall\typei\in\Typesi)\sum_{\alloci\in\Alloci}  \csignali(\alloci|\typei) \payi(\alloci,\typei)&\geq \payi(\ooi,\typei).
\end{align*}
Because the individual rationality and incentive constraints must hold for each pair $(\state,\cor)$, the designer's problem is separable across variables for different $\state,\cor$: the sets of variables $\mechanism(\cdot,\state,\cor)$ appear in different sets of constraints and the objective function is additively separable across those variables. Consequently, the designer's problem can be solved as a collection of independent problems, one for each $\state,\cor$.
\end{proof}

\color{black}
\section{Omitted proofs from \autoref{sec:main}}\label{appendix:main}
In this section, we present the proofs of \autoref{theorem:cmd-two-stage} and \autoref{theorem:cmd-two-stage-multi}. We proceed as follows: We first prove \autoref{theorem:cmd-two-stage-multi}, as when $N=1$ its proof implies the ``if'' direction of \autoref{theorem:cmd-two-stage}. We then prove the ``only if'' direction of \autoref{theorem:cmd-two-stage}. 

\begin{proof}[Proof of \autoref{theorem:cmd-two-stage-multi}]
We focus on the case in which types and states are independently distributed, and explain how to extend the proof when they are not.

 Let $\outcome\in\Delta(\Arg)$ denote the outcome distribution implemented by an incentive compatible and individually rational calibrated mechanism. We show transition probabilities $\beta:\States\to\Delta(\Beliefs^{N})$, $\bar\allocr:\Types\times\States\times\Beliefs^{N}\to\Delta(\Alloc)$, and $\allocr_i:\Typesi\times\Beliefs\to\Delta(\Alloci)$ for $i\in\{1,\dots,N\}$ exist such that 
\begin{align}\label{eq:decomp-beg}
\outcome\aarg=\prior(\state)\typed(\type)\int_{\Beliefs^{N}}\bar\allocr(\alloc|\type,\state,\belief_1,\dots,\belief_N)\beta(d(\belief_1,\dots,\belief_N)|\state),
\end{align}
and for all $i\in\{1,\dots,N\}$, (i) $\allocr_i$ satisfies \autoref{itm:gts-iar} of \autoref{definition:ts-gral}, and (ii) on the support of $\prior\otimes\beta$, $\allocr_i(\cdot|\cdot,\belief_i)$ is incentive compatible and individually rational when agent $i$ holds belief $\belief_i$.
%The proof follows closely that of the single agent case.

Let $\excdf_{\state,i}:[0,1]\to S_i^*$ denote the mapping $\cor\mapsto\excdf_i(\state,\cdot)$. For $\tilde S_i^*\in\Delta(\Alloci)^{\Typesi}$, define
\begin{align*}
\mathrm{Pr}_i(\{\state\}\times\tilde S_i^*)=\prior(\state)\lebesgue(\excdf_{\state,i}^{-1}(\tilde S_i^*))=\int_{\tilde S_i^*}\belief_i(\state|s_i^*)\bsplit_{\rmech,i}(ds_i^*),
\end{align*}
where the second equality follows from disintegration of $\Pr_i\in\Delta(\States\times S_i^*)$ along $S_i^*$, and corresponds to the definition of Bayes rule for agent $i$. Define the measurable mappings, $T_i:S_i^*\to\Beliefs$ and $T:S^*\to\Beliefs^N$ as follows: $T_i(s_i^*)=\belief_i(\cdot|s_i^*)$ and $T(s^*)=(T_1(s_1^*),\dots, T_N(s_N^*))$.

Define a joint distribution $Q\in\Delta(\Arg\times\Beliefs^N)$ as follows:
\begin{align*}
Q(\{\aarg\}\times\times_{i=1}^N\tilde{\Delta}_i)=\prior(\state)\typed(\type)\int_{\excdf_\state^{-1}(T^{-1}(\times\tilde{\Delta}_i))}\rmech(\alloc|\type,\state,\cor)\lebesgue(d\cor),
\end{align*}
where $\excdf_\state^{-1}(T^{-1}(\times\tilde{\Delta}_i))=\cap_{i=1}^N\{\cor:T_i(\excdf_{\state,i}(\cor))\in\tilde{\Delta}_i\}$.

We note the following properties of $Q$. First, consider its marginal over $\Types\times\States\times\Beliefs^{N}$,
\begin{align*}
Q_{\Types\States\Delta^N}(\{(\type,\state)\}\times\times_{i=1}^N\tilde{\Delta}_i)=\prior(\state)\typed(\type)\lebesgue(\cap_{i=1}^N\{\cor:T_i(\excdf_{\state,i}(\cor))\in\tilde{\Delta}_i\}),
\end{align*}
which implies that the disintegration of $Q_{\Types\States\Delta^N}$ along $\Types\times\States$, $\beta:\Types\times\States\to\Delta(\Beliefs^N)$ does not depend on \type. This automatically implies that $Q$ admits the following disintegration:
\begin{align}\label{eq:decomp-end}
Q(\{\aarg\}\times\times_{i=1}^N\tilde{\Delta}_i)=\prior(\state)\typed(\type)\int_{\times_{i=1}^N\tilde{\Delta}_i}\bar\allocr(\alloc|\type,\state,\belief_1,\dots,\belief_N)\beliefr(d(\belief_1,\dots,\belief_N)|\state),
\end{align}
which, in turn, delivers \autoref{eq:decomp-beg}.
 Moreover, note that the marginal of $\beta$ on the beliefs of agent $i$, $\beliefr_i:\States\to\Delta(\Beliefs)$, satisfies
\begin{align*}
\beliefr_i(\tilde{\Delta}_i|\state)=\lebesgue(\excdf_{\state,i}^{-1}(T_i^{-1}(\tilde{\Delta}_i))).
\end{align*}
Consider now the marginal on $\Alloci\times\Typesi\times\States\times\Beliefs$ of $Q$, $Q_{\Alloci\Typesi\States\Delta_i}$, which satisfies:
\begin{align}\label{eq:a-ind-1}
&Q_{\Alloci\Typesi\States\Delta_i}(\{\argi\}\times\tilde{\Delta}_i)=\sum_{\typemi\in\Typesmi}\sum_{\allocmi\in\Allocmi}Q(\{\aarg\}\times\tilde{\Delta}_i\times\Beliefs^{N-1})=\\
&=\prior(\state)\typed_i(\typei)\int_{\excdf_{\state,i}^{-1}(T_i^{-1}(\tilde{\Delta}_i))}\left(\sum_{\typemi\in\Typesmi}\typedmi(\typemi)\sum_{\allocmi\in\Allocmi}\rmech(\alloci,\allocmi|\typei,\typemi,\state,\cor)\right)\lebesgue(d\cor)\nonumber\\
&=\prior(\state)\typed_i(\typei)\int_{\excdf_{\state,i}^{-1}(T_i^{-1}(\tilde{\Delta}_i))}\excdf_{\state,i}(\cor)(\alloci|\typei)\lebesgue(d\cor)=\prior(\state)\typed_i(\typei)\int_{T_i^{-1}(\tilde{\Delta}_i)}s_i^*(\alloci|\typei)(\lebesgue\circ\excdf_{\state,i}^{-1})(ds_i^*).\nonumber
\end{align}
Lastly, $Q_{\Alloci\Typesi\States\Delta_i}$ admits the following representation via disintegration:
\begin{align}\label{eq:a-ind-2}
&Q_{\Alloci\Typesi\States\Delta_i}(\{(\alloci,\typei,\state)\}\times\tilde{\Delta}_i)=\prior(\state)\typedi(\typei)\int_{\tilde{\Delta}_i}\allocr_i(\alloci|\typei,\state,\beliefi)\beliefr_i(d\beliefi|\state)\nonumber\\&
=\prior(\state)\typedi(\typei)\int_{\tilde{\Delta}_i}\allocr_i(\alloci|\typei,\state,\beliefi)(\lebesgue\circ\excdf_{\state,i}^{-1}\circ T_i^{-1})(d\beliefi)
\end{align}
Together with the uniqueness of disintegration and the sufficiency property of beliefs, Equations \ref{eq:a-ind-1} and \ref{eq:a-ind-2} imply that $\allocr_i$ does not depend on \state. The incentive compatibility and individual rationality of $\allocr_i$ follows from that of the calibrated mechanism. 

Finally, consider the case in which \type\ and \state\ are not independent. Then, the experiment \beliefr\ in \autoref{eq:decomp-end} induces a joint distribution over the beliefs of $N$ fictitious agents whose prior over the state is given by \prior. Agent $i$'s updated beliefs when her type is \typei\ obtain from a transformation of \belief\ \citep{alonso2016bayesian,laclau2017public}.\footnote{Let $\prior(\cdot|\typei)\in\Beliefs$ denote the prior of the agent with type \typei\  and $\belief(\cdot|\signali)$ denote the updated belief of an agent with prior \prior\ upon observing signal \csignali. When the signal is \csignali, the agent with type \typei\ updates her beliefs to:
\begin{align*}
    \belief_i(\cdot|\typei,\csignali)=\frac{\prior(\cdot|\typei)\cdot\frac{\belief(\cdot|\csignali)}{\prior(\cdot)}}{\left\|\prior(\cdot|\typei)\cdot\frac{\belief(\cdot|\csignali)}{\prior(\cdot)}\right\|},
\end{align*}
where the $\cdot$ and $/$ operations are meant componentwise, and $\|\cdot\|$ is the $l^1$-norm.\label{ftn:belief-transform}} Thus, up to changing $\typed(\type)$ by $\typed(\type|\state)$, and interpreting the draw from the Blackwell experiment as the posterior of an agent with prior belief \prior, the result follows.
\end{proof}
\begin{proof}[Proof of \autoref{theorem:cmd-two-stage}]
~
\paragraph{``Only if'' direction} Similar to the proof of \autoref{theorem:cmd-two-stage-multi}, we focus on the case in which \type\ and \state\ are independent. Suppose $\outcome\in\Delta\left(\Arg\right)$ is implemented by an incentive compatible and individually rational two-stage mechanism. That is,
\begin{align}\label{eq:ts}
\outcome\aarg=\prior(\state)\typed(\type)\int_{\Beliefs}\allocr(\alloc|\type,\belief)\beliefr(d\belief|\state),
\end{align}
and \allocr\ is incentive compatible and individually rational on the support of $\prior\otimes\beta$. We construct an incentive compatible and individually rational calibrated mechanism that implements \outcome.

First, if \outcome\ satisfies \autoref{eq:ts}, \cite{rubin1958note} and Carath\'eodory's theorem \citep[Theorem 5.32]{aliprantis2006infinite} imply that a finite support $\beliefr^\prime:\States\to\Delta\left(\{\belief_1,\dots,\belief_K\}\right)$ exists such that  for all $\aarg\in\Arg$\footnote{Namely, \autoref{eq:ts} implies that for all $(\type,\state)$, $\outcome(\cdot|\type,\state)\in\mathrm{cl}\,\mathrm{co}\{\allocr(\cdot|\type,\belief):\belief\in\Beliefs\}$. \cite{rubin1958note} implies that our under assumptions $\mathrm{cl}\,\mathrm{co}\{\allocr(\cdot|\type,\belief):\belief\in\Beliefs\}=\mathrm{co}\{\allocr(\cdot|\type,\belief):\belief\in\Beliefs\}$, and the rest of the claim follows from Carath\'eodory's theorem.}
\begin{align*}
\outcome\aarg=\prior(\state)\typed(\type)\sum_{k=1}^K\allocr(\alloc|\type,\belief_k)\beliefr^\prime(\{\belief_k\}|\state).
\end{align*}
For each $\state\in\States$, partition $[0,1]=\cup_{k=1}^{K-2}[b_k^\state,b_{k+1}^\state)\cup[b_{K-1}^\state,1]$, where $b_1=0$, and for all $k\in\{1,\dots,K-2\}$, $b_{k+1}^\state=\sum_{l=1}^k\beliefr^\prime(\{\belief_l\}|\state)$. Define for $\cor\in[b_k^\state,b_{k+1}^\state)$
\begin{align*}
\rmech(\alloc|\type,\state,\cor)=\allocr(\alloc|\type,\belief_k).
\end{align*}
The calibrated information structure is $\cexp(\state,\cor)=\rmech(\cdot|\cdot,\state,\cor)=\allocr(\cdot|\type,\belief_k)$ for $\cor\in[b_k^\state,b_{k+1}^\state)$ if $m\leq K-2$ or $\cor\in[b_{K-1}^\state,1]$.

We now show that for all \type\ and all $s\in\mathrm{supp}\; \cexp$, the mechanism \rmech\ is incentive compatible and individually rational. Note that \cexp\ has finite support, and  let $s\in\mathrm{supp}\;\cexp$ and let $\belief(\cdot|s)$ denote the updated posterior. Then, $k$ exists such that the following holds:
\begin{align*}
\belief(\state|s)=\frac{\prior(\state)\lebesgue(\{\cor:\excdf(\state,\cor)=s\})}{\sum_{\stateb\in\States}\prior(\stateb)\lebesgue(\{\cor:\excdf(\stateb,\cor)=s\})}=\frac{\prior(\state)(b_{k+1}^\state-b_k^\state)}{\sum_{\stateb\in\States}\prior(\stateb)(b_{k+1}^{\stateb}-b_k^{\stateb})}=\belief_k(\state).
\end{align*}
Moreover, because $\rmech(\cdot|\cdot,\state,\cor)=\allocr(\cdot|\cdot,\belief_k)$, then it satisfies the agent's incentive compatibility and individual rationality constraints when she holds belief $\belief_k$.\footnote{To extend the result to the case in which the agent's type is correlated with the state, note the following. Knowing \prior\ updates to $\belief_k$ conditional on \signal\ is enough to pin down the agent's belief $\belief(\cdot|\type,\signal)$, with respect to which the agent's incentive compatibility and individual rationality constraints are defined (see \autoref{ftn:belief-transform}).} 
\end{proof}

\newpage
{\LARGE{\textbf{Supplementary Appendix}}}
\section{Omitted proofs from \autoref{sec:microf}}\label{appendix:microf}

\subsection{Repeated Mechanisms}\label{appendix:repeated}
In this section, we present the proof of \autoref{theorem:repeated}. To do so, we first complete the formal definition of the game induced by repeating mechanism $\rmech:\mssgs\times\States\times\fcor\to\Delta(\Alloc)$, by specifying the histories, strategy space, and the distribution over terminal histories induced by the agent's strategy and the mechanism. Having laid this groundwork, we describe the proof strategy, and then provide the formal details of the proof. Throughout this section, we use the shorthand $\estates=\States\times\fcor$, and denote its elements by \estate.

\paragraph{Histories and strategies} Histories through period $t\in\naturals$ are defined as $H^t\equiv(\Types\times\mssgs\times\Alloc)^{t-1}$. The set of infinite histories from the agent's point of view is $H^\infty$. The set of terminal histories is $\terminalsh\equiv\estates\times H^\infty$, where recall \fcor\ is finite and endowed with some measure  \cord. %\footnote{When the mechanism's randomization device is $[0,1]$}

The agent's behavioral strategy is defined as a collection $\strat\equiv(\strat_t)_{t\in\naturals}$ such that for all $t\geq 1$
\[\strat_t:H^t\times\Types\to\Delta(\mssgs).\]
The tuple of distributions $(\prior,\cord,\typed)$ together with the mechanism $\rmech$ and the agent's strategy $\strat$ determine a joint distribution over $\terminalsh$ by the Ionescu-Tulcea theorem \citep[Theorem 10.7.3]{bogachev2007measure}. We provide more details on this probability distribution below. Denote by $\mathbb{P}_{(\prior,\cord,\typed,\rmech,\strat)}$ and $\mathbb{E}_{(\prior,\cord,\typed,\rmech,\strat)}$ the probability distribution over the terminal histories and the expectation with respect to this distribution, respectively. Whenever it is not likely to lead to confusion, we drop the dependence on $(\prior,\cord,\typed,\rmech,\strat)$, and whenever we want to emphasize the dependence on the agent's strategy we note the dependence on \strat.
\color{black}
\paragraph{The distribution over terminal histories \terminalsh} For future use, we review the construction of $\mathbb{P}_\strat$. For each $t$, the distributions  $(\prior,\cord,\typed)$ together with the mechanism $\rmech$ and the agent's strategy $\strat$  determine a distribution over $\estates\times H^t$, which we denote by $\mathbb{P}_\strat^t\in\Delta(\estates\times H^t)$. Note that for any subset $\tilde\terminals^t\subset\estates\times H^t$, 
\begin{align}\label{eq:p-t-t+1}
\mathbb{P}_\strat^t(\tilde\terminals^t)=\mathbb{P}_\strat^{t+1}(\tilde\terminals^t\times(\Types\times\mssgs\times\Alloc)).
\end{align}
Moreover, 
\begin{align}\label{eq:p-t-t+1-2}
\mathbb{P}_\strat^{t+1}(\estate,h^t,\type,\mssg,\alloc)=\mathbb{P}_\strat^t(\estate,h^t)\typed(\type)\strat_t(h^t,\type)(m)\rmech(\alloc|\mssg,\estate).
\end{align}
By the Ionescu-Tulcea theorem, the distribution $\mathbb{P}_\strat\in\Delta(\estates\times H^\infty)$ is the unique distribution that satisfies that for all $t\in\naturals$, $\tilde\terminals^t\subset\estates\times H^t$, 
\begin{align}\label{eq:p-sigma}
\mathbb{P}_\strat(\tilde\terminals^t\times\prod_{s=t+1}^\infty(\Types\times\mssgs\times\Alloc))=\mathbb{P}_\strat^t(\tilde\terminals^t).
\end{align}

\paragraph{Belief system} The agent's beliefs over \estates\ at the beginning of each $t$ are determined by the belief system, which in a slight abuse of notation we denote by $\belief_t:H^t\to\Delta(\estates)$. The belief system satisfies
\begin{align*}
\mathbb{P}_\strat^t(h^t)\belief_t(\estate|h^t)=\mathbb{P}_\strat^t(\estate,h^t).
\end{align*}
That is, whenever $h^t$ is such that $\mathbb{P}_\strat\left(\{\tilde{h}\in H^\infty:\tilde{h}^t=h^t\}\right)>0$,
\begin{align*}
\belief_t(\estate|h^t)=\frac{\mathbb{P}_\strat^t(\estate,h^t)}{\mathbb{P}_\strat^t(h^t)}=\mathbb{P}_\strat^t(\estate|h^t).
\end{align*}
Given $\mathbb{P}_\strat\in\Delta(\terminalsh)$, define $\belief_\infty(\estate|h^\infty)\equiv\mathbb{P}_\strat(\estate|h^\infty)$ to be the belief system conditional on the whole terminal history $h^\infty$.
\begin{remark}[Belief system and strategies as functions on $\terminalsh$] Whereas the beliefs and strategies are defined on the finite histories, it is sometimes convenient to write them as functions on $\terminalsh$ that are adapted to $H^t$. \end{remark}

\paragraph{A property of the belief system} We collect here a property of the belief system which we use in our proofs below.

\begin{lemma}[Martingale property under weak$*$ convergence]\label{lemma:beliefs-0}
    $\belief_t(h^\infty)\weakc\belief_\infty(h^\infty)$ $\mathbb{P}_\strat$-almost surely.
    \end{lemma}
    This and the proof of other technical results are in \autoref{appendix:aux}.

\subsubsection{Proof of \autoref{theorem:repeated} (necessity)} We are now ready to present the proof of \autoref{theorem:repeated}, starting by the ``only if'' direction. Let $\outcome\in\Delta(\Arg)$ denote the outcome distribution implemented by repeated mechanism \rmech\ under best response \strat, and let  $\occup_\strat$ denote the associated occupation measure, the definition of which we reproduce below for ease of reference:
\begin{align}\label{eq:occup}
\occup_\strat\argme=\lim_{T\to\infty}\frac{1}{T}\mathbb{E}_\strat\left[\sum_{t=1}^T\mathbbm{1}\left[(\alloc_t,\type_t,\mssg_t,\estate^\prime)=\argme\right]\right]=\lim_{T\to\infty}\occup_\strat^T\argme,
\end{align}
where recall limits are in the weak* sense. We show that \outcome\ can be implemented by an incentive compatible and individually rational two-stage mechanism. 

To this end, we consider two sequences of \emph{extended} occupation measures on \Argbe:
\begin{align}\label{eq:eoccup-r-1}
\eoccup_\strat^{T,1}(\{\argme\}\times\tilde\Delta)=\frac{1}{T}\mathbb{E}_\strat\left[\sum_{t=1}^T\mathbbm{1}[(\alloc_t,\type_t,\mssg_t,\estate^\prime)=\argme]\mathbbm{1}[\belief_t\in\tilde\Delta]\right],\\
\eoccup_\strat^{T,2}(\{\argme\}\times\tilde\Delta)=\frac{1}{T}\mathbb{E}_\strat\left[\sum_{t=1}^T\mathbbm{1}[(\alloc_t,\type_t,\mssg_t,\estate^\prime)=\argme]\mathbbm{1}[\belief_{t+1}\in\tilde\Delta]\right].\label{eq:eoccup-r-2}
\end{align}
We note the following. First, \autoref{eq:eoccup-r-1} counts the beliefs at the beginning of period $t$, while \autoref{eq:eoccup-r-2} counts the beliefs at the end of period $t$ (after the realization of \type, \mssg, and \alloc.) \autoref{eq:eoccup-r-1} is key to obtain the (limit) independence of the belief and type distributions, while \autoref{eq:eoccup-r-2} allows us to obtain the (limit) independence of the allocation and the state, conditional on the induced belief. Second, $\occup_\strat^T$ is the marginal of both $\eoccup_\strat^{T,1}$ and $\eoccup_\strat^{T,2}$. Third, by \autoref{lemma:beliefs-0}, $\belief_t\weakc\belief_\infty$, and hence both $\eoccup_{\strat}^{T,1}$ and $\eoccup_{\strat}^{T,2}$ have the same set of subsequential limits, which we record for future reference below (see \autoref{appendix:aux} for the proof):
\begin{lemma}\label{prop:extended-occup}
The occupation measures $\eoccup_{\strat}^{T,1}$ and $\eoccup_{\strat}^{T,2}$ have the same set of subsequential limits.
\end{lemma}

The proof of necessity of \autoref{theorem:repeated} proceeds in five steps. First, we show that the marginal of $\eoccup_\strat^{T,1}$ on \ebeliefs, which we denote by $\bsplit_\strat^T$ weak*-converges to $\mathbb{P}_\strat\circ\belief_\infty^{-1}$. We denote this limit measure by $\bsplit_\strat$. By \autoref{prop:extended-occup}, $\bsplit_\strat$ is also the (limit) marginal of $\eoccup_\strat^{T,2}$ on \ebeliefs.

Second, we show that up to a subsequence $\eoccup_\strat^{T,1},\eoccup_\strat^{T,2}\weakc\eoccup_\strat$. Furthermore, transition probabilities $\bsplit_\strat\in\Delta(\ebeliefs),\rep:\Types\times\ebeliefs\to\Delta(\mssgs),\allocrb:\mssgs\times\ebeliefs\to\Delta(\Alloc)$ exist such that
\begin{align}\label{eq:occup-decomp}
\occup_\strat\argme=\int_{\ebeliefs}\belief(\estate)\typed(\type)\rep(\mssg|\type,\belief)\allocrb(\alloc|\mssg,\belief)\bsplit_\strat(d\belief).
\end{align}
Hence, the agent's payoff when faced with mechanism \rmech\ and playing strategy \strat\ can be written as:
\begin{align}\label{eq:payoff-r}
\mathbb{E}_{\eoccup_{\strat}}\left[\pay\aarg\right]=
\int_{\ebeliefs} \sum_{\type\in\Types}\typed(\type)\sum_{\mssg\in\mssgs}\rep(m|\type,\belief)\mathbb{E}_{\estate\sim\belief}\left[\sum_{\alloc\in\Alloc}\allocrb(\alloc|\belief,\mssg)\pay\aarg\right]\bsplit_\strat(d\belief).
\end{align}
Third, we show that for all $\type\in\Types$
\begin{align}\label{eq:max-r}
\mathbb{E}_{\bsplit_\strat} \left\{\sum_{\mssg\in\mssgs}\rep(m|\type,\belief)\mathbb{E}_{\estate\sim\belief}\left[\sum_{\alloc\in\Alloc}\allocrb(\alloc|\belief,\mssg)\pay\aarg\right]-\max_{\mssg\in\mssgs}\mathbb{E}_{\estate\sim\belief}\left[\sum_{\alloc\in\Alloc}\allocrb(\alloc|\belief,\mssg)\pay\aarg\right]\right\}=0.
\end{align}
Equations \ref{eq:payoff-r} and \ref{eq:max-r} allow us to identify the incentive compatible and individually rational allocation rule of the two-stage mechanism that implements \outcome.

Fourth, whereas the previous steps identify a two-stage mechanism expressed in terms of posterior beliefs over \estates, we show how to obtain a two-stage mechanism expressed in terms of posterior beliefs over \States. Finally, we show that the agent's payoff in \autoref{eq:payoff-r} coincides with the payoff she would get when best responding to the information structure calibrated to \rmech.

\paragraph{Step 1} Having defined the extended occupation measure in \autoref{eq:eoccup-r-1}, we present here a property we use in our proof. Let $\bsplit_\strat^T$ denote the marginal of $\eoccup_\strat^{T,1}$ on \ebeliefs. That is, for any measurable subset $\tilde\Delta\subset\ebeliefs$, define
\begin{align}\label{eq:belief-occup-T}
\bsplit_\strat^T(\tilde\Delta)=\frac{1}{T}\mathbb{E}_\strat\left[\sum_{t=1}^T\mathbbm{1}[\belief_t\in \tilde\Delta]\right].
\end{align}
In \autoref{appendix:aux}, we prove the following:
\begin{lemma}\label{prop:belief-occup}
The sequence of measures $(\bsplit_\strat^T)_{T\in\naturals}$ defined by \autoref{eq:belief-occup-T} converges in the weak* sense to the push-forward measure $\bsplit_\strat\equiv\mathbb{P}_\strat\circ\belief_\infty^{-1}$, where $\belief_\infty(h^\infty)=\mathbb{P}_\strat(\cdot|h^\infty)$. 
\end{lemma}

\paragraph{Step 2} To show that \autoref{eq:occup-decomp} holds, we show the following properties of  $\eoccup_\strat^{T,1}$ and $\eoccup_\strat^{T,2}$. On the one hand, $\eoccup_\strat^{T,1}$ satisfies that for all $\testf\in C_b(\Argbe)$,
\begin{align}\label{eq:dis-11}
\int_{\Argbe}\testf\argmeb d\eoccup_\strat^{T,1}=\int_{\Types\times\mssgs\times\ebeliefs}\mathbb{E}_\belief\left[\mathbb{E}_{\rmech(\cdot|\mssg,\estate)}\left[\testf\argmeb\right]\right]d\eoccup_{\strat,\Types\mssgs\Delta}^{T,1},
\end{align}
and for all $\testfb\in C_b(\Types\times\Beliefs)$,
\begin{align}\label{eq:dis-12}
\int_{\Types\times\ebeliefs}\testfb(\type,\belief)d\eoccup_{\strat,\Types\Delta}^{T,1}=\int_{\ebeliefs}\int_\Types\typed(\type)q(\type,\belief)d\bsplit_\strat^T.
\end{align}
On the other hand, $\eoccup_\strat^{T,2}$ satisfies that for all $\testf\in C_b(\Argbe)$,
\begin{align}\label{eq:dis-2}
\int_{\Argbe}\testf\argmeb d\eoccup_\strat^{T,2}=\int_{\Types\times\mssgs\times\Alloc\times\ebeliefs}\mathbb{E}_{\estate\sim\belief}\left[\testf\argmeb\right]d\eoccup_{\strat,\Types\mssgs\Alloc\Delta}^{T,2}.
\end{align}
In the expressions above, the subscripts on $\eoccup_\strat^{T,k}$ next to \strat\ are the spaces over which we take the marginals, and $\Delta$ is shorthand notation for \ebeliefs. Because $\Delta(\Argbe)$ is compact \citep[Theorem 15.11]{aliprantis2006infinite}, $\eoccup_{\strat}^{T,1}$ has a convergent subsequence $(\eoccup_{\strat}^{T_n,1})_{n\in\naturals}$, which by \autoref{prop:extended-occup} is also a convergent subsequence of $\eoccup_{\strat}^{T,2}$. Let $\eoccup_\strat$ denote the weak$^*$ limit along $T_n$. The continuity of the projection implies that  $\occup_\strat$ is the marginal of $\eoccup_\strat$ on \Argme, and $\bsplit_\strat\equiv\mathbb{P}_\strat\circ\belief_\infty^{-1}$ is the marginal on \Beliefs. Equations  \ref{eq:dis-12} and \ref{eq:dis-2} together imply that $\eoccup_\strat$ admits the decomposition in the right hand side of \autoref{eq:occup-decomp}, and the result follows.

To show \autoref{eq:dis-11} holds, use that $\eoccup_\strat^{T,1}$ has finite support to write it as follows:
\begin{align}\label{eq:belief-fact}
\eoccup_\strat^{T,1}\argmeb&=\frac{1}{T}\sum_{t=1}^T\sum_{h^t\in H^t}\mathbb{P}_\strat^t(\estate,h^t)\typed(\type)\strat_t(h^t,\type)(\mssg)\rmech(\alloc|\mssg,\estate)\mathbbm{1}[\belief_t=\belief]\nonumber\\
&=\frac{1}{T}\sum_{t=1}^T\sum_{h^t\in H^t}\mathbb{P}_\strat^t(h^t)\belief(\estate)\typed(\type)\strat_t(h^t,\type)(\mssg)\rmech(\alloc|\mssg,\estate)\mathbbm{1}[\belief_t=\belief],
\end{align}
where the second equality uses that $\belief_t(h^t)=\mathbb{P}_\strat^t(\cdot|h^t)$.

\autoref{eq:belief-fact} implies the following holds for every bounded continuous function  $\testf\in C_b(\Argbe)$:
\begin{align*}
\mathbb{E}_{\eoccup_\strat^{T,1}}\left[\testf\argmeb\right]=\mathbb{E}_{\eoccup_{\strat,\Types\mssgs\Delta}^{T,1}}\left[\mathbb{E}_{\belief}\left[\mathbb{E}_{\rmech(\cdot|\mssg,\estate)}\left[\testf\argmeb
\right]\right]\right],
\end{align*}
This completes the proof that \autoref{eq:dis-11} holds. Letting $V_\testf(\type,\mssgs,\belief)=\mathbb{E}_{\belief}\left[\mathbb{E}_{\rmech(\cdot|\mssg,\estate)}\left[\testf\argmeb\right]\right]$, we have that 
\begin{align*}
\int_{\Argbe} \testf d\eoccup_\strat^{T,1}=\int_{\Types\times\mssgs\times\ebeliefs}V_\testf d\eoccup_{\strat,\Types\mssgs\Delta}^{T,1}\Leftrightarrow\int_{\Argbe}(\testf-V_\testf)d\eoccup_\strat^{T,1}=0.
\end{align*}
Because $\testf-V_g\in C_b(\Argbe)$ and $\eoccup_\strat^{T_n,1}\weakc\eoccup_\strat$, we conclude that 
\begin{align}\label{eq:fact-1}
\mathbb{E}_{\eoccup_\strat}\left[\testf\argmeb\right]=\int_{\Types\times\mssgs\times\ebeliefs}\mathbb{E}_{\belief}\left[\mathbb{E}_{\rmech(\cdot|\mssg,\estate)}\left[\testf\argmeb\right]\right]d\eoccup_{\strat,\Types\mssgs\Delta}.
\end{align}
To show that \autoref{eq:dis-12} holds, note that the marginal of  $\eoccup_\strat^{T,1}$ on $\Types\times\ebeliefs$ equals $\bsplit_\strat^T\otimes\typed$. Indeed, fix any continuous function $\testfb\in C_b(\Types\times\ebeliefs)$ and note that for all $T$
\begin{align*}
\mathbb{E}_{\eoccup_{\strat,\Types\Delta}^{T,1}}\left[\testfb(\type,\belief)\right]=\mathbb{E}_{\bsplit_\strat^T}\left[\sum_{\type\in\Types}\typed(\type)\testfb(\type,\belief)\right].
\end{align*}
Letting $V_\testfb(\belief)=\sum_{\type\in\Types}\typed(\type)\testfb(\type,\belief)$, we have that for all $T$
\begin{align*}%\label{eq:fact-2}
\int_{\ebeliefs\times\Types}\left(\testfb(\type,\belief)-V_\testfb(\belief)\right)d\eoccup_{\strat,\Types\Delta}^{T,1}=0.
\end{align*}
Because $\testfb-V_\testfb\in C_b(\Types\times\ebeliefs)$ and $\eoccup_\strat^{T_n,1}\weakc\eoccup_\strat$, we conclude that 
\begin{align}\label{eq:fact-2}
\int_{\Types\times\ebeliefs}\testfb(\type,\belief)d\eoccup_{\strat,\Types\Delta}=\int_{\ebeliefs}\int_\Types\typed(\type)\testfb(\type,\belief)d\bsplit_\strat.
\end{align}
Lastly, to show that \autoref{eq:dis-2} holds, note that we can write $\eoccup_\strat^{T,2}$ as follows (once again, we use that for finite $T$, it has finite support):
\begin{align*}
&\eoccup_\strat^{T,2}\argmeb=\frac{1}{T}\sum_{t=1}^T\sum_{h^t\in H^t}\mathbb{P}_\strat^t(h^t)\mathbb{P}_\strat^t(\estate|h^t)\typed(\type)\strat_t(h^t,\type)(\mssg)\rmech(\alloc|\mssg,\estate)\mathbbm{1}[\belief_{t+1}(h^t,\type,\mssg,\alloc)=\belief]=\\
&=\frac{1}{T}\sum_{t=1}^T\sum_{h^t\in H^t:\belief_{t+1}(h^t,\type,\mssg,\alloc)=\belief}\mathbb{P}_\strat^{t+1}(\estate,h^t,\type,\mssg,\alloc)=\frac{1}{T}\sum_{t=1}^T\sum_{h^t\in H^t:\belief_{t+1}(h^t,\type,\mssg,\alloc)=\belief}\belief(\estate)\mathbb{P}_\strat^{t+1}(h^t,\type,\mssg,\alloc)\\
&=\belief(\estate)\frac{1}{T}\sum_{t=1}^T\sum_{h^t\in H^t:\belief_{t+1}(h^t,\type,\mssg,\alloc)=\belief}\mathbb{P}_\strat^{t+1}(h^t,\type,\mssg,\alloc)=\belief(\estate)\eoccup_{\strat,\Alloc\Types\mssgs\Delta}^{T,2}(\alloc,\type,\mssg,\belief),
\end{align*}
where the last expression follows from noting that the term multiplying $\belief(\state)$ in the first expression in the third line is $\sum_{\estate}\occup_\strat^{T,2}\argmeb$.

Then, for every $\testf\in C_b(\Argbe)$, we have that
\begin{align*}
\mathbb{E}_{\eoccup_\strat^{T,2}}\left[\testf\argmeb\right]=\mathbb{E}_{\eoccup_{\strat,\Alloc\Types\mssgs\ebeliefs}^{T,2}}\left[\mathbb{E}_{\estate\sim\belief}\left[\testf\argmeb\right]\right],
\end{align*}
which completes the proof that \autoref{eq:dis-2} holds. Letting 
$V_\testf(\alloc,\type,\mssg,\belief)=\mathbb{E}_{\estate\sim\belief}\left[\testf\argmeb\right]$, we have that
\begin{align*}
\int\left(\testf-V_\testf\right)d\eoccup_\strat^{T,2}=0.
\end{align*}
Because $\testf-V_\testf\in C_b(\Argbe)$ and $\eoccup_\strat^{T_n,2}\weakc\eoccup_\strat$, we conclude that
\begin{align}\label{eq:fact-3}
\mathbb{E}_{\eoccup_\strat}\left[\testf\argmeb\right]=\int_{\Alloc\times\Types\times\mssgs\times\ebeliefs}\mathbb{E}_{\estate\sim\belief}\left[\testf\argmeb\right]\eoccup_{\strat,\Alloc\Types\mssgs\Delta}(d(\alloc,\type,\mssg,\belief)).
\end{align}

Equations \ref{eq:fact-2} and \ref{eq:fact-3} imply $\eoccup_\strat$ admits the following disintegration:
\begin{align}\label{eq:concl-1}
\eoccup_\strat(\{(\alloc,\type,\mssg,\estate)\}\times\tilde\Delta)=\int_{\tilde\Delta}\belief(\estate)\typed(\type)\allocrb(\alloc|\type,\mssg,\belief)\rep(\mssg|\type,\belief)\bsplit_\strat(d\belief),
\end{align}
where we disintegrated $\eoccup_{\strat,\Alloc\Types\mssgs\Delta}$ first along $\Types\times\ebeliefs$\textemdash and used \autoref{eq:fact-2} to obtain the independence of \Types\ and \ebeliefs\textemdash and then further disintegrated the distribution of $\Alloc\times\mssgs$ conditional on $\Theta\times\ebeliefs$. Now, \autoref{eq:fact-1} implies that the following also holds
\begin{align}\label{eq:concl-2}
\eoccup_\strat(\{(\alloc,\type,\mssg,\estate)\}\times\tilde\Delta)=\int_{\tilde\Delta}\belief(\estate)\typed(\type)\mechanism(\alloc|\mssg,\estate)\rep(\mssg|\type,\belief)\bsplit_\strat(d\belief),
\end{align}
where once again we use the uniqueness of disintegration. Because Equations \ref{eq:concl-1} and \ref{eq:concl-2} hold for any tuple $(\alloc,\type,\mssg,\estate)$ and measurable subset $\tilde\Delta$ of $\ebeliefs$, we conclude that (i) $\allocrb(\alloc|\type,\mssg,\belief)$ does not depend on \type\ $\bsplit_\strat$-almost everywhere, and (ii) $\rmech(\cdot|\mssg,\estate)$ is constant on $\estate$ in the support of $\belief$  $\bsplit_\strat$-almost everywhere. This concludes the proof of Step 2.

\paragraph{Step 3} We now argue that the agent achieves 
\begin{align}\label{eq:rmech-max}
\pay^*(\belief)\equiv\sum_{\type\in\Types}\typed(\type)\max_{\mssg\in\mssgs} \sum_{\estate\in\estates}\belief(\estate)\sum_{\alloc\in\Alloc}\rmech(\alloc|\mssg,\estate)\pay\aarg=\sum_{\type\in\Types}\typed(\type)\max_{\mssg\in\mssgs} \sum_{\estate\in\estates}\belief(\estate)\sum_{\alloc\in\Alloc}\allocrb(\alloc|\mssg,\belief)\pay\aarg,
\end{align}
on the support of $\bsplit_\strat$, where the second equality follows from Step 2.  Toward a contradiction, suppose this is not the case; that is,
\begin{align*}
\mathbb{E}_{\bsplit_\strat}\left[\sum_{\type\in\Types}\typed(\type)\sum_{\mssg\in\mssgs}\rep(\mssg|\type,\belief)\sum_{\estate\in\estates}\belief(\estate)\sum_{\alloc}\rmech(\alloc|\mssg,\estate)\pay\aarg\right]<\mathbb{E}_{\bsplit_\strat}\left[\sum_{\type\in\Types}\pay^*(\belief)\right]=U^*.
\end{align*}
We show that the agent can achieve a payoff arbitrarily close to $U^*$ by playing according to $\strat$ until some finite $T$ and then best-responding to her beliefs at time $T$ in every period thereafter; a contradiction.

Consider a strategy \stratb\ which until some period $T$ plays according to \strat\ and after period $T$ best responds to $\belief_T(h^T)\in\ebeliefs$. Because  payoffs accumulated on a finite number of periods are irrelevant to long-run payoffs, this strategy results in a payoff:
\begin{align*}
    &\sum_{h^T\in H^T}\mathbb{P}_{\strat}^T(h^T)\sum_{\type\in\Types}\typed(\type)\max_{\mssg\in\mssgs}\left[\sum_{\estate\in\estates}\belief_T(\estate)\sum_{\alloc\in\Alloc}\rmech(\alloc|\mssg,\estate)\pay\aarg\right]=    \sum_{h^T\in H^T}\mathbb{P}_{\strat}^T(h^T)\pay^*(\belief_T(h^T))\\
    &=\mathbb{E}_{\mathbb{P}_\strat^T\circ\belief_T^{-1}}\left[\pay^*(\belief)\right]=\mathbb{E}_{\mathbb{P}_\strat\circ\belief_T^{-1}}\left[\pay^*(\belief)\right],
\end{align*}
where the last equality follows as $\belief_T$ is adapted to the histories through $T$. Similar arguments to \autoref{prop:belief-occup} imply that $\mathbb{P}_\strat\circ\belief_T^{-1}\weakc\mathbb{P}_\strat\circ\belief_\infty^{-1}\equiv\bsplit_\strat$. Noting that $u^*:\ebeliefs\to\reals$ is continuous and bounded (as it is the maximum of linear functions in beliefs), we obtain that as $T\to\infty$,
\begin{align*}
\mathbb{E}_{\mathbb{P}_\strat\circ\belief_T^{-1}}\left[u^*(\belief)\right]\to\mathbb{E}_{\bsplit_\strat}\left[\pay^*(\belief)\right].
\end{align*}
It follows that for every $\delta>0$, we can find $T$ large enough so that $|\mathbb{E}_{\mathbb{P}_\strat\circ\belief_T^{-1}}\left[\pay^*(\belief)\right]- U^*|<\delta$, contradicting the optimality of \strat.

We conclude that $\allocr:\Types\times\ebeliefs\to\Delta(\Alloc)$ defined as follows:
\begin{align*}
\allocr(\alloc|\type,\belief)=\sum_{\mssg\in\mssgs}\rep(\mssg|\type,\belief)\allocrb(\alloc|\mssg,\belief),
\end{align*}
is incentive compatible and individually rational $\bsplit_\strat$-almost everywhere. 

\paragraph{Step 4} We now show how to derive a two-stage mechanism $\beliefr^*:\States\to\Delta(\Beliefs)$ and an allocation rule $\allocr^*:\Types\times\Beliefs\to\Delta(\Alloc)$ that implement \outcome. First, note that the agent's payoff when her type is \type\ and the induced belief is $\belief\in\ebeliefs$, can be written as
\begin{align*}
\sum_{\estate\in\estates}\belief(\estate)\sum_{\alloc\in\Alloc}\allocr(\alloc|\type,\belief)\pay\aarg=\sum_{\state\in\States}\belief_\States(\state)\sum_{\alloc\in\Alloc}\allocr(\alloc|\type,\belief)\pay\aarg,
\end{align*} 
where $\belief_\States$ is the marginal of \belief\ on \States\ and the equality follows because the realization of \cor\ is payoff-irrelevant. By Step 3,  $\allocr(\cdot|\type,\belief)$ is individually rational and incentive compatible when the agent holds belief $\belief_\States$. 

Furthermore, for each $\aarg\in\Arg$, we have 
\begin{align*}
\outcome\aarg=\sum_{\cor\in\fcor}\int_{\ebeliefs}\belief(\state,\cor)\typed(\type)\allocr(\alloc|\type,\belief)\bsplit_\strat(d\belief)=\int_{\ebeliefs}\belief_\States(\state)\typed(\type)\allocr(\alloc|\type,\belief)\bsplit_\strat(d\belief).
\end{align*}
For each $\type\in\Types$, consider the joint distribution $Q_\type\in\Delta(\Alloc\times\Beliefs)$ defined as follows:
\begin{align*}
Q_\type(\{\alloc\}\times\tilde\Delta)=\int_{\ebeliefs}\mathbbm{1}[\belief_{\States}\in\tilde\Delta]\allocr(\alloc|\type,\belief)\bsplit_\strat(d\belief)=\int_{\tilde\Delta} \allocr^*(\alloc|\type,\belief_\States)\bsplit^*(d\belief_\States),
\end{align*}
where the third equality follows from disintegration (note $\allocr^*(\cdot|\cdot,\belief_\States)=\mathbb{E}[\allocr(\cdot|\cdot,\tilde\belief)|\tilde\belief_\States=\belief_\States]$). By the first argument in Step 4, $\allocr^*(\cdot|\cdot,\belief_\States)$ is individually rational and incentive compatible when the agent holds $\belief_\States$. We obtain that
\begin{align*}
\outcome\aarg=\int_{\Beliefs}\belief_\States(\state)\typed(\type)\allocr^*(\alloc|\type,\belief_\States)\bsplit^*(d\belief_\States).
\end{align*}
Defining for all $\state\in\States$ and measurable subsets $\tilde\Delta\in\Beliefs$, 
\begin{align*}
\beliefr^*(\tilde\Delta|\state)=\int_{\tilde\Delta}\frac{\belief(\state)}{\prior(\state)}\bsplit^*(d\belief),
\end{align*}
Steps 2-4 together imply that \outcome\ can be implemented by the incentive compatible and individually rational two-stage mechanism $(\beliefr^*,\allocr^*)$.
\paragraph{Step 5: The agent adequately learns} Finally, we argue that the agent earns the same payoff as if she had access to the information structure calibrated to \rmech, \cexp.
\begin{lemma}\label{lemma:adequate-learning}
Let $\bsplit_\rmech$ denote the belief distribution induced by the calibrated information structure \cexp. Then, the agent's payoff under \strat\ equals
\begin{align}\label{eq:target}
U(\bsplit_\rmech)\equiv\mathbb{E}_{\bsplit_\rmech}\left[\sum_{\type\in\Types}\typed(\type)\max_{m\in\mssgs}\sum_{\estate\in\estates}\belief(\estate)\sum_{\alloc\in\Alloc}\rmech(\alloc|\mssg,\estate)\pay\aarg\right].
\end{align}
\end{lemma}
The proof of this is standard, and hence we defer it to \autoref{appendix:aux}.

\subsubsection{Proof of \autoref{theorem:repeated} (sufficiency)}
Suppose $\outcome\in\Delta\left(\Arg\right)$ is implemented by an incentive compatible and individually rational two-stage mechanism. That is,
\begin{align}\label{eq:ts-bis}
\outcome\aarg=\prior(\state)\typed(\type)\int_{\Beliefs}\allocr(\alloc|\type,\belief)\beliefr(d\belief|\state),
\end{align}
and \allocr\ is incentive compatible and individually rational on the support of $\prior\otimes\beta$. As in the proof of \autoref{theorem:cmd-two-stage}, a finite support $\beliefr^\prime:\States\to\Delta\left(\{\belief_1,\dots,\belief_K\}\right)$ exists such that $(\beliefr^\prime,\allocr)$ implement \outcome. As in \cite{green2022two}, the experiment \beliefrb\ can be generated by a finite information structure $\excdf:\States\times\fcor\to\Beliefs$, where (i) \fcor\ is finite, (ii) \fcor\ is independent of \States, and (iii) $\belief=\excdf(\state,\cor)$. 

Construct a mechanism $\rmech:\Types\times\States\times\fcor\to\Delta(\Alloc)$ such that $\rmech(\cdot|\type,\state,\cor)=\allocr(\cdot|\type,\excdf(\state,\cor))$. (Note that \excdf\ is information structure calibrated to \rmech, but expressed in beliefs.) Consider now the extensive form game induced by such a mechanism.\footnote{We could expand the mechanism by allowing the agent to have a message which triggers the outside option, but this is not necessary as \allocr\ is individually rational.} 

If the agent truthfully reports her type, then the occupation measure induces outcome distribution \outcome. Hence, under truthtelling, the agent's payoff is:
\begin{align}\label{eq:truth-pay}
&U(\strat_{\texttt{truth}})=\sum_{\aarg\in\Arg}\outcome\aarg\pay\aarg=\mathbb{E}_{\bsplit_\rmech}\left[\sum_{\type\in\Types}\typed(\type)\sum_{\alloc\in\Alloc}\allocr(\alloc|\type,\belief)\sum_{\state\in\States}\belief(\state)\pay\aarg\right]\nonumber\\
&=\mathbb{E}_{\bsplit_\rmech}\left[\sum_{\type\in\Types}\typed(\type)\max\left\{\max_{\typeb\in\Types}\sum_{\alloc\in\Alloc}\allocr(\alloc|\typeb,\belief)\sum_{\state\in\States}\belief(\state)\pay\aarg,\sum_{\state\in\States}\belief(\state)\pay(\oo,\type,\state)\right\}\right]=\\
&=\mathbb{E}_{\fcor}\left[\sum_{\type\in\Types}\typed(\type)\max\left\{\max_{\typeb\in\Types}\sum_{\state\in\States}\belief(\state)\sum_{\alloc\in\Alloc}\rmech(\alloc|\typeb,\state,\cor)\pay\aarg,\sum_{\state\in\States}\belief(\state)\pay(\oo,\type,\state)\right\}\right],\nonumber
\end{align}
where (i) $\bsplit_\rmech$ is the belief distribution induced by the information structure \excdf, and (ii) the first equality is by definition of the occupation measure, the second is the definition that \outcome\ is implemented by the two-stage mechanism, the third follows from incentive compatibility and individual rationality of \allocr, and the fourth is definitional.%\footnote{By definition of \excdf, $\rmech(\cdot|\cdot,\state,\cor)$ is constant on the support of $\excdf(\state,\cor)$, so we could take \rmech\ outside in the last line.}

Moreover, the payoff in the last line of \autoref{eq:truth-pay} is the payoff the agent obtains by using the ``learning'' strategy in \autoref{lemma:adequate-learning}, which first extracts all the mechanism can teach her about the state and then uses that information to optimize over her participation and reporting strategies. It follows that truthtelling (and participation) are optimal and \outcome\ is implemented by repeated mechanism \rmech.

\subsection{Dynamic Mechanisms}\label{appendix:dyn-mech}
In this section, we present the proof of \autoref{theorem:dynamic}. To do so, we first complete the formal definition of the game, by specifying the histories, strategy space, and the distribution over terminal histories induced by the agent's strategy and the mechanism. Having laid this groundwork, we describe the proof strategy, and then provide the formal details of the proof.

\paragraph{Mechanisms, histories, and strategies} A dynamic mechanism $(\dmech_t)_{t\in\naturals}$ is a sequence of mappings that condition on the state, the agent's report history, the allocation history, and today's report and output an allocation. By the revelation principle, it is without loss of generality to restrict attention to mechanisms that solicit type reports. 

As in the main text, we expand the set of type reports and allocations by the non-participation decision and the outside option, which we denote by $\Types\Alloc_\emptyset\equiv\Alloc\times\Types\cup\{(\emptyset,\oo)\}$. Then, $\hat{H}^t=(\Types\Alloc_\emptyset)^{t-1}$  denotes the histories of reports (inclusive of the non-participation decision) and allocations at the beginning of time $t\in\naturals$, and let $\hat\terminals^{t}=\States\times\hat{H}^t$. Similarly, let $\hat{H}^\infty=\times_{t\in\naturals}\left(\Types\Alloc_\emptyset\right)$ denote the set of all possible report-allocation outcome paths, and let $\hat\terminals^\infty=\States\times\hat{H}^\infty$.  A dynamic mechanism is then a collection of mappings $(\dmecht)_{t\in\naturals}$ such that $\dmecht:\hat{\terminals}^t\times\Types\to\Delta(\Alloc)$.

To define the agent's strategy, let $H^t=\Types^{t-1}\times\hat{H}^{t-1}$, where the coordinates denote the sequence of realized types, reports (inclusive of participation decisions), and allocations through period $t-1$. A behavioral strategy is a mapping $(\pp_t,\strat_{t}):H^t\times\Types\to[0,1]\times\Delta(\Types)$. 

\paragraph{The distribution over terminal histories \terminalsh} To obtain the complete description of the paths on the tree we need to append $\States$ to  $H^t$; hence the paths through period $t-1$ are $\States\times H^t\equiv\terminals^t$.  The distributions over states, agent's types, the agent's strategy, and the mechanism induce a distribution over the terminal histories $\terminalsh\equiv\States\times H^\infty$, which we denote by $\pr_{\estrat}\in\Delta(\States\times H^\infty)$. We denote by $\mathbb{E}_{\estrat}$ the expectation under this measure. The distribution $\pr_{\estrat}\in\Delta(\States\times H^\infty)$ is the unique distribution that satisfies that for all $t\in\naturals$, $\tilde\terminals^t\subset\States\times\fcor\times H^t$, 
\begin{align*}%\label{eq:p-sigma}
\mathbb{P}_{\estrat}(\tilde\terminals^t\times\prod_{s=t+1}^\infty(\Types\times\Types\Alloc_\emptyset)=\mathbb{P}_\strat^t(\tilde\terminals^t),
\end{align*}
where the distributions $(\pr_{\estrat}^t)_{t\in\naturals}$ satisfy (under participation and truthtelling)
\begin{align*}%\label{eq:p-t-t+1-2}
\mathbb{P}_{\estrat}^{t+1}(\state,h^t,\type,\typeb,\alloc)=\pr_{\estrat}^t(\state,h^t)\typed(\type)\mathbbm{1}[\typeb=\type]\dmecht(\alloc|\state,\hat{h}^t,\typeb).
\end{align*}

\paragraph{Implementation} We focus on incentive-compatible mechanisms $\dmech$ for which (i) a best response, \estrat, exists, and (ii) the occupation measure $\occup_\strat\in\Delta(\Arg)$ exists, where
\begin{align}\label{eq:dyn-occup}
\occup_{\estrat}\aarg&=\lim_{T\to\infty}\frac{1}{T}\mathbb{E}_{\estrat}\left[\sum_{t=1}^T\mathbbm{1}\left[(\alloc_t,\type_t,\stateb)=\aarg\right]\right],
\end{align}
where the limit is in the weak$^*$ sense. In contrast to \autoref{appendix:repeated}, we do not keep track of the agent's type reports in the occupation measure, only the agent's types. Under \estrat\ only truthtelling histories have positive probability.

\subsubsection{Proof of \autoref{theorem:dynamic} (necessity)} 
Let $\outcome\in\Delta(\Arg)$ denote the outcome distribution implemented by an incentive compatible dynamic mechanism \dmech, and let  $\occup_{\estrat}$ denote the occupation measure under the agent's truthtelling strategy.  Below, we show that $\occup_{\estrat}$, and hence \outcome, can be implemented by a two-stage mechanism which lacks profitable undetectable deviations and is ex ante individually rational.

Analogously to the proof of \autoref{theorem:repeated}, we define two sequences of extended occupation measures on \Argb\ defined as follows. Letting $\tilde\Delta$ denote a measurable subset of \Beliefs, define
\begin{align}
\eoccup_{\estrat}^{T,1}(\{\aarg\}\times\tilde\Delta)&=\frac{1}{T}\sum_{t=1}^T\sum_{h^t\in H^t}\pr_{\estrat}^t(\state,h^t)\typed(\type)\strat_t(h^t,\type)(\type)\dmech(\state,\hat{h}^t,\type)(a)\mathbbm{1}[\belief_t(h^t)\in \tilde\Delta],\label{eq:occup-1}\\
\eoccup_{\estrat}^{T,2}(\{\aarg\}\times\tilde\Delta)&=\frac{1}{T}\sum_{t=1}^T\sum_{h^t\in H^t}\pr_{\estrat}^t(\state,h^t)\typed(\type)\strat_t(h^t,\type)(\type)\dmech(\state,\hat{h}^t,\type)(a)\mathbbm{1}[\belief_{t+1}(h^t,\type,\type,\alloc)\in \tilde\Delta].\label{eq:occup-2}
\end{align}
The proof proceeds similarly to that in \autoref{appendix:repeated}. First, we show that the occupation measure $\occup_{\estrat}\in\Delta(\Arg)$ admits the following decomposition
\begin{align}\label{eq:occup-limit}
\occup_{\estrat}\aarg=\int_{\Beliefs}\typed(\type)\belief(\state)\allocr(\alloc|\type,\belief)\bsplit_{\estrat}(d\belief),
\end{align}
where $\bsplit_{\estrat}$ is the distribution over terminal beliefs (cf. \autoref{prop:belief-occup}) and the transition probability $\allocr:\Types\times\Beliefs\to\Delta(\Alloc)$ is our candidate allocation rule. Consequently, the agent's equilibrium payoff can be written as follows:
\begin{align}\label{eq:payoff-dynamic}
\sum_{\aarg\in\Arg}\occup_{\estrat}\aarg\pay\aarg=\int_{\Beliefs}\left[\sum_{\type\in\Types}\typed(\type)\sum_{\state\in\States}\belief(\state)\sum_{\alloc\in\Alloc}\allocr(\alloc|\type,\belief)\pay\aarg\right]\bsplit_{\estrat}(d\belief).
\end{align}
Second, we show that the allocation rule lacks profitable undetectable deviations and is ex ante individually rational.

\paragraph{The occupation measure satisfies \autoref{eq:occup-limit}} To prove that \autoref{eq:occup-limit} holds, we first show that  for all $\testf\in C_b(\Argb)$ and all $T\in\naturals$,
\begin{align}\label{eq:step-1}
\int_{\Argb}\testf\argb d\eoccup_{\estrat}^{T,2}=\int_{\Alloc\times\Types\times\Beliefs}\mathbb{E}_\belief [\testf\argb] d\eoccup_{\estrat,\Alloc\Types\Beliefs}^{T,2}
\end{align}
and for all $\testfb\in C_b(\Types\times\Beliefs)$ and all $T\in\naturals$,
\begin{align}\label{eq:step-2}
\int_{\Types\times\Beliefs}\testfb(\type,\belief)d\eoccup_{\estrat,\Types\Delta}^{T,1}=\int_{\Beliefs}\sum_{\type\in\Types}\typed(\type)q(\type,\belief)d\eoccup_{\estrat,\Delta}^{T,1},
\end{align}
where the subscripts on \eoccup\ next to \estrat\ are the spaces over which we take the marginals, and $\Delta$ is shorthand notation for \Beliefs. We skip the proof of this step as it basically repeats the proof of the analogous step in \autoref{appendix:repeated}.

Because $\Delta(\Argb)$ is compact \citep[Theorem 15.11]{aliprantis2006infinite}, $\eoccup_{\estrat}^{T,1}$ has a convergent subsequence $(\eoccup_{\estrat}^{T_n,1})_{n\in\naturals}$, which by \autoref{prop:extended-occup} is also a convergent subsequence of $\eoccup_{\estrat}^{T,2}$. Let $\eoccup_{\estrat}$ denote the weak$^*$ limit along $T_n$. The continuity of the projection implies that  $\occup_{\estrat}$ is the marginal of $\eoccup_{\estrat}$ on \Arg, and $\bsplit_{\estrat}\equiv\mathbb{P}_\strat\circ\belief_\infty^{-1}$ is the marginal on \Beliefs. Moreover,  \autoref{eq:step-1} and \autoref{eq:step-2} together imply that $\occup_{\estrat}$ admits the decomposition on the right hand side of \autoref{eq:occup-limit}, and the result follows.

\paragraph{The allocation rule lacks profitable undetectable deviations} We now show the allocation rule \allocr\  admits no profitable undetectable deviations. An undetectable deviation is a transition probability \stratb\ from $\Types\times\Beliefs$ to \Posteriors\ such that for all $\belief\in\Beliefs$ and $\typeb\in\Types$ 
\begin{align}\label{eq:undetectable}
\sum_{\type\in\Types}\typed(\type)\stratb(\typeb|\type,\belief)=\typed(\typeb). 
\end{align}
Consider a deviation by the agent to $(\pp,\stratb)$ instead of \estrat. That is, when his type is \type\ and belief is \belief, the agent chooses type \typeb\ with probability $\stratb(\typeb|\type,\belief)$. In what follows, we index the induced distributions over histories only by \strat\ and \stratb\ as we are only changing the agent's reporting strategy. In particular, denote by $\mathbb{P}_{\stratb}$ the induced probability distribution over  terminal histories when the agent uses $(\pp,\stratb)$ instead of \estrat. 

 We first claim that for every $t$ the marginal of $\mathbb{P}_{\stratb}^t$ over $\States\times\hat H^t$ coincides with that of $\mathbb{P}_\strat^t$. Recall that for every $t$ we have that
\begin{align*}
\mathbb{P}_{\stratb}^{t+1}(\state,h^t,\type,\typeb,\alloc)=\mathbb{P}_{\stratb}^t(\state,h^t)\typed(\type) \stratb(\typeb|\belief_t(h^t),\type)\dmecht(\alloc|\state,\hat{h}^t,\typeb).
\end{align*}
Adding up over \type\ on both sides and using \autoref{eq:undetectable}, we get:
\begin{align*}
\sum_{\type\in\Types}\mathbb{P}_{\stratb}^{t+1}(\state,h^t,\type,\typeb,\alloc)=\mathbb{P}_{\stratb}^t(\state,h^t)\typed(\typeb)\dmecht(\alloc|\state,\hat{h}^t,\typeb).
\end{align*}
Now, note that $h^t=(\hat h^t,\tilde{\type}^{t-1})$ for some sequence $\tilde{\type}^{t-1}\in\Types^{t-1}$. If we add up on both sides over all such sequences we get
\begin{align*}
\sum_{\type\in\Types,\tilde{\type}^{t-1}\in\Types^{t-1}}\mathbb{P}_{\stratb}^{t+1}(\state,\hat{h}^t,\tilde{\type}^{t-1},\type,\typeb,\alloc)=\sum_{\tilde{\type}^{t-1}\in\Types^{t-1}}\mathbb{P}_{\stratb}^t(\state,\hat{h}^t,\tilde{\type}^{t-1})\typed(\typeb)\dmecht(\alloc|\state,\hat{h}^t,\typeb).
\end{align*}
Note that if the distribution over $\States\times\hat{H}^t$ induced by \stratb\ up to period $t$ is the same as that induced by \strat, we get that the right-hand side equals:
\begin{align*}
\mathbb{P}_{\strat,\hat{\terminals}^t}^t(\state,\hat{h}^t)\typed(\typeb)\dmecht(\alloc|\state,\hat{h}^t,\typeb),
\end{align*}
and hence $\mathbb{P}_{\stratb,\hat{\terminals}^{t+1}}^{t+1}(\state,\hat{h}^t,\typeb,\alloc)=\mathbb{P}_{\strat,\hat{\terminals}^t}^t(\state,\hat{h}^t)\typed(\typeb)\dmecht(\alloc|\state,\hat{h}^t,\typeb)=\mathbb{P}_{\strat,\hat{\terminals}^{t+1}}^{t+1}(\state,\hat{h}^t,\typeb,\alloc)$. By definition of $\mathbb{P}_{\stratb}$, we conclude that 
$
\mathbb{P}_{\stratb,\hat{\terminals}^\infty}=\mathbb{P}_{\strat,\hat{\terminals}^\infty}
$. Hence, the joint distribution over states, reports, and allocations is the same under \strat\ and \stratb.

Let $\eoccup_{\stratb}^{T,1},\eoccup_{\stratb}^{T,2}\in\Delta(\Alloc\times\Types\times\hat\Types\times\States\times\Beliefs)$ denote the analogue of the occupation measures in Equations \ref{eq:occup-1} and \ref{eq:occup-2} corresponding to \stratb, extended to account for the agent's reports. Below, the notation $\hat\Types$ signifies those are the agent's reports.  In what follows, recalling that the belief system depends only on the reported history and not the type history is useful. \autoref{eq:undetectable} implies that for all measurable subsets $\tilde\Delta$ of \Beliefs,
\begin{align}\label{eq:occup-11}
&\sum_{\type\in\Types}\eoccup_{\stratb}^{T,1}(\{(\alloc,\type,\typeb,\state)\}\times\tilde\Delta)=\frac{1}{T}\sum_{t=1}^T\sum_{h^t\in H^t}\mathbb{P}_{\stratb}^t(\state,h^t)\typed(\typeb)\dmecht(\state,\hat{h}^t,\typeb)(a)\mathbbm{1}[\belief_t(h^t)\in \tilde\Delta]\nonumber\\
&=\frac{1}{T}\sum_{t=1}^T\sum_{\hat{h}^t\in \hat{H}^t}\mathbb{P}_{\stratb}^t(\state,\hat{h}^t)\typed(\typeb)\dmecht(\state,\hat{h}^t,\typeb)(\alloc)\mathbbm{1}[\belief_t(\hat{h}^t)\in \tilde\Delta]\\
&=\frac{1}{T}\sum_{t=1}^T\sum_{\hat{h}^t\in \hat{H}^t}\mathbb{P}_{\strat}^t(\state,\hat{h}^t)\typed(\typeb)\dmecht(\state,\hat{h}^t,\typeb)(\alloc)\mathbbm{1}[\belief_t(\hat{h}^t)\in \tilde\Delta]=\eoccup_{\strat}^{T,1}(\{(\alloc,\typeb,\state)\}\times\tilde\Delta).\nonumber
\end{align}
The first equality uses the definition of undetectability, the second uses that all the terms depend only on the reported history, the third uses that \strat\ and \stratb\ induce the same distribution over states, reports, and allocations, and the last is the definition of the occupation measure induced by \strat. In words, the marginal of $\eoccup_{\stratb}^{T,1}$ over allocations, reports, states, and beliefs, $\eoccup_{\stratb,\Alloc\hat\Types\States\Delta}^{T,1}$ coincides with $\eoccup_{\strat}^{T,1}$. 

We now show that $\eoccup_{\stratb}^{T,1}$ and $\eoccup_{\stratb}^{T,2}$ have a convergent subsequence with limit $\eoccup_{\stratb}\in\Delta(\Alloc\times\Types\times\hat\Types\times\States\times\Beliefs)$ that admits the following decomposition: \begin{align*}
\mathbb{E}_{\eoccup_{\stratb}}\left[\pay\aarg\right]=\int_{\Beliefs}\left[\sum_{\type\in\Types}\typed(\type)\sum_{\typeb\in\Types}\stratb(\type,\belief)(\typeb)\sum_{\alloc}\allocr(\alloc|\typeb,\belief)\pay(\alloc,\type,\belief)\right]d\bsplit_{\estrat},
\end{align*}
where $\pay(\alloc,\type,\belief)$ is the linear extension of $\pay(\alloc,\type,\cdot)$. %\ld{The above equation should eventually have \bsplit\ and}

We proceed as follows: First, we show that for each $T$, under $\occup_{\stratb}^{T,1}$, the allocation is independent of the true type conditional on the period-$t$ belief and the reported type. Indeed, 
\begin{align*}
\sum_{\state\in\States}\occup_{\stratb}^{T,1}(\alloc,\type,\typeb,\state,\belief)&=\frac{1}{T}\sum_{t=1}^T\sum_{h^t\in H^t}\mathbb{P}_{\stratb}(h^t)\left(\sum_{\state\in\States}\mathbb{P}_{\stratb}(\state|h^t)\dmech_t(\state,\hat h^t,\typeb)(\alloc)\right)\typed(\type)\stratb(\belief_t(h^t),\type)(\typeb)\mathbbm{1}[\belief_t(h^t)=\belief]\\
&=\frac{\typed(\type)\stratb(\type,\belief)(\typeb)}{\typed(\typeb)}\left(\frac{1}{T}\sum_{t=1}^T\sum_{\hat{h}^t:\belief_t(\hat{h}^t)=\belief}\mathbb{P}_{\stratb}^t(\hat{h}^t)\sum_{\state\in\States}\mathbb{P}_{\stratb}(\state|\hat h^t)\typed(\typeb)\dmech_t(\hat{h}^t,\typeb)(\alloc)\right)\\
&=\frac{\typed(\type)\stratb(\type,\belief)(\typeb)}{\typed(\typeb)}\occup_{\stratb,\Alloc\hat\Types\Delta}^{T,1}(\alloc,\typeb,\belief)=\frac{\typed(\type)\stratb(\type,\belief)(\typeb)}{\typed(\typeb)}\occup_{\strat,\Alloc\Types\Delta}^{T,1}(\alloc,\typeb,\belief),
\end{align*}
where the third and fourth equalities use \autoref{eq:occup-1}.  Moreover, the same analysis as that under \strat\ implies the agent's true type is independent of the belief.

Therefore, $\eoccup_{\stratb,\Alloc\Types\hat\Types\Delta}^{T,1}$ admits decomposition:
\begin{align}\label{eq:undetectable-decomposition}
\eoccup_{\stratb,\Alloc\Types\hat\Types\Delta}^{T,1}(\alloc,\type,\typeb,\belief)=\frac{\typed(\type)\stratb(\typeb|\type,\belief)}{\typed(\typeb)}\eoccup_{\stratb,\Alloc\hat\Types\Delta}^{T,1}(\alloc,\typeb,\belief)=\frac{\typed(\type)\stratb(\typeb|\type,\belief)}{\typed(\typeb)}\eoccup_{\strat,\Alloc\Types\Delta}^{T,1}(\alloc,\typeb,\belief),
\end{align}
where the second equality follows from \autoref{eq:occup-11}.

Second, by the same arguments as in \autoref{appendix:repeated},   $\eoccup_{\stratb}^{T,2}(\alloc,\type,\typeb,\state,\belief)$ admits decomposition $\belief(\state)\eoccup_{\stratb,\Alloc\Types\hat\Types\Delta}^{T,2}(\alloc,\type,\typeb,\belief)$ for each $T$.

Third, convergent subsequences $\occup_{\stratb}^{T_{n_m},1}$ and $\occup_{\stratb}^{T_{n_m},2}$ exist with limit $\eoccup_{\stratb}$ (cf. \autoref{prop:extended-occup}).\footnote{\autoref{prop:extended-occup} implies that $\eoccup_{\stratb}^{T,1}$ and $\eoccup_{\stratb}^{T,2}$ have the same set of subsequential limits. Indeed, let $\testf$ denote any continuous bounded function on $\States\times\Beliefs\times\Types\times\Types\times\Alloc$. Let $D_T(\testf)=\mathbb{E}_{\eoccup_{\stratb}^{T,2}}\left[g\right]-\mathbb{E}_{\eoccup_{\stratb}^{T,1}}\left[g\right]=\frac{1}{T}\sum_{t=1}^T\mathbb{E}_{\stratb}\left[g(\alloc_t,\type_t,\typeb_t,\state,\belief_{t+1})-g(\alloc_t,\type_t,\typeb_t,\state,\belief_t)\right]$. The argument in \autoref{prop:extended-occup} implies that $D_T(g)\to 0$ as $T\to\infty$ (this does not rely on the existence of a limit, just the convergence of beliefs and the continuity of $g$). Now, let $T_n$ be such that $\eoccup_{\stratb}^{T_n,1}\weakc\tilde\occup$. Note that
\begin{align*}
\int gd\eoccup_{\stratb}^{T_n,2}=\int gd\eoccup_{\stratb}^{T_n,1}+D_{T_n}(g)\to \int gd\tilde\occup+0,
\end{align*}
so a subsequential limit of $\eoccup_{\stratb}^{T,1}$ is a subsequential limit of $\eoccup_{\stratb}^{T,2}$. Switching the role of 1 and 2, we obtain the opposite set inclusion.} We note two things. On the one hand, because our previous arguments show that the set of measures admitting the above decompositions is closed, the limit $\eoccup_{\stratb}$ admits the decomposition. That is,
\begin{align*}
\int_{\Alloc\times\Types\times\hat\Types\times\States\times\Beliefs} \pay\aarg\eoccup_{\stratb}(d(\alloc,\type,\typeb,\state,\belief))=\int_{\Alloc\times\hat\Types\times\Beliefs}\sum_{\type\in\Types}\frac{\typed(\type)\stratb(\typeb|\type,\belief)}{\typed(\typeb)}\left(\sum_{\state\in\States}\belief(\state) \pay\aarg\right)d\eoccup_{\stratb,\Alloc\hat\Types\Delta}.
\end{align*}
On the other hand, because $T_{n_m}$ is a subsequence of $T_n$ \emph{and} $\eoccup_{\stratb,\Alloc\hat\Types\States\Delta}^{T,1}=\eoccup_{\strat,\Alloc\Types\States\Delta}^{T,1}$ \emph{and} $\eoccup_{\strat,\Alloc\Types\States\Delta}^{T_n,1}\weakc\eoccup_\strat$, we can conclude that $\eoccup_{\stratb,\Alloc\hat\Types\Delta}=\eoccup_{\strat,\Alloc\Types\Delta}$ and admits the same decomposition as $\bar\occup_\strat$. We conclude that 
\begin{align*}
\mathbb{E}_{\eoccup_{\stratb}}\left[ \pay\aarg\right]=\int_{\Beliefs}\left[\sum_{\type,\typeb\in\Types}\typed(\type)\stratb(\typeb|\type,\belief)\sum_{\alloc}\allocr(\alloc|\typeb,\belief)\left(\sum_{\state\in\States}\belief(\state) \pay\aarg\right)\right]d\bsplit_{\estrat}.
\end{align*}

Consequently, 
\begin{align*}
\lim\sup_{T\to\infty}\mathbb{E}_{\stratb}\left[U_T\right]\geq \lim_{m\to\infty}\mathbb{E}_{\stratb}\left[U_{T_{n_m}}\right]&=\mathbb{E}_{\eoccup_{\stratb}}\left[\pay\aarg\right]\\
&=\mathbb{E}_{\bsplit_{\estrat}}\left[\sum_{\type,\typeb,\alloc}\typed(\type)\stratb(\type,\belief)(\typeb)\allocr(\alloc|\typeb,\belief)\pay(\alloc,\type,\belief)\right],
\end{align*}
where $\pay(\alloc,\type,\belief)$ is the linear extension of $\pay(\alloc,\type,\cdot)$. Because \strat\ is a best response, we have that 
\begin{align*}
\mathbb{E}_{\bsplit_{\estrat}}\left[\sum_{\type,\alloc}\typed(\type)\allocr(\alloc|\type,\belief)\pay(\alloc,\type,\belief)\right]\geq\mathbb{E}_{\bsplit_{\estrat}}\left[\sum_{\type,\typeb,\alloc}\typed(\type)\stratb(\type,\belief)(\typeb)\allocr(\alloc|\typeb,\belief)\pay(\alloc,\type,\belief)\right],
\end{align*}
which implies the two-stage mechanism lacks profitable undetectable deviations.
\paragraph{The allocation rule is ex ante individually rational} Define 
\begin{align*}
U_{\mathrm{net}}(\belief)=\sum_{\type\in\Types}\typed(\type)\sum_{\state\in\States}\belief(\state)\left[\sum_{\alloc\in\Alloc}\allocr(\alloc|\type,\belief)\pay\aarg-\pay(\oo,\type,\state)\right],
\end{align*}
to be the agent's (ex ante) payoff net of the outside option at belief \belief. Ex ante individual rationality of \allocr\ is equivalent to $U_{\mathrm{net}}(\belief)\geq0$ for all $\belief$ in the support of $\bsplit_{\estrat}$.%, where recall $\estrat$ is the strategy under which the agent always participates and truthfully reports.

Toward a contradiction, assume that $U_{\mathrm{net}}(\belief)<0$ with positive probability under $\bsplit_{\estrat}$. By \autoref{prop:ball} in \autoref{appendix:aux}, a set $B\subset\Beliefs$ open relative to \Beliefs\ exists such that\footnote{A set $X\subseteq\Delta(\States)$ is open relative to \Beliefs\ if an open set $Y\subseteq\reals^{|\States|}$ exists such that $X=Y\cap\Beliefs$. The boundary relative to $\Delta(\States)$ is analogously defined via  open sets relative to \Beliefs.}
\begin{align*}
\int_BU_{\mathrm{net}}(\belief)\bsplit_{\estrat}(d\belief)<0.
\end{align*}
Moreover, we can pick $B$ such that $\bsplit_{\estrat}(\partial B)=0$, where $\partial B $ denotes the boundary of $B$ relative to $\Delta(\States)$.\footnote{\autoref{prop:ball} provides an interval of radii $r\in(0,r_0)$ such that $\int_{B(\hat\belief,r)}U_{\mathrm{net}}(\belief)\bsplit_{\estrat}(d\belief)<0$. Note that only countable many such $r$ can have $\bsplit_{\estrat}(\partial B(\hat\belief,r))>0$ (the boundaries for different radii are disjoint), so we can always pick $r$ such that $\bsplit_{\estrat}(\partial B(\hat\belief,r))=0$ and preserve the negative sign.} Lastly, let $\delta>0$ be such that 
\begin{align}\label{eq:step-0}
\int_BU_{\mathrm{net}}(\belief)\bsplit_{\estrat}(d\belief)\leq-2\delta.
\end{align}

Let $(\belief_t(h^t))_{t\in\naturals,h^t\in H^t}$ denote the belief process under \estrat. For $L\in\naturals$, define a strategy $(p^L,\strat^L)$ as follows:
\begin{enumerate}
\item $(p_t^L(h^t,\cdot),\strat_t^L(h^t,\cdot))=(p_t(h^t,\cdot),\strat_t(h^t,\cdot))$ if either $t<L$ OR ($t\geq L$ and $\belief_L(h^L)\notin B$), where $h^L$ precedes $h^t$,
\item Otherwise, $(p_t^L(h^t,\cdot),\strat_t^L(h^t,\cdot))=(0,\strat_t(h^t,\cdot))$ (note that when the agent quits the strategy can be specified arbitrarily.)
\end{enumerate}

Note the agent's average payoff through period $T$ under $(p^L,\strat^L)$ can be written as follows:
\begin{align*}
\mathbb{E}_{(p^L,\strat^L)}\left[U_T\right]=\mathbb{E}_{\estrat}\left[U_T\right]-\mathbb{E}_{\estrat}\left[\frac{1}{T}\sum_{t=1}^T\left(\pay(\alloc_t,\type_t,\state)-\pay(\oo,\type_t,\state)\right)\mathbbm{1}[t\geq L \text{ and } \belief_L\in B]\right].
\end{align*}

We show that for sufficiently large $L$, $(p^L,\strat^L)$ is a profitable deviation. For $T\geq L$, write
\begin{align}\label{eq:nn}
&\mathbbm{E}_{\estrat}\left[\frac{1}{T}\sum_{t=L}^T\left(\pay(\alloc_t,\type_t,\state)-\pay(\oo,\type_t,\state)\right)\mathbbm{1}[\belief_L\in B]\right]-\int_BU_{\mathrm{net}}(\belief)\bsplit_{\estrat}(d\belief)=\\
&=\mathbbm{E}_{\estrat}\left[\frac{1}{T}\sum_{t=1}^T\left(\pay(\alloc_t,\type_t,\state)-\pay(\oo,\type_t,\state)\right)\mathbbm{1}[\belief_t\in B]\right]-\int_BU_{\mathrm{net}}(\belief)\bsplit_{\estrat}(d\belief)\nonumber\\
&-\mathbbm{E}_{\estrat}\left[\frac{1}{T}\sum_{t=1}^{L-1}\left(\pay(\alloc_t,\type_t,\state)-\pay(\oo,\type_t,\state)\right)\mathbbm{1}[\belief_t\in B]\right]\nonumber\\
&+\mathbbm{E}_{\estrat}\left[\frac{1}{T}\sum_{t=L}^T\left(\pay(\alloc_t,\type_t,\state)-\pay(\oo,\type_t,\state)\right)\left(\mathbbm{1}[\belief_L\in B]-\mathbbm{1}[\belief_t\in B]\right)\right].\nonumber
\end{align}

Let $K=\max_{\type,\state,\alloc}|\left(\pay(\alloc,\type,\state)-\pay(\oo,\type,\state)\right)|$, and note that we can bound the term in the last line of \autoref{eq:nn} as follows:
\begin{align*}
\mid\mathbbm{E}_{\estrat}\left[\frac{1}{T}\sum_{t=L}^T\left(\pay(\alloc_t,\type_t,\state)-\pay(\oo,\type_t,\state)\right)\left(\mathbbm{1}[\belief_L\in B]-\mathbbm{1}[\belief_t\in B]\right)\right]\mid\leq K\mathbbm{E}_{\estrat}\left[\frac{1}{T}\sum_{t=L}^T|\mathbbm{1}[\belief_L\in B]-\mathbbm{1}[\belief_t\in B]|.\right]
\end{align*}
Because $\belief_t\weakc \belief_\infty$ $\pr_{\estrat} \text{-a.s.}$ (\autoref{lemma:beliefs-0}) and $\bsplit_{\estrat}(\partial B)=0$, we conclude:\footnote{The property that $\bsplit_{\estrat}(\partial B)=0$ ensures that $\mathbbm{1}[\belief_t(h^\infty)\in B]$ is eventually constant almost surely. Let $E=\{h^\infty:\belief_\infty(h^\infty)\notin \partial B\}$. On $E$, either $\belief_\infty$ is in the  interior of $B$ (relative to $\Delta(\States)$) or in the  interior of $B^\complement$  (relative to $\Delta(\States)$), that is, an $\epsilon>0$ exists such that $(B(\belief_\infty,\epsilon)\cap\Delta(\States))\subset B$ or $B^\complement$. In either case, for each $h^\infty$, there exists $N(h^\infty)$ such that for all $t\geq N(h^\infty)$, $\belief_t(h^\infty)\in (B(\belief_\infty,\epsilon)\cap\Delta(\States))$ and hence $\mathbbm{1}[\belief_t(h^\infty)\in B]$ is eventually constant. When $\bsplit_{\estrat}(\partial B)=0$, we have that $E$ has probability 1 under $\pr_{\estrat}$.}
\begin{align*}
\lim_{L\to\infty}\sup_{t\geq L}|\mathbbm{1}[\belief_L\in B]-\mathbbm{1}[\belief_t\in B]|=0\;\; \pr_{\estrat} \text{-a.s.}
\end{align*}
Then, 
\begin{align*}
\mathbbm{E}_{\estrat}\left[\frac{1}{T}\sum_{t=L}^T|\mathbbm{1}[\belief_L\in B]-\mathbbm{1}[\belief_t\in B]|\right]\leq \mathbb{E}_{\estrat}\left[\sup_{t\geq L}|\mathbbm{1}[\belief_L\in B]-\mathbbm{1}[\belief_t\in B]|\right]
\end{align*}
and choose $\bar L$ large enough so that for all $L\geq \bar L$, we have that:
\begin{align}\label{eq:step-2-end}
 \mathbb{E}_{\estrat}\left[\sup_{t\geq L}|\mathbbm{1}[\belief_L\in B]-\mathbbm{1}[\belief_t\in B]|\right]\leq \delta/K.
\end{align}
Consider now the term in the third line of \autoref{eq:nn} and note that it is bounded in absolute value by $K(L-1)/T$, which tends to $0$ as $T\to \infty$. Similarly, the term in the second line of \autoref{eq:nn} vanishes as $T\to \infty$.\footnote{Indeed, $\mathbbm{E}_{\estrat}\left[\frac{1}{T}\sum_{t=1}^T\left(\pay(\alloc_t,\type_t,\state)-\pay(\oo,\type_t,\state)\right)\mathbbm{1}[\belief_t\in B]\right]=\mathbbm{E}_{\eoccup_{\estrat}^{T,1}}\left[\left(\pay(\alloc,\type,\state)-\pay(\oo,\type,\state)\right)\mathbbm{1}[\belief\in B]\right]$ and our previous analysis implies it converges to $\int_BU_{\mathrm{net}}(\belief)\bsplit_{\estrat}(d\belief)$.} Thus, for $L\geq \bar L$, we can find $\overline T$ such that for all $T\geq \overline T$\footnote{Choose $\overline T$ so that for all $T\geq \overline T$:
\begin{align*}
&|\mathbbm{E}_{\estrat}\left[\frac{1}{T}\sum_{t=1}^{T}\left(\pay(\alloc_t,\type_t,\state)-\pay(\oo,\type_t,\state)\right)\mathbbm{1}[\belief_t\in B]\right]-\int_BU_{\mathrm{net}}(\belief)\bsplit_{\estrat}(d\belief)|\\
&+|\mathbbm{E}\left[-\frac{1}{T}\sum_{t=1}^{L-1}\left(\pay(\alloc_t,\type_t,\state)-\pay(\oo,\type_t,\state)\right)\mathbbm{1}[\belief_t\in B]\right]|\leq \delta/2.
\end{align*}} 
\begin{align}\label{eq:yy}
\mid \mathbbm{E}_{\estrat}\left[\frac{1}{T}\sum_{t=L}^T\left(\pay(\alloc_t,\type_t,\state)-\pay(\oo,\type_t,\state)\right)\mathbbm{1}[\belief_L\in B]\right]-\int_BU_{\mathrm{net}}(\belief)\bsplit_{\estrat}(d\belief)\mid \leq \frac{3}{2}\delta,
\end{align}
and hence 
\begin{align}
\mathbbm{E}_{\estrat}\left[\frac{1}{T}\sum_{t=L}^T\left(\pay(\alloc_t,\type_t,\state)-\pay(\oo,\type_t,\state)\right)\mathbbm{1}[\belief_L\in B]\right]\leq-\frac{1}{2}\delta.
\end{align}
We conclude that
\begin{align*}
\lim\sup_{T\to\infty}\mathbb{E}_{(p^L,\strat^L)}\left[U_T\right]\geq \lim_{T\to\infty}\mathbb{E}_{\estrat}[U_T]+\frac{1}{2}\delta,
\end{align*}
a contradiction.

\subsubsection{Proof of \autoref{theorem:dynamic} (sufficiency)} 
We now show that all outcome distributions $\outcome\in\Delta(\Arg)$ that admit the decomposition in \autoref{theorem:dynamic} can be implemented via a dynamic mechanism. To this end, let $\bsplit$ and $\allocr:\Types\times\Beliefs\to\Delta(\Alloc)$ denote the belief distribution and the ex ante individually rational allocation rule without profitable undetectable deviations corresponding to \outcome.  That is,
\begin{align}\label{eq:dynamic-suff}
\outcome\aarg=\int_{\Beliefs}\belief(\state)\typed(\type)\allocr(\alloc|\type,\belief)\bsplit(d\belief).
\end{align}

 The proof proceeds as follows:
\begin{enumerate}
\item\label{itm:dyn-imp} We first consider a fictitious setting in which there is no state uncertainty and we are given an allocation rule $\allocrb:\Types\to\Delta(\Alloc)$ that is ex ante individually rational and lacks profitable undetectable deviations for some utility function $\payb:\Alloc\times\Types\to\reals$. \autoref{prop:rahman-2}  shows that a dynamic mechanism exists that implements $\allocrb$. 

\item\label{itm:fin-sup} We then show that if \outcome\ satisfies \autoref{eq:dynamic-suff}, then a finite support belief distribution \bsplitb\ exists such that \outcome\ and \allocr\ satisfies \autoref{eq:dynamic-suff} with \bsplitb\ instead of \bsplit.

\item\label{itm:end} Lastly, we use this result to construct a dynamic game that implements \outcome.

\end{enumerate}

\paragraph{Step \ref{itm:dyn-imp}} For this step, we consider a fictitious setting in which there is no state uncertainty and the designer faces a privately informed agent with payoffs $\payb:\Alloc\times\Types\to\reals$, where $\type\sim\typed\in\Delta(\Types)$.\footnote{Anticipating our construction in \autoref{itm:end}, for each belief $\belief$ in the support of \bsplit, the allocation rule $\allocrb(\cdot|\cdot,\belief)$ has no profitable undetectable deviations relative to payoff function $\payb(\alloc,\type)=\sum_{\state\in\States}\belief(\state)\pay\aarg$.}

Suppose we are given an allocation rule $\allocrb:\Types\to\Delta(\Alloc)$ that admits no profitable undetectable deviations relative to \payb\ as in \autoref{def:rahman-ic} and is individually rational as in \autoref{def:dynamic-ir}.
We have the following result:
\begin{proposition}\label{prop:rahman-2}
Let $\outcomeb=\typed(\type)\allocrb(\alloc|\type)\in\Delta(\Alloc\times\Types)$ such that $\allocrb$ lacks profitable undetectable deviations and is ex ante individually rational. Then, a dynamic mechanism exists that implements $\outcomeb$.
\end{proposition}
\begin{proof}[Proof of \autoref{prop:rahman-2}]
The proof is constructive. We build on the analysis of  \cite{margaria2018dynamic} and present a dynamic mechanism that alternates between communication and adjustment phases. In all phases, the mechanism selects allocations using reports \typeb\ according to \allocrb. In a communication phase, the reports are those sent by the agent. In an adjustment phase, the agent's reports are disregarded; instead, the mechanism simulates reports to guarantee that the occupation measure over reports coincides with \typed\ and these simulated reports are used to determine the allocation. The mechanism ensures that under any agent's strategy, the occupation measure over reports and allocations exists and equals $\outcomeb$; thus, any strategy corresponds to an undetectable deviation. The length of communication phases grows in time. Thus,  under truthtelling the relative length of adjustment phases vanishes in time, and the expected occupation measure over types and allocations exists and equals $\outcomeb$. Because $\allocrb$ lacks profitable undetectable deviations, it follows that truthtelling is optimal for the agent. Because $\allocrb$ is ex ante individually rational, it follows that the participation constraints are satisfied.

Formally, the mechanism consists of sequential blocks, each block starting with a communication phase followed by an adjustment phase. The lengths of communication phases are fixed at $L_1,L_2,\dots$ such that $L_n\to\infty$ and $L_n/\sum_{k\leq n} L_k\to0$, e.g., $L_n=n$. The length of  adjustment phase $N_n$ depends on the agent's reports in the communication phase in block $n$.  Denote by $T_{n}$ the first period of  block $n$, which is the first period of the corresponding communication phase. The first period of the corresponding adjustment phase  is $T_{n}+L_n+1$. Denote by $\freq_n^1$ the average report frequencies in this block at the beginning of the adjustment stage:
\begin{align}
\freq_n^1(\typeh)\triangleq\frac{1}{L_n}\sum_{t=T_n}^{T_n+L_n-1} 1(\typeh_t=\typeh).
\end{align}
If $\freq_n^1=\typed$, then the adjustment phase is empty, and the mechanism proceeds to the next block. Otherwise, in the adjustment phase, the mechanism generates reports over $N_n$ periods to guarantee that at the end of the adjustment phase the \emph{expected} frequency of reports in this block equals  \typed\, that is,
\begin{align}\label{eq:adjustment_target}
\mathbb{E}[\freq_n^2|\freq_n^1]=\typed
\end{align}
where 
\begin{align}
\freq_n^2(\typeh)\triangleq\frac{1}{L_n+N_n}\sum_{t=T_n}^{T_n+L_n+N_n-1} 1(\typeh_t=\typeh).
\end{align}

To do so, denote by $\eta\triangleq\min_\type \typed(\type)$ and observe that $f\in\Delta(\Types)$ can be surrounded by a ball of radius $\eta$ within the simplex $\Delta(\Types)$.  The adjustment phase lasts for $N_n$ periods where:\footnote{Any $||.||_p$ would work, but larger $p$ results in weakly shorter adjustment phases.}
\begin{align}
    N_n=\left\lceil L_n\frac{||\freq_n^1-\typed||_\infty}{\eta}\right\rceil,
\end{align}
and in each period of the adjustment phase the mechanism generates the reports i.i.d. according to $\tilde f_n^a$:
\begin{align}
    f^a_n=f-(\freq_n^1-f)\frac{L_n}{N_n}.
\end{align}
The construction ensures that $f^a_n\in\Delta(\Types)$, because $||f^a_n-f||_\infty\leq\eta$, and that (\ref{eq:adjustment_target}) holds, because
\begin{align*}
   \mathbb{E}[\freq_n^2|\freq_n^1]=\frac{1}{L_n+N_n}(L_n \freq_n^1+N_n f^a_n)=f.
\end{align*}

This in turn guarantees that the long-run   distribution of reports  (generated jointly by the agent and the mechanism) exists and equals \typed\ irrespectively of the agent's strategy. Intuitively, the fact that each block becomes negligible relative to past history over time ensures the agent's reports in each block have less and less effect on the long run frequency of reports, whereas the adjustment phase ensures that the frequency of reports converges to \typed. Formally, for any history and $T$ denote by $\nlast(T)$ the number of the block to which $T$ belongs and by $\Tlast(T)$ the first period of that block. Observe that for any agent's strategy:
\begin{align}
    N_n\leq L_n \left(\max_{f'} \frac{||f'-f||_\infty}{\eta}+1\right)\triangleq L_n \overline\rho. 
\end{align}
Therefore, 
\begin{align}\label{eq:adjustment-bound}
    \frac{|T-\Tlast(T)|}{\Tlast(T)}\leq \frac{L_{\nlast(T)}(1+\overline\rho)}{\sum_{k<\nlast(T)} L_k}\xrightarrow[T\to\infty]{a.s.}0,
\end{align}
where the limit result holds because $\nlast(T)\xrightarrow[T\to\infty]{a.s.}\infty$ and $L_n/\sum_{k\leq n} L_k\xrightarrow[n\to\infty]{}0$.

Then, for any agent's strategy, for any $\typeh\in\Types$,
\begin{align}
\lim_{T\to\infty}\frac{1}{T}\sum_{t=1}^T \Pr(\typeh_t=\typeh)&=\lim_{T\to\infty} \mathbb{E} \left[\frac{\typed(\typeh)\Tlast(T)+\flast(T)(T-\Tlast(T))}{\Tlast(T)+T-\Tlast(T)}\right]\\
    &=\lim_{T\to\infty} \mathbb{E} \left[\frac{\typed(\typeh)+\flast(T)(T-\Tlast(T)/\Tlast(T)}{1+(T-\Tlast(T))/\Tlast(T)}\right]=\typed(\typeh),\nonumber
\end{align}
where $\flast(T)\in\Delta(\Types)$ is the report frequency in the last block up to period $T$, and the last line follows from \autoref{eq:adjustment-bound}.

Since the mechanism chooses allocations in all periods according to $\allocrb$, it follows that for any agent's strategy \stratb\ the induced occupation measure over allocations and type reports satisfies:%\annotation{I am confused about whether this is the occupation measure of reports or types; i think the former.}
\begin{align}\label{eq:deviation-decomposition}
\lim_{T\to\infty}\frac{1}{T}\mathbb{E}_{\stratb}\left[\sum_{t=1}^T\mathbbm{1}[(\alloc_t,\typeh_t)=(\alloc,\typeh)]\right]=\typed(\typeh)\allocrb(\alloc|\typeh)=\outcomeb(\alloc,\typeh).
\end{align}
In other words, for any reporting strategy the occupation measure over allocations and reports exists.

We now show that under truthtelling the occupation measure over types and allocations exists and equals $\typed(\type)\allocrb(\alloc|\type)=\outcome(\alloc,\type)$. To this end, assume that the agent always reports her true type. For any $T$, denote by $\Ltotal(T)$ the total number of periods spent in  communication phases before $T$ and by $\Ntotal(T)$ the total number of periods spent in adjustment phases before $T$. Observe that by the strong law of large numbers, because $L_n\to\infty$,
\begin{align*}
    \frac{N_n}{L_n}\leq \frac{||\freq^1_n-f||_\infty}{\eta}+\frac{1}{L_n}\xrightarrow[n\to\infty]{a.s.}0.
\end{align*}
Therefore,
\begin{align*}
    \frac{\Ntotal(T)}{\Ntotal(T)+\Ltotal(T)}\xrightarrow[T\to\infty]{a.s.}0,
\end{align*}
because whenever  $N_n/L_n\to 0$, $\lim_{T\to\infty} N^{\mathrm{total}(T)}/(N^{\mathrm{total}}(T)+L^{\mathrm{total}}(T))=\lim_{n\to\infty} N_n/(L_n+N_n)=0.$

It follows that
\begin{align*}
&\lim_{T\to\infty}\frac{1}{T}\sum_{t=1}^T \Pr((\type_t,\typeh_t,\alloc_t)=(\type,\typeh,\alloc))\\
&=\lim_{T\to\infty}\mathbb{E}\left[\frac{\Ltotal(T)1(\type=\typeh)\typed(\typeh)\allocrb(a|\typeh)+\Ntotal(T) \typed^\mathrm{adj}(T)(\type,\typeh,\alloc)}{\Ntotal(T)+\Ltotal(T)}\right]\\
&=1(\type=\typeh)\typed(\typeh)\allocrb(\alloc|\typeh),
\end{align*}
where $f^\mathrm{adj}(T)\in\Delta(\Types\times\Types\times A)$ is the average frequency of types, reports, and allocations in the adjustment phases before $T$. Therefore, under truthtelling, the occupation measure over allocations and types equals
\begin{align}
\outcomeb(\alloc,\type)=    \typed(\type)\allocrb(a|\type).
\end{align}
Hence, the agent's payoff  in the dynamic mechanism under truthtelling is:
\begin{align}
    \Utruth=\sum_{(\alloc,\type)} \typed(\type)\allocrb(\alloc|\type)\payb(\alloc,\type).
\end{align}
It remains to show that the agent cannot achieve more than $\Utruth$ under any other strategy. To this end, fix and alternative strategy $\strat$, and denote by $U(\strat)=\limsup_{T\to\infty}U_T(\strat)$ where:
\begin{align*}
   U_T(\strat)=\frac{1}{T}\sum_{t=1}^T
\sum_{\alloc,\type}\Pr((\alloc_t,\type_t)=(\alloc,\type))\payb(\alloc,\type).
\end{align*}
Consider any convergent subsequence $(U_{T_n})_{n=1}^\infty$ along times $\{T_n\}_{n=1}^\infty$. Because $\Delta(\Alloc\times\Types\times\Types)$ is compact \citep[Theorem 15.11]{aliprantis2006infinite}, a convergent (sub)subsequence at times $\{T_k\}_{k=1}^\infty\subseteq \{T_n\}_{n=1}^\infty$ exists along which the occupation measure induced by \strat\ 
\begin{align*}
\occup_\strat^{T_k}&\weakc\occup_\strat,
\end{align*}
for some $\occup_\strat\in\Delta(\Alloc\times\Types\times\Types)$, which by (\ref{eq:deviation-decomposition}) satisfies $\occup_\strat(\alloc,\typeh)=\typed(\typeh)\allocrb(\alloc|\typeh)$. It follows that for some undetectable deviation $\occup_\strat(\typeh|\type)$:%\annotation{I am confused by the first equality?  }
\begin{align*}
    \lim_{n\to\infty} U_{T_n}=\lim_{k\to\infty} U_{T_k}=\sum_{\type,\typeh,a} \typed(\type)\occup_\strat(\typeh|\type)\allocrb(\alloc|\typeh)\payb(\alloc,\type)\leq \Utruth,
\end{align*}
where the inequality follows because  $\allocrb(a|\typeh)$ lacks profitable undetectable deviations. Because this inequality holds for any convergent subsequence $(U_{T_n})_{n=1}^\infty$,
\begin{align*}
   U(\strat)=\limsup_{T\to\infty}U_T(\strat)\leq\Utruth.
\end{align*}
Finally, observe that the construction ensures that after every history, truthtelling from there on delivers the continuation payoff $\Utruth$. Since $\allocrb$ is ex ante individually rational, $\Utruth\geq\sum_{\type} \typed(\type)\payb(\oo,\type)$, and thus the participation constraints are satisfied. This concludes the proof.
\end{proof}

\paragraph{Step \ref{itm:fin-sup}} Consider now the outcome distribution $\outcome\in\Delta(\Arg)$ satisfying \autoref{eq:dynamic-suff}. As we argue in the proof of \autoref{theorem:cmd-two-stage}, 
 a finite $K\leq|\Alloc||\Types||\States|$, $\{\belief_1,\dots,\belief_K\}\in\Beliefs$, and $\bsplitb\in\Delta(\Beliefs)$ exists such that
\begin{align*}
\outcome\aarg=\typed(\type)\sum_{k=1}^K\bsplitb(\belief_k)\belief_k(\state)\allocr(\alloc|\type,\belief_k).
\end{align*}

\paragraph{Step \ref{itm:end}} We now use steps \ref{itm:dyn-imp} and \ref{itm:fin-sup} to complete the proof of \autoref{theorem:dynamic}, so in what follows we use the finite support representation of \outcome\ in the previous step. By Bayes plausibility, a dynamic mechanism can generate the belief split $\bsplitb$ in $T$ periods with $T\leq \lceil\log_{|\Alloc|}(|\States||\Types||\Alloc|)\rceil$, by treating each sequence of allocations of length $T$ as a message. This can be achieved by making the mechanism constant on the agent's type reports during the first $T$ periods. Since each $\allocr(\cdot|\cdot,\belief_k)$ for $k\in\{1,\dots,K\}$ lacks profitable undetectable deviations and is individually rational,  \autoref{prop:rahman-2} implies that a dynamic mechanism exists that implements \outcome\ by first generating the belief split \bsplitb\ and then implementing $\allocr(\cdot|\cdot,\belief_k)$ in the corresponding continuation play. 

\section{Proof of auxiliary results}\label{appendix:aux}
\subsection{Revelation principle for calibrated mechanism design}\label{appendix:model}
In the main text, we restricted attention to incentive compatible and individually rational calibrated mechanisms. We show in this appendix that this restriction is without loss of generality by considering mechanisms with arbitrary message spaces and participation and reporting decisions by the agents that constitute an equilibrium of the game induced by the mechanism and its calibrated information structure.

\paragraph{Mechanisms} Let $2^{[N]}\setminus\emptyset$ denote the nonempty subsets of agents. Then, we can define a mechanism as a collection $\left\{\left(\mssgs_J,\mechanism_J\right):J\in 2^{[N]}\setminus\emptyset\right\}$, where
\begin{align*}
    \mechanism_J:\mssgs_J\times\States\times[0,1]\to\Delta(\Alloc_J),
\end{align*}
is the mechanism when agents in $J$ participate, where $\mssgs_J=\times_{i\in J}\mssgs_i$ and $\Alloc_J=\times_{i\in J}\Alloci$.

\paragraph{Information Structure} Let $\asignals_i=\Delta(\Alloci)^{\mssgs_i}$ denote the collection of menus of lotteries with labels $\mssgs_i$, and let $\asignals=\times_{i\in[N]}\asignals_i$. An information structure is $(\excdf,\asignals)$, where  $\excdf:\States\times[0,1]\to\asignals$. 

\paragraph{Participation and reporting strategies} It is notationally convenient to allow each agent to have her own randomization device $\cori\sim U[0,1]$ and write agents' strategies as mappings  $(p_i,\strat_i):\Typesi\times\asignals_i\times[0,1]\to\{0,1\}\times \mssgs_i$, where $p_i$ denotes agent $i$'s participation decision, and $\strat_i$ her reporting strategy, conditional on participating. To distinguish the agents' randomization from that of the original mechanism, we reserve $\cor_0$ for the realization of the mechanism's randomization device.

 Given $(p_i,\strat_i)_{i\in[N]}$ and a mechanism $(\mechanism,\mssgs)$, fix a profile $(\type,\asignal,\overline{\cor})\equiv(\typei,\asignal,\cori)_{i\in[N]}$. This determines a set of agents that participate, 
\[J(\type,\asignal,\overline{\cor})=\{j\in[N]:p_j(\type_j,\asignal_j,\cor_j)=1\},\]
and let $J_{-i}(\type,\asignal,\overline{\cor})$ denote the projection of $J(\type,\asignal,\overline{\cor})$ on $J\setminus\{i\}$. Note that $J_{-i}$ only depends on $(\typemi,\asignal_{-i},\overline{\cor}_{-i})$. Lastly, write $\rmech_{J_{-i}(\typemi,\asignal_{-i},\cor_{-i})\cup\{i\}}(\mssg_i, \strat_{J_{-i}(\typemi,\asignal_{-i},\cor_{-i})},\state,\cor_0)\in\Delta(\Alloc_{J_{-i}(\typemi,\asignal_{-i},\cor_{-i})\cup\{i\}})$ for
\begin{align*}
\sum_{\mssg_{J_{-i}(\typemi,\asignal_{-i},\cor_{-i})}}\left(\prod_{j\in J_{-i}(\typemi,\asignal_{-i},\cor_{-i})}\strat_j(\type_j,\asignal_j)(m_j)\right)\mechanism_{J_{-i}(\typemi,\asignal_{-i},\cor_{-i})\cup\{i\}}(m_i,\mssg_{J_{-i}(\typemi,\asignal_{-i},\cor_{-i})},\state,\cor_0)
\end{align*}

\paragraph{Calibrated information structures} Given $(p_i,\strat_i)_{i\in[N]}$ and a mechanism $(\mechanism,\mssgs)$, the information structure $\left(\excdf, \asignals\right)$ is calibrated with the mechanism and the agents' strategies if whenever $\pi(\state,\cor_0)=(\asignal_1,\dots,\asignal_N)$, then for all $i,\mssg_i$
\[\asignal_i(\cdot|\mssg_i)=\mathbb{E}_{\tilde{\type}_{-i}\sim \typedmi(\cdot|\state),\epsilon_{-i}}\left[\sum_{\allocmi\in\Allocmi}\mechanism_{J_{-i}(\typemi,\asignal_{-i},\cor_{-i})\cup\{i\}}(m_i,\strat_{J_{-i}(\typemi,\asignal_{-i},\cor_{-i})},\state,\epsilon_0)(\cdot,\allocmi)\right].\]
Below, to keep the presentation simple, we focus on the case in which the calibrated information structure has finite support.

\paragraph{Equilibrium} Given $(p_i,\strat_i)_{i\in[N]}$, a mechanism $(\mechanism,\mssgs)$ and an information structure $\left(\excdf, \asignals\right)$ calibrated with the mechanism and the agents' strategies, $(p_i,\strat_i)_{i\in[N]}$ is an equilibrium if for all $i\in[N]$, all $\typei\in\Typesi$, all $\asignal_{i}\in\asignals_i$, and $\cor_i\in[0,1]$, the following hold:
\begin{align*}
&\strat_i(\typei,\asignal_i,\cor_i)\in\arg\max_{\mssg_i\in\mssgs_i}\sum_{\alloci\in\Alloci}\asignal_i(\alloci|\mssg_i)\mathbb{E}_{\state\sim\belief_i(\cdot|\typei,\asignal_i)}\left[\payi(\alloci,\typei,\state)\right],\\
&p_i(\typei,\asignal_i,\cor_i)\in\arg\max_{p\in\{0,1\}}p\sum_{\alloci\in\Alloci}\asignal_i(\alloci|\strat_i(\typei,\asignal_i,\cor_i))\mathbb{E}_{\state\sim\belief_i(\cdot|\typei,\asignal_i)}\left[\payi(\alloci,\typei,\state)\right]+(1-p)\mathbb{E}_{\state\sim\belief_i(\cdot|\typei,\asignal_i)}\left[\payi(\alloc_{i\emptyset},\typei,\state)\right],
\end{align*}
where $\belief_i(\typei,\asignal_i)\in\Beliefs$ denotes agent $i$'s updated beliefs about the state when her type is \typei\ conditional on receiving signal $\asignal_i$.

\paragraph{Revelation Principle} Fix $(p_i,\strat_i)_{i\in[N]}$, a mechanism $(\mechanism,\mssgs)$ and an information structure $(\excdf, \asignals)$ calibrated with the mechanism such that  $(p_i,\strat_i)_{i\in[N]}$ is an equilibrium. We construct a direct mechanism $(\mechanism^*,\Types)$ and a calibrated information structure $\left(\excdf^*,\csignals\right)$ calibrated with the mechanism under truthtelling and full participation such that truthtelling and full participation is an equilibrium.

First, note that we can extend each $\mechanism_J(\cdot)\in \Delta(\Alloc_J)$ to a mechanism $\overline{\mechanism}_J(\cdot)\in\Delta(\Alloc)$ as follows: for all $\mssg\in \mssgs_J,\state\in\States,\cor_0\in[0,1],$ and $\alloc_J\in\Alloc_J$, 
\[\overline{\mechanism}_{J}(\mssg_J,\state,\cor_0)(\alloc)=\mechanism_J(\mssg_J,\state,\cor_0)(\alloc_J)\times\delta_{\alloc_{-J,\emptyset}}.\]
Define a ``pseudo''-mechanism as follows:
\begin{align*}
    \hat\mechanism_N\left(\type,\state,\cor_0,\bar\cor\right)=\overline{\mechanism}_{J(\type,\excdf(\state,\cor_0),\bar\cor)}\left(\strat_{J(\type,\asignal,\overline{\cor})},\state,\cor_0\right).
\end{align*}
where $\strat_{J(\type,\asignal,\bar\cor)}$ is the message vector generated by the strategies. Define the full participation mechanism $\mechanism_N^*:\Types\times\States\times[0,1]\mapsto\Delta(\Alloc)$ to be
\begin{align*}
    \mechanism_N^*(\type,\state,\cor_0)(\alloc)=\int_{[0,1]^N}\hat\mechanism_N(\type,\state,\cor_0,\bar\cor)(\alloc)\lebesgue^N(d\bar\cor)
\end{align*}
Let $\csignals_i=\Delta(\Alloci)^{\Typesi}$ and define $\excdf^*(\state,\cor_0)=(\csignal_1,\dots,\csignal_N)\in\times_{i\in[N]}\csignals_i$, where
\begin{align*}
    \csignal_i(\cdot|\hat\typei)=\mathbb{E}_{\typemi\sim \typedmi(\cdot|\state)}\left[\sum_{a_{-i}}\mechanism_N^*\left(\hat\typei,\typemi,\state,\cor_0\right)(\cdot,a_{-i})\right].
\end{align*}
By definition, the information structure is calibrated relative to full participation and truthful reporting.

We now show that full participation and truthful reporting is a best response to others participating and truthfully reporting into the mechanism. To this end, consider agent $i$'s payoff from submitting report $\typeb_i$ when observing $\csignal_i$. Denoting by $\Sigma(\state,\csignal_i)$ the set of $\cor_0$ such that $\excdf_i^*=\csignal_i$, this payoff is given by:\footnote{In the expressions that follow, recall the full participation mechanism $\mechanism_N^*$ already averages over the agents' own randomization devices.}
{\small{
\begin{align*}
    &\sum_{\state\in\States}\frac{\prior(\state)f_i(\typei|\state)}{Pr(\csignal_i|\typei)}\sum_{\typemi}f_{-i}(\typemi|\state)\int_{\Sigma(\state,\csignal_i)}\sum_{\alloc}\mechanism_N^*(\typeb_i,\typemi,\state,\cor_0)(\alloci,\allocmi)\lebesgue(d\cor_0) \payi(\alloci,\typei,\state)=\\
    &\sum_{\alloci\in\Alloci}\sum_{\state\in\States}\frac{\prior(\state)f_i(\typei|\state)}{Pr(\csignal_i|\typei)}\payi(\alloci,\typei,\state)\sum_{\typemi}f_{-i}(\typemi|\state)\int_{\Sigma(\state,\csignal_i)}\sum_{\allocmi}\mechanism_N^*(\typeb_i,\typemi,\state,\cor_0)(\alloci,\allocmi)\lebesgue(d\cor_0)\\
    &=\int_{0}^1\left[\sum_{\alloci}\sum_{\state}\frac{\prior(\state)f_i(\typei|\state)}{Pr(\csignal_i|\typei)}\payi(\alloci,\typei,\state)\int_{\Sigma(\state,\csignal_i)}\left(\star\right)\lebesgue(d\cor_0)\right]\lebesgue(d\cor_i),
\end{align*}}}
\normalsize
where
\begin{align*}
    &\star=\mathbb{E}_{\typemi|\state,\bar\cor_{-i}}\left[\sum_{\allocmi}\hat\mechanism_N(\typeb_i,\typemi,\cor_0,\bar\cor_{-i})(a_i,a_{-i})\right]\\
    &=\mathbb{E}_{\typemi|\state,\bar\cor_{-i}}\left[\sum_{\allocmi}\bar\mechanism_{J(\typeb_i,\typemi,\asignal(\state,\cor_0),\bar{\cor})}(\strat_{J(\typebi,\typemi,\asignal(\state,\cor_0),\overline{\cor})},\state,\cor_0)(a_i,a_{-i})\right]\\
    &=\mathbbm{1}[p_i(\typeb_i,\asignal_i(\state,\cor_0),\cor_i)=1]\asignal_{i}(\alloci|\strat_i(\typeb_i,\asignal_i,\cor_i))+(1-\mathbbm{1}[p_i(\typeb_i,\asignal_i(\state,\cor_0),\cor_i)=1])\delta_{a_{i,\emptyset}}(\alloci).
\end{align*}
Because agent $i$ of type $\typei$ could have imitated type $\typeb_i$, reporting $\typei$ dominates. By the same logic, when the agent reports $\typei$, she obtains at least the payoff from participating in the mechanism.
\subsection{Optimal Calibrated Auction}\label{appendix:app}
\begin{proof}[Proof of \autoref{prop:no-gap-auction}]
The pointwise solution to the Myersonian problem allocates the good to agents in $N^*(\type,\state)=\arg \max_{i\in \setplayers\cup \{0\}} [\ddp_i(\type,\state)+\virtual_i(\typei)\state_i+\state_{0i}$] where $i=0$ corresponds to an outside option with $\ddp_0\equiv\virtual_0\equiv\state_{00}\equiv0$. The conditions of the proposition ensure that for all $i\in N$ and $j\in \setplayers\cup \{0\}$,
\begin{align*}
   \frac{d}{d \typei}\left( \ddp_{i}(\type,\state)+\virtual_i(\typei,\typecdfi)\state_i+\state_{0i}\right)\geq  \frac{d}{d \typei} \left(\ddp_{j}(\type,\state)+\virtual_j(\type_j,\typecdf_j)\state_j+\state_{0j}\right).
\end{align*}
Thus,  an optimal selection $q^*(\type,\state)$ exists such that for each $i$, $\typemi$, and $\state$, $q^*(\typei,\typemi,\state)$ is non-decreasing in $\typei$ (e.g., one that uniformly randomizes over $N^*(\type,\state)$).

Denote  by $Q_{\text{\texttt{full}}}$ the set of allocation rules implementable under full state disclosure. These are the rules such that for all $i$ and $\state$, $\expect_{\typecdfmi} [q_i(\typei,\typemi,\state)]$ is non-decreasing in $\typei$. It follows that $q^*\in Q_{\text{\texttt{full}}}$, and hence $q^*\in\Qmy$. Thus, $q^*$ solves the Myersonian problem and can also be implemented by fully disclosing the state to the agents and conducting an optimal mechanism state-by-state. By revenue equivalence, the expected revenue of such implementation is the same as under no disclosure, and thus the designer obtains payoff $\DP_\myerson$.
\end{proof}

\subsection{Technical results from \autoref{appendix:microf}}\label{appendix:microf-aux}
\begin{proof}[Proof of \autoref{lemma:beliefs-0}]%\annotation{I need to straighten out notation in this proof.}
The set of continuous bounded functions on \estates\ is separable and hence it has a countable dense subset $\{g_k\}_{k\in\naturals}\subset C_b(\estates)$. It is immediate to see that $\belief_n\weakc\belief$ if and only if for all $k\in\naturals$ $\int g_kd\belief_n\to\int g_kd\belief$.

For each $k\in\naturals$ define a real-valued, bounded, martingale on $(\terminalsh, \borel_{\terminalsh},\mathbb{P}_\strat)$ as follows:
\begin{align*}
M_t^k(\estate,h^\infty)=\int_{\estates}g_k(\stateb,\cor)d\belief_t(\estate,h^\infty)(\stateb,\cor).
\end{align*}
Doob's martingale convergence theorem implies that $M_t^k(\estate,h^\infty)=\mathbb{E}[g_k|h^t]\to M_\infty^k(\estate,h^\infty)=\mathbb{E}[g_k|h^\infty]$ $\mathbb{P}_\strat$-a.s. Let $E_k$ denote the subset of \terminalsh\ where convergence happens, and note that $\mathbb{P}_\strat(E_k)=1$.

Let $E=\cap_k E_k$ and note that $\mathbb{P}_\strat(E)=1$. Then, on $E$, we have that for all $k\in\naturals$,
\begin{align}\label{eq:wp1}
\int_{\estates}g_k(\stateb,\cor)d\belief_t(\estate,h^\infty)(\stateb,\cor)\to M_\infty^k(\estate,h^\infty).
\end{align}

Fix now a terminal history $(\estate,h^\infty)$. Because \ebeliefs\ is compact \citep[Theorem 15.11]{aliprantis2006infinite}, the sequence $(\belief_t(\estate,h^\infty))_{t\in\naturals}$ has a convergent subsequence $\belief_{t_j}(\estate,h^\infty)\weakc \tilde{\belief}$. Passing the limit along $t_j$ in \autoref{eq:wp1} we have that for all $k\in\naturals$
\begin{align*}
\int_{\estates}g_k(\stateb,\cor)d\belief_t(\estate,h^\infty)(\stateb,\cor)\to\int_{\estates}g_k(\stateb,\cor)d\tilde{\belief}.
\end{align*}
Because the set $\{g_k:k\in\naturals\}$ determines the convergent subsequences, any subsequential limits must be equal. Hence $\belief_t(\cdot)$ converges on $E$ and call this limit $\tilde\belief_\infty$. Hence on $E$ we have that $\belief_t(h^\infty)\weakc\tilde\belief_\infty(h^\infty)$.

Now, for each $k$, $\int g_kd\belief_\infty=\mathbb{E}[g_k|h^\infty]$ almost surely. Hence, $\tilde\belief_\infty$ is a version of the law of \estates\ conditional on $h^\infty$. This is $\mathbb{P}_\strat(\cdot|h^\infty)$, completing the proof.
\end{proof}%
\begin{proof}[Proof of \autoref{prop:extended-occup}]
Fix a continuous and bounded function $\testf\in C_b(\Argb)$ and define for each terminal history $(\state,h^\infty)$
\begin{align*}
\Delta_t(\state,h^\infty)=\testf(\alloc_t(h^\infty),\type_t(h^\infty),\state,\belief_t(h^\infty))-\testf(\alloc_t(h^\infty),\type_t(h^\infty),\state,\belief_{t+1}(h^\infty)).
\end{align*}
% \begin{align*}
% \Delta_t(\state,h^\infty)(\alloc,\type,\state,\belief_t,\belief_{t+1})=\testf(\alloc,\type,\state,\belief_t(h^\infty))-\testf(\alloc,\type,\state,\belief_{t+1}(h^\infty)).
% \end{align*}
%
%
As in \autoref{lemma:beliefs-0}, let $E$ denote the probability-1 subset of \terminalsh\ on which  $\belief_t\weakc\belief_\infty$.\footnote{To be sure, the proof of \autoref{lemma:beliefs-0} is written in the context of repeated mechanisms but it extends verbatim to dynamic mechanisms with simple notational adjustments.} Then, on $E$, 
$
\Delta_t(\state,h^\infty)\to0
$, 
and hence,
\begin{align*}
\mathbb{E}_\strat\left[\Delta_t(\state,h^\infty)\right]\to0,
\end{align*}
as $t\to\infty$. 
Now, for every $T$,
\begin{align*}
D_T(\testf)\equiv\mathbb{E}_{\eoccup_\strat^{T,1}}\left[\testf\right]-\mathbb{E}_{\eoccup_\strat^{T,2}}\left[\testf\right]=\frac{1}{T}\sum_{t=1}^T\mathbb{E}_{\strat}\left[\Delta_t\right],
\end{align*}
and hence the left-hand side goes to $0$ as $T\to\infty$. 

Now, let $T_n$ be such that $\eoccup_{\strat}^{T_n,2}\weakc\eoccup$ for some $\eoccup\in\Delta(\Argb)$. Note that
\begin{align*}
\mathbb{E}_{\eoccup_{\strat}^{T_n,1}}\left[\testf\right]=\mathbb{E}_{\eoccup_{\strat}^{T_n,2}}\left[\testf\right]+D_{T_n}\left[\testf\right]\to\mathbb{E}_{\eoccup}\left[\testf\right]+0,
\end{align*}
so a subsequential limit of $\eoccup_{\strat}^{T_n,2}$ is a subsequential limit of $\eoccup_{\strat}^{T_n,1}$. Switching the role of 1 and 2, we obtain the opposite set inclusion and the result follows.
\end{proof}
\begin{proof}[Proof of \autoref{prop:belief-occup}]
Fix a continuous function $g$ on \ebeliefs. Then, we want to show that
\begin{align*}
\mathbb{E}_{\bsplit_T}\left[g\right]\to\mathbb{E}_{\mathbb{P}_\strat\circ\belief_\infty^{-1}}\left[g\right].
\end{align*}
By \autoref{lemma:beliefs-0},  $\belief_t\weakc\belief_\infty$ $\mathbb{P}_\strat$-almost surely and $g$ is continuous, we have that $\mathbb{E}_\strat[g(\belief_t)]\to \mathbb{E}_\strat\left[g(\belief_\infty)\right]$ by dominated convergence theorem.\footnote{To be sure,
\begin{align*}
\mathbb{E}_\strat\left[g(\belief_t)\right]=\int_{H^\infty}g(\belief_t(h^\infty))\pr_\strat(d h^\infty).
    \end{align*}}
Because eventually constant sequences have Ces\`aro limits, we have that
\[\frac{1}{T}\sum_{t=1}^T\mathbb{E}_{\strat}\left[g(\belief_t)\right]\to \mathbb{E}_{\strat}\left[g(\belief_\infty)\right].\]
And now we are basically done, because
\begin{align*}
    \mathbb{E}_{\bsplit_T}[g]=\int_{H^\infty}\frac{1}{T}\sum_{t=1}^Tg(\belief_t(h^\infty))\pr_\strat(d h^\infty)=\frac{1}{T}\sum_{t=1}^T\mathbb{E}_{\strat}\left[g(\belief_t)\right]\to \mathbb{E}_\strat[g(\belief_\infty)]=\int g d(\mathbb{P}_\strat\circ\belief_\infty^{-1}).
\end{align*}
In other words, the occupation measure on beliefs induced by the strategy (and the prior, the type distribution, and the mechanism) is the push-forward measure $\left(\mathbb{P}_\strat\circ\belief_\infty^{-1}\right)$. In particular, that $\pr_\strat$ is a measure implies that $\left(\mathbb{P}_\strat\circ\belief_\infty^{-1}\right)$ is a measure itself \citep[Chapter 3.6]{bogachev2007measure}.
\end{proof}

\begin{proof}[Proof of \autoref{lemma:adequate-learning}] We now show the agent can ensure the payoff $U(\bsplit_\rmech)$ in \autoref{eq:target}, which corresponds to the agent's maximum payoff under the calibrated information structure \cexp. Recall that this information structure is the one that corresponds to the 
partition of \estates, $\mathcal{P}$, defined as follows: $\estate,\estateb$ in the same cell $P$ of $\mathcal{P}$ if for all $(\alloc,\mssg)$, $\rmech(\alloc|\mssg,\estate)=\rmech(\alloc|\mssg,\estateb)$. Conditional on cell $P$, the associated posterior is $\belief(|P)\in\ebeliefs$. Let $\bsplit_\rmech\in\Delta(\ebeliefs)$ denote the induced belief distribution (with mean $\prior\otimes\cord$).

To prove the result, we consider the strategy $\stratb_N$ parameterized by a number $N$  and defined as follows. In the exploration phase of $\stratb_N$,  the agent plays each message \mssg\ for $N$ rounds. Let $\belief_{N|\mssgs|}(h^{N|\mssgs|})$ denote the agent's beliefs as a function of the realized sequence of allocations implied by $h^{N|\mssgs|}$. For each history that succeeds $h^{N|\mssgs|}$, the agent of type \type\ plays the message \mssg\ that solves
\[\max_{\mssg\in\mssgs}\sum_{\estate}\belief_{N|\mssgs|}(h^{N|\mssgs|})(\estate)\sum_{\alloc\in\Alloc}\rmech(\alloc|\mssg,\estate)u\aarg.\]
It is immediate to verify that the agent's (limit) average payoff under $\stratb_N$ is given by:
\begin{align}\label{eq:adeq-strat}
U(\stratb_N)&=\mathbb{E}_{\stratb_N}\left[\sum_{\type\in\Types}\typed(\type)\max_{\mssg\in\mssgs}\sum_{\estate}\belief_{N|\mssgs|}(h^{N|\mssgs|})(\estate)\sum_{\alloc\in\Alloc}\rmech(\alloc|\mssg,\estate)u\aarg\right]=
\nonumber\\
&
=
\sum_{h^{N|\mssgs|}\in H^{N|\mssgs|}}\mathbb{P}_{\stratb_N}^{N|\mssgs|}(h^{N|\mssgs|})\pay^*(\belief_{N|\mssgs|}(h^{N|\mssgs|})),
\end{align} 
where $\pay^*$ is as in \autoref{eq:rmech-max}.

Below, we show that the distribution of beliefs under $\stratb_N$, $\mathbb{P}_{\stratb_N}\circ\belief_{N|\mssgs|}^{-1}\weakc\bsplit_\rmech$. Consequently, as $\pay^*$ is continuous and bounded on \ebeliefs, for any $\delta>0$, we can choose $N_\delta$ so that for all $N\geq N_\delta$,
\begin{align*}
|U(\stratb_N)-U(\bsplit_\rmech)|\leq\delta.
\end{align*}
Consequently, the agent's payoff under \strat\ must be $U(\bsplit_\rmech)$ because by definition for all $\delta>0$\footnote{The argument shows that the agent's equilibrium payoff is at least $U(\bsplit_\rmech)$. However, it is immediate that $U(\bsplit_\rmech)$ is the most the agent can make in the game as $\bsplit_\rmech$ extracts all information from the mechanism.}
\begin{align*}
\liminf_{T\to\infty}\mathbb{E}_{\strat}\left[U_T\right]&\geq\limsup_{T\to\infty} \mathbb{E}_{\stratb_{N_\delta}}\left[U_T\right]=U(\stratb_{N_\delta})\geq U(\bsplit_\rmech)-\delta,
\intertext{ and hence, }
\liminf_{T\to\infty}\mathbb{E}_{\strat}\left[U_T\right]&\geq\lim_{\delta\to 0} U(\bsplit_\rmech)-\delta=U(\bsplit_\rmech).
\end{align*}
which completes the proof.

We now complete the missing step:

\paragraph{The law of $\belief_{N|\mssgs|}$ converges to $\bsplit_\rmech$} It is useful to write the bottom line of \autoref{eq:adeq-strat} as follows:
\begin{align*}
\sum_{\estate\in\estates}(\prior\otimes\cord)(\estate)\mathbb{E}_{\mathbb{P}_{\stratb_N}^{N|\mssgs|}(\cdot|\estate)}\left[\pay^*(\belief_{N|\mssgs|})\right].
\end{align*}
We show that the conditional law of $\belief_{N|\mssgs|}$, $\mathbb{P}_{\stratb_N}(\cdot|\estate)$, converges to the Dirac measure on $\belief(\cdot|P(\estate))$. Noting that $\bsplit_\rmech=\sum_{\estate\in\estates}(\prior\otimes\cord)(\estate)\delta_{\belief(\cdot|P(\estate))}$ completes the proof.

For the exploration block, define for each $\mssg\in\mssgs$, the empirical frequency $\hat{\rmech}_{N,\mssg}:(\Types\times\mssgs\times\Alloc)^{N|\mssgs|}\to\Delta(\Alloc)$, as follows
\begin{align*}
\hat{\rmech}_{N,\mssg}(h^{N|\mssgs|})(\alloc)=\frac{1}{N}\sum_{n=1}^N\mathbbm{1}[\Alloc_{\mssg,n}=\alloc],
\end{align*}
where $\Alloc_{\mssg,n}$ is the $n^{th}$ draw from \Alloc\ when the message is $\mssg$. Let $\hat\rmech_N:H^{N|\mssgs|}\to\Delta(\Alloc)^{\mssgs}$ denote the vector of empirical frequencies. The agent's belief at history $h^{N|\mssgs|}$ is given by:
\begin{align}\label{eq:explore-post}
\belief_{N|\mssgs|}(h^{N|\mssgs|})(\estate)=\frac{(\prior\otimes\cord)(\estate)\prod_{\mssg\in\mssgs} \prod_{\alloc\in \Alloc}\rmech(\alloc|\mssg,\estate)^{N\hat{\rmech}_{N,\mssg}(h^{N|\mssgs|})(\alloc)}}{\sum_{\estateb}(\prior\otimes\cord)(\estateb)\prod_{\mssg\in\mssgs} \prod_{\alloc\in\Alloc}\rmech(\alloc|\mssg,\estateb)^{N\hat{\rmech}_{N,\mssg}(h^{N|\mssgs|})(\alloc)}}.
\end{align}
Denote by $\bsplit_{N|\mssgs|,\estate}\in\Delta(\ebeliefs)$ the law of $\belief_{N|\mssgs|}$ conditional on \estate, i.e., $\bsplit_{N|\mssgs|,\estate}=\mathbb{P}_{\stratb_N}(\cdot|\estate)\circ\belief_{N|\mssgs|}^{-1}$.  Below, we show that $\bsplit_{N|\mssgs|,\estate}$ converges weakly to $\delta_{\belief(\cdot|P(\estate))}$.

Suppose the true state is \testate. Then, $(\Alloc_{m,1},\dots,\Alloc_{\mssg,N})$ are drawn i.i.d. from distribution $\rmech(\cdot|\mssg,\testate)$. Fix a continuous and bounded function $\testf$ on \Alloc. Then, almost surely,\footnote{This almost surely is under the law of \Alloc\ under $\rmech(\cdot|\mssg,\testate)$.}
\begin{align*}
\int_\Alloc \testf d\hat{\rmech}_N(\cdot|m)=\frac{1}{N}\sum_{n=1}^N\testf(\Alloc_{\mssg,n})\to\mathbb{E}_{\rmech}[\testf(\Alloc_{m,1})]=\sum_{\alloc\in\Alloc}\testf(\alloc)\rmech(\alloc|\mssg,\testate),
\end{align*}
by the strong law of large numbers applied to the i.i.d random variables $(\testf(\Alloc_{m,1}),\dots,\testf(\Alloc_{m,N}))$. Because this holds for all \testf, then $\hat{\rmech}_N(\cdot|\mssg)\weakc\rmech(\cdot|\mssg,\testate)$ almost surely when the true state is \testate. 

Fix an arbitrary state \estate\ and consider the ratio of the right-hand side of \autoref{eq:explore-post} at \estate\ and \testate:
\begin{align}\label{eq:log-odds}
\frac{\belief_{N|M|}(h^{N|M|})(\estate)}{\belief_{N|M|}(h^{N|M|})(\testate)}=\frac{(\prior\otimes\cord)(\estate)}{(\prior\otimes\cord)(\testate)}\prod_{\alloc\in\Alloc}\prod_{\mssg\in\mssgs}\left(\frac{\rmech(\alloc|\mssg,\estate)}{\rmech(\alloc|\mssg,\testate)}\right)^{N\hat{\rmech}_{N,\mssg}(h^{N|\mssgs|})(\alloc)}.
\end{align}
Suppose $\estate\notin P(\testate)$. By definition of the partition $\mathcal{P}$, a message $\mssg\in\mssgs$ and allocation $\alloc\in\Alloc$ exist such that $\rmech(\alloc|\mssg,\estate)\neq\rmech(\alloc|\mssg,\testate)$. Taking logarithm on both sides of \autoref{eq:log-odds} and dividing by $N$, 
\begin{align}
\frac{1}{N}\log\left(\frac{\belief_{N|M|}(h^{N|M|})(\estate)}{\belief_{N|M|}(h^{N|M|})(\testate)}\right)=\frac{1}{N}\log\left(\frac{(\prior\otimes\cord)(\estate)}{(\prior\otimes\cord)(\testate)}\right)+\sum_{\alloc^\prime\in\Alloc}\sum_{\mssg^\prime\in\mssgs}\hat{\rmech}_{N,\mssg^\prime}(h^{N|\mssgs|})(\alloc^\prime)\log\left(\frac{\rmech(\alloc^\prime|\mssg^\prime,\estate)}{\rmech(\alloc^\prime|\mssg^\prime,\testate)}\right).
\end{align}
Because $\hat{\rmech}_N(\cdot|\mssg)\weakc\rmech(\cdot|\mssg,\testate)$ almost surely when the true state is \testate, 
\begin{align}\label{eq:limit}
\lim_{N\to\infty}\frac{1}{N}\log\left(\frac{\belief_{N|M|}(h^{N|M|})(\estate)}{\belief_{N|M|}(h^{N|M|})(\testate)}\right)=-\sum_{\mssg\in\mssgs}\mathrm{D_{KL}}\left(\rmech(\cdot|\mssg,\testate)|\rmech(\cdot|\mssg,\estate)\right),
\end{align}
where $\mathrm{D_{KL}}$ is the Kullback-Leibler divergence. Note that at least one of the terms in the KL-divergence is positive as $\rmech(\cdot|\mssg,\estate)\neq\rmech(\cdot|\mssg,\testate)$. Hence, 
\begin{align*}
\lim_{N\to\infty}\log\left(\frac{\belief_{N|M|}(h^{N|M|})(\estate)}{\belief_{N|M|}(h^{N|M|})(\testate)}\right)=-\infty,
\end{align*}
meaning that $\belief_{N|M|}(h^{N|M|})(\estate)/\belief_{N|M|}(h^{N|M|})(\testate)\to0$. 

Suppose now that $\estate\in P(\testate)$. Then, \autoref{eq:log-odds} reduces to 
\begin{align}\label{eq:fixed-ratio}
\frac{\belief_{N|M|}(h^{N|M|})(\estate)}{\belief_{N|M|}(h^{N|M|})(\testate)}=\frac{(\prior\otimes\cord)(\estate)}{(\prior\otimes\cord)(\testate)},
\end{align}
for all $N$. 

Collecting both cases, we conclude that conditional on the true state being \testate,
\begin{align*}
\sum_{\estate\in P(\testate)}\belief_{N|M|}(\estate)\rightarrow_{N\to\infty}1,
\end{align*}
and moreover, within the cell, the fixed-ratio property implies the law $\bsplit_{N|\mssgs|,\testate}\weakc\delta_{\belief(\cdot|P(\testate))}$. We conclude that the unconditional belief distribution, $\sum_{\estate\in\estates}(\prior\otimes\cord)(\estate)\bsplit_{N|\mssgs|,\estate}\weakc\sum_{\estate\in\estates}(\prior\otimes\cord)(\estate)\delta_{\belief(P(\estate))}=\bsplit_\rmech$. In particular, 
\begin{align*}
U(\stratb_N)=\mathbb{E}_{\mathbb{P}_{\stratb_N}\circ\belief_{N|\mssgs|}^{-1}}\left[\pay^*(\belief)\right]\to\mathbb{E}_{\bsplit_\rmech}\left[\pay^*(\belief)\right],
\end{align*}
completing the proof.
\end{proof}
\begin{lemma}\label{prop:ball} 
Suppose $\pay:\Beliefs\to\reals$ satisfies that
\begin{align*}
\int_{\Delta(\States)} u(\belief)\,\bsplit(d\belief)<\int_{\Delta(\States)} \max\{u(\belief),0\}\,\bsplit(d\belief),
\end{align*}
then a set $B\subseteq\Delta(\States)$ open relative to \Beliefs\ exists such that
\begin{align*}
\int_B u(\belief)\,\bsplit(d\belief)<0.
\end{align*}
\end{lemma}
\begin{proof}
Define the positive and negative parts of \pay:
\begin{align*}
u_+(\belief):=\max\{u(\belief),0\},\qquad u_-(\belief):=\max\{-u(\belief),0\}.
\end{align*}
Then $u=u_+-u_-$ pointwise. Integrating and using the assumed strict inequality,
\begin{align*}
\int u\,d\bsplit=\int u_+d\bsplit-\int u_-d\bsplit
<\int u_+d\bsplit
\quad\Longrightarrow\quad
\int u_-d\bsplit>0.
\end{align*}
Hence the set $N:=\{\belief\in\Delta(\States):u(\belief)<0\}$ has strictly positive mass under \bsplit.

Embed $\Delta(\States)\subset\mathbb R^{|\States|}$. Extend $\bsplit$ to a finite Borel measure $\tilde\bsplit$
on $\mathbb R^d$ by
\begin{align*}
\tilde\bsplit(A):=\bsplit(A\cap\Delta(\States))\qquad(A\subseteq\mathbb R^d\ \text{Borel}),
\end{align*}
and extend $u$ to $\tilde u:\mathbb R^d\to[-1,1]$ by $\tilde u=u$ on $\Delta(\States)$ and $\tilde u=0$
on $\mathbb R^d\setminus\Delta(\States)$. Then $\tilde u\in L^1(\tilde\bsplit)$ and $\tilde\bsplit(N)=\bsplit(N)>0$.

Let $B_{\reals^{|\States|}}(x,r)$ denote the ball in $\reals^{|\States|}$ with center $x$ and radius $r$. By the Lebesgue differentiation theorem for finite Borel measures on $\mathbb R^d$,
there is a $\tilde\bsplit$-full-measure set $D\subseteq\mathbb R^d$ such that for every $x\in D$,
\begin{align*}
\lim_{r\downarrow 0}\frac{1}{\tilde\bsplit(B_{\mathbb R^d}(x,r))}
\int_{B_{\mathbb R^d}(x,r)} \tilde u\,d\tilde\bsplit
=\tilde u(x),
\end{align*}
whenever $\tilde\bsplit(B_{\mathbb R^d}(x,r))>0$ (and this positivity holds for all sufficiently small $r$
for $\tilde\bsplit$-a.e.\ $x$). Since $\tilde\bsplit(N\cap D)>0$, choose $\belief_0\in N\cap D$. Then
$\tilde u(\belief_0)=u(\belief_0)<0$. Therefore the above limit is strictly negative, so there exists $r_0>0$
such that for all $0<r<r_0$,
\begin{align*}
\int_{B_{\mathbb R^d}(\belief_0,r)} \tilde u\,d\tilde\bsplit<0.
\end{align*}
For such an $r$, let $B:=B_\Delta(\belief_0,r)=\Delta(\States)\cap B_{\mathbb R^d}(\belief_0,r)$, which is an
open ball in $\Delta(\States)$. Using the definitions of $\tilde\bsplit$ and $\tilde u$,
\begin{align*}
\int_B u\,d\bsplit
=
\int_{B_{\mathbb R^d}(\belief_0,r)} \tilde u\,d\tilde\bsplit
<0.
\end{align*}
This proves the claim.
\end{proof}

\subsubsection{Revelation principle for limit of means preferences}\label{appendix:dyn-rp}
We show in this section that when the designer uses dynamic mechanisms, it is without loss of generality for the designer to employ direct dynamic mechanisms that (i) implement the outside option at all histories after the agent first exercises her option not to participate in the mechanism, and (ii) for which the agent's best response is to always participate and truthfully report her type. This justifies the class of mechanisms we employ in the analysis of \autoref{sec:dynamic}.

\paragraph{Histories, mechanisms, and strategies} As in the main text, to simplify notation, we do not include the agent's decision to participate in the mechanism in the histories of the game. Instead, we follow the convention that if the agent does not participate, it is as if she reported $\emptyset$ and the allocation is \oo. Formally, let $\mssgs\Alloc_\emptyset=(\mssgs\times\Alloc)\cup\{(\emptyset,\oo)\}$.  With this notation, a history through period $t$ is an element of $\hat{H}_\mssgs^t\equiv(\mssgs\Alloc_\emptyset)^{t-1}$ and let $\hat{\terminals}_\mssgs^t=\States\times \hat{H}_\mssgs^t$.\footnote{We index histories by the messages to distinguish these histories from those when the designer uses direct mechanisms.}

A mechanism is a collection $\dmech\equiv(\dmecht)_{t=1}^\infty$ such that the mechanism in period $t$ is a mapping $\dmecht:\hat{\terminals}_\mssgs^t\times\mssgs\to\Delta(\Alloc)$. 

Let $H_\mssgs^t=(\Types\times\mssgs\Alloc_\emptyset)^{t-1}=\Types^{t-1}\times\hat{H}_\mssgs^t$. The agent's strategy, $(\pp,\strat)$, is given by her participation strategy $p_t:H_\mssgs^t\times\Types\to[0,1]$, and conditional on participating, her reporting strategy $\strat_t:H_\mssgs^t\times\Types\to\Delta(\mssgs)$.

\paragraph{The distribution over terminal histories} To obtain the complete description of the paths on the tree we need to append $\States$ to  $H_\mssgs^t$; hence the paths through period $t-1$ are $\States\times H_\mssgs^t\equiv\terminals_\mssgs^t$.  
 The distributions over states and agent's types, the agent's strategy, and the mechanism induce a distribution over the terminal histories $\terminalsh_\mssgs\equiv\States\times H_\mssgs^\infty$, which we denote by $\pr_{\dmech,\estrat}\in\Delta(\States\times H_\mssgs^\infty)$, as it is now useful to keep track of the mechanism. We denote by $\mathbb{E}_{\estrat}$ the expectation under this measure. The distribution $\pr_{\dmech,\estrat}\in\Delta(\States\times H_\mssgs^\infty)$ is the unique distribution that satisfies that for all $t\in\naturals$, $\tilde\terminals_\mssgs^t\subset\States\times H_\mssgs^t$, 
\begin{align*}%\label{eq:p-sigma}
\pr_{\dmech,\estrat}(\tilde\terminals_\mssgs^t\times\prod_{s=t+1}^\infty(\Types\times\mssgs\Alloc_\emptyset))=\pr_{\dmech,\estrat}^t(\tilde\terminals_\mssgs^t),
\end{align*}
where the distributions $(\pr_{\dmech,\estrat}^t)_{t\in\naturals}$ satisfy 
\begin{align*}%\label{eq:p-t-t+1-2}
\pr_{\dmech,\estrat}^{t+1}(\state,h_\mssgs^t,\type,\mssg,\alloc)&=\pr_{\dmech,\estrat}^t(\state,h_\mssgs^t)\typed(\type)\pp_t(h_\mssgs^t,\type)\strat_t(h_\mssgs^t,\type)(\mssg)\dmecht(\state,\hat{h}_\mssgs^t,\mssg)(\alloc),\\
\pr_{\dmech,\estrat}^{t+1}(\state,h_\mssgs^t,\type,\emptyset,\alloc)&=\pr_{\dmech,\estrat}^t(\state,h_\mssgs^t)\typed(\type)(1-\pp_t(h_\mssgs^t,\type))\mathbbm{1}[\alloc=\oo].
\end{align*}
\paragraph{Outcome distribution} Our interest is in the distribution over payoff-relevant outcomes, $\States\times (\Types\times\Alloc)^\infty$, and hence on the marginal of $\pr_{\dmech,\estrat}$ on $\States\times (\Types\times\Alloc)^\infty$, which we denote by $\bar{\pr}_{\dmech,\estrat}$.

\paragraph{Best response} We say that strategy $\estrat$ is a best response for the agent if for all alternative strategies $\estratb$, we have that
\begin{align}\label{eq:br-1}
\liminf_{T\to\infty}\mathbb{E}_{\estrat}\left[U_T\right]\geq\limsup_{T\to\infty}\mathbb{E}_{\estratb}\left[U_T\right],
\end{align}
where recall $U_T$ is the agent's average payoff until period $T$.

\paragraph{Direct and full participation mechanisms}  A special case of the above game is that in which $\mssgs=\Types$, and whenever the agent does not participate, the mechanism chooses \oo\ with probability 1 for any message in all continuation histories. We call these mechanisms direct and full participation mechanisms. Below, when $\mssgs=\Types$, we drop the dependence of the set histories on \mssgs.

Formally, let $\hat{\terminals}_\emptyset^t$ denote the subset of $\hat{\terminals}^t$ such that at some point the sequence $(\emptyset,\oo)$ appears. We define mechanisms 
\[\tilde{\dmech}_t:\hat{\terminals}^t\times\Types\to\Delta(\Alloc),\]
such that $\tilde{\dmech}_t(\state,\hat{h}^t,\cdot)=\mathbbm{1}[\alloc=\oo]$ whenever $(\state,\hat{h}^t)\in\hat{\terminals}_\emptyset^t$. 

\begin{theorem}\label{theorem:dynamic-rp}
Suppose that \estrat\ is a best response to mechanism \dmech. Then, a direct and full participation mechanism $\tilde{\dmech}$ exists such that 
\begin{enumerate}
\item Participation with probability 1 and truthtelling are a best response for the agent,
\item The distribution over $\States\times(\Types\times\Alloc)^\infty$ induced by $(\dmech,\estrat)$ is the same as that induced by $\tilde{\dmech}$ under participation and truthtelling.
\end{enumerate}
\end{theorem}
\begin{proof}
Fix a mechanism $\dmech = (\dmech^t)_{t\geq1}$ and a best response $\estrat$ for the agent in the sense of \autoref{eq:br}. Let $\pr_{\dmech,\estrat}$ denote the induced distribution over $\terminals^\infty$. We write $(\state,(\type_t,\mssg_t,\alloc_t)_{t\ge1})$ for a generic realization, where $\mssg_t=\emptyset\Rightarrow\alloc_t=\oo$.

The proof proceeds in three steps. In the first step, we construct a direct (but not full participation) mechanism \dmechc, which under participation and truthtelling after every history on path implements the same outcome distribution as $(\dmech,\estrat)$. In the second step, we verify that participation and truthtelling after every history on path is a best response to \dmechc. In the third step, we construct a direct and full participation mechanism from \dmechc. That the agent can always quit the mechanism at each step and obtain \oo\ and Step 3 implies that participation and truthtelling after every history is also a best response to the full participation mechanism obtained from $\dmechc$.

\textbf{Step 1:} We first construct the direct mechanism $\dmechc_t:\hat\terminals^t\times\Types\to\Delta(\Alloc)$.  Define a collection of transition probabilities $\kappa_t:\terminals_\mssgs^t\times\Types\to\Delta(\mssgs\cup\{\emptyset\})$ as follows:
\begin{align*}
\kappa_t(\mssg_t|h_\mssgs^t,\type_t)=(1-p_t(h_\mssgs^t,\type_t))\mathbbm{1}[\mssg_t=\emptyset]+p_t(h_\mssgs^t,\type_t)\strat_t(h_\mssgs^t,\type_t)(\mssg_t).
\end{align*}
We construct \dmechc\ recursively. In period 1, if the state is \state\ and the report is $\type_1$, the designer draws fictitious $\mssg_1$ from $\kappa_1(\cdot|\type_1)$, and implements $\alloc_1=\oo$ if $\mssg_1=\emptyset$, and otherwise draws $\alloc_1\sim\dmech_1(\state,\mssg_1)$.

Recursively, for $t\geq 2$, if the sequence of reports, fictitious messages, and allocations is $(\type^{\prime^{t-1}},\mssg^{t-1},\alloc^{t-1})=(\typeb_s,\mssg_s,\alloc_s)_{s=1}^{t-1}$ and the agent reports $\type_t$, the designer draws $\mssg_t$ from $\kappa_t(\cdot|\type^{\prime^{t-1}},\mssg^{t-1},\alloc^{t-1},\type_t)$ and implements $\alloc_t=\oo$ if $\mssg_t=\emptyset$, and otherwise draws $\alloc_t\sim\dmech_t(\state,\mssg^{t-1},\alloc^{t-1},\mssg_t)$.\footnote{Recall the fictitious reports encode the agent's participation decisions in the original mechanism.}

It is immediate that under truthtelling and participation the mechanism \dmechc\ implements the same distribution over $\States\times(\Alloc\times\Types)^{\infty}$.\footnote{In fact, if we kept track of the designer's draws of fictitious messages, the new mechanism implements the same distribution over terminal histories $\terminalsh_\mssgs$, and a fortiori, its marginal over $\States\times(\Alloc\times\Types)^{\infty}$ is the same.}

\textbf{Step 2:}  Let \estratopt\ denote the agent's strategy that participates and truthfully reports after every history. We now show that $\estratopt$ is a best response to $\dmechc$ in the sense of \autoref{eq:br-1}. 

To do so, we show that for any strategy $(\tilde\pp,\tilde\strat)$ in the game induced by the direct mechanism $\dmechc$, a strategy $(\pp^\prime,\stratb)$ exists such that 
\begin{align}\label{eq:payoff-match}
\expect_{\dmechc,(\tilde\pp,\tilde\strat)}[U_T] = \expect_{\dmech,\estratb}[U_T]
\quad\text{for all } T,
\end{align}
where $U_T$ is the average payoff through period $T$. Given \autoref{eq:payoff-match} and the best-response property of $\estrat$ to $\dmech$,
\[
\liminf_{T\to\infty} \expect_{\dmech,\estrat}[U_T]
\;\ge\; \limsup_{T\to\infty} \expect_{\dmech,\estratb}[U_T]
\quad\text{for all }\estratb.
\]
Using Step~1, we have $\expect_{\dmech,\estrat}[U_T]=\expect_{\dmechc,\estratopt}[U_T]$ for all $T$, and by \autoref{eq:payoff-match} we have $\expect_{\dmech,\estratb}[U_T]=\expect_{\dmechc,(\tilde\pp,\tilde\strat)}[U_T]$ for all $T$.
Hence
\[
\liminf_{T\to\infty} \expect_{\dmechc,\estratopt}[U_T]
\;\ge\; \limsup_{T\to\infty} \expect_{\dmechc,\estratc}[U_T]
\quad\text{for all }\estratc,
\]
which is exactly the definition of $\estratopt$ being a best response to $\dmechc$. It remains to construct $\estratb$ and verify~\eqref{eq:payoff-match}.

We show how the agent can emulate the strategy \estratc\ in the game induced by the indirect mechanism via strategy \estratb. Define $\tilde\kappa_t: H^t\times\Types\to\Delta(\Types\cup\{\emptyset\})$ as follows
\begin{align*}
\tilde\kappa_t(\typeb|h^t,\type_t)=(1-\tilde\pp(h^t,\type_t))\mathbbm{1}[\typeb=\emptyset]+\tilde\pp(h^t,\type_t)\tilde\strat(h^t,\type_t)(\typeb).
\end{align*}
The strategy \estratb\ privately simulates the report process induced by \estratc\ in the direct mechanism, and conditional on the fictitious reports, generates the actual message in the indirect mechanism \dmech\ using the kernel $\kappa_t$ in Step 1, evaluated at the fictitious type history.

Formally, for $t\geq 1$ given the history of types, messages, and allocations through period $t$, $(\type^{t-1},\mssg^{t-1},\alloc^{t-1})$, and the privately tracked fictitious type reports $\type^{\prime^{t-1}}$,\footnote{Recall that to minimize notation and make history lengths symmetric across participation and nonparticipation, we record the agent's rejection of the mechanism as the empty message $\emptyset$.} the agent of type $\type_t$ draws a fictitious report $\typeb_t\sim\tilde\kappa_t(\cdot|\type^{t-1},\type^{\prime^{t-1}},\alloc^{t-1},\type_t)$. If $\typeb_t=\emptyset$, then $\mssg_t=\emptyset$ (the agent does not participate in period $t$). Otherwise, $\typeb_t\in\Types$ and $\mssg_t$ is drawn from $\kappa_t(\cdot|\type^{\prime^{t-1}},\mssg^{t-1},\alloc^{t-1},\typeb_t)$. The allocation is \oo\ upon rejection, and $\alloc_t\sim\dmech_t(\state,\mssg^{t-1},\alloc^{t-1},\mssg_t)$, otherwise.

It is immediate that \autoref{eq:payoff-match} holds and hence \estratopt\ is a best response to \dmechc:\footnote{In fact, the construction ensures the stronger property that $\pr_{\dmechc,\estratc}^T=\pr_{\dmech,\estratb}^T$, when in a slight abuse of notation we keep track of the fictitious messages in $\pr_{\dmechc,\estratc}^T$.} In the extensive form game induced by \dmechc, the designer simulates the agent's participation and reporting strategies using \estrat\ based on the agent's type reports and determines allocations in the mechanism. When the agent's strategy is given by \estratc, the process described in the above paragraph correspond to the designer's simulated participation and reporting strategies, and allocations continued to be determined by \dmech. Hence, in the extensive form game induced by \dmechc, when the agent plays \estratc, it is as if she faces mechanism \dmech\ and plays strategy \estratb. The best response property of \estrat\ implies that playing \estratb\ yields a weakly worse payoff, and hence \estratc\ is not a profitable deviation from \estratopt\ in the direct mechanism \dmechc.

\paragraph{Step 3:} Modify $\dmechc$ at all histories that include at least one non-participation decision, so that the mechanism implements the outside option \oo. With this modification, the mechanism satisfies the full participation property. It is immediate that participation and truthtelling after every history remains a best response.
\end{proof}
\end{document}